\documentclass[twoside,12pt]{article} 
\usepackage[margin=1.0in]{geometry}

\usepackage{amssymb,amsmath,amsthm}
\usepackage{mathtools}
\usepackage{bm,lscape}
\usepackage{bbm}
\usepackage{algpseudocode}
\usepackage[ruled,vlined]{algorithm2e}
\usepackage{mathrsfs,dsfont}
\usepackage{wasysym}
\RequirePackage[
    colorlinks=true,
    linkcolor=blue,
    urlcolor=blue,
    citecolor=blue]{hyperref}
\RequirePackage{natbib}
\usepackage[utf8]{inputenc}
\usepackage[english]{babel}

\usepackage{graphicx}
\usepackage{color}
\usepackage{rotating}
\usepackage{authblk}
\usepackage{appendix}

\allowdisplaybreaks

\allowdisplaybreaks

\newcommand{\R}{\mathbb{R}}

\newcommand{\ddr}{\mathrm{d}}

\newtheorem{thm}{Theorem}[section]
\newtheorem{lem}[thm]{Lemma}
\newtheorem{defi}[thm]{Definition}

\newtheorem{prp}[thm]{Proposition}

\newtheorem{cor}[thm]{Corollary}
\newtheorem{remark}[thm]{Remark}

\providecommand{\keywords}[1]
{
  \small	
  \textbf{\textit{Keywords:}} #1
}

\begin{document}

\title{Quasi-Bayesian sequential deconvolution}


\author[1]{Stefano Favaro\thanks{stefano.favaro@unito.it}}
\author[2]{Sandra Fortini\thanks{sandra.fortini@unibocconi.it}}
\affil[1]{\small{Department of Economics and Statistics, University of Torino and Collegio Carlo Alberto, Italy}}
\affil[2]{\small{Department of Decision Sciences, Bocconi University, Italy}}

\maketitle

\begin{abstract}
Density deconvolution is the inverse problem of estimating a probability density from observations contaminated by additive noise. Traditionally studied in static or batch settings, it increasingly arises with streaming data, where existing frequentist and Bayesian  procedures face substantial computational bottlenecks. We develop a quasi-Bayesian nonparametric method for sequential density deconvolution based on Newton's recursive algorithm. The resulting estimate is straightforward to evaluate and scalable to massive datasets, as its per-observation computational cost remains constant as new data arrive. The quasi-Bayesian interpretation enables uncertainty quantification: local and uniform central limit theorems yield asymptotic credible intervals and bands, respectively. Under a frequentist data-generating model, we establish $L^1$-consistency for the proposed estimate and show that it asymptotically agrees with the estimate that would be obtained if the uncontaminated variables were directly observed. Further, under additional regularity conditions, we derive an $L^1$-Wasserstein convergence rate and establish merging, at an explicit rate, with the Bayesian nonparametric posterior mean estimate under a Dirichlet process mixture model. Synthetic-data experiments, together with an acquisition-ordered flow-cytometry application, demonstrate accuracy comparable to Bayesian nonparametric and Fourier deconvolution methods, while offering a substantial computational advantage.
\end{abstract}



\section{Introduction}
\subsection{Background and motivations}
Density deconvolution, or simply deconvolution, is the inverse problem of estimating a probability density from observations contaminated by additive noise. More precisely, one observes $n\geq1$ real-valued random variables $Y_1,\ldots,Y_n$ generated according to the convolution model
\begin{displaymath}
Y_i=X_i+Z_i,\qquad i=1,\ldots,n,
\end{displaymath}
where the $X_i$'s are i.i.d. unobserved signal variables with unknown density $f_X$, and the $Z_i$'s are i.i.d. noise variables with known density $f_Z$, independent of the $X_i$'s. The objective is to estimate $f_X$ from $Y_1,\ldots,Y_n$, whose density satisfies the convolution identity $f_Y=f_X\ast f_Z$. The frequentist literature is extensive, encompassing parametric and nonparametric methods with strong theoretical guarantees \citep{Car(88),Ste(90),Dig(93),Fan(91),Fan(93),Pen(99),Del(02),Hes(04),But(08a),But(08b)}; see \citet[Chapters~10--14]{Yi(21)} for an overview. By comparison, the Bayesian literature is more recent and limited, focusing primarily on theoretical guarantees under Dirichlet process mixture models for the signal, with methodological development remaining less mature \citep{Gao(16),Scr(18),Su(20),Rou(23)}.

Density deconvolution arises when observations are blurred or distorted, often by measurement error in the underlying signal \citep{Buo(10),Yi(21)}. Traditionally studied in static or batch settings, it has applications across the natural and social sciences, including medical imaging, bioinformatics, genomics, chemistry, spectroscopy, astronomy, and econometrics. The increasing availability of streaming data in these areas, together with emerging applications in social media analytics, remote sensing, e-commerce, and privacy-preserving data analysis, creates a growing need for online or sequential methods that enable timely inference and rapid adaptation as new observations arrive. Although existing frequentist and Bayesian procedures can, in principle, be adapted to streaming data, they typically incur substantial computational costs due to repeated optimization or posterior simulation. Moreover, no principled sequential method for density deconvolution currently combines theoretical guarantees with demonstrated empirical effectiveness. We fill this gap with a scalable and theoretically grounded framework specifically designed for streaming data.

\subsection{Preview of our contributions}
We propose a nonparametric method for sequential density deconvolution, assuming that the unknown density $f_X$ of the unobserved $X_i$'s admits a mixture representation of the form
\begin{equation}\label{eq:mixture}
f_G^{(X)}(x)=\int_{\Theta}k(x\mid\theta)G(\ddr\theta),
\qquad x\in\mathbb{R},
\end{equation}
where $k(\cdot\mid\theta)$ is a known positive kernel, $\theta\in\Theta\subset\mathbb{R}$, and $G$ is an unknown mixing distribution on $\Theta$. Since $f_Y=f_X\ast f_Z$, the density of the observed $Y_i$'s has a mixture representation over $\Theta$ with kernel $k\ast f_Z$ and the same mixing distribution $G$. Starting from an initial guess $\widetilde G_0$, our method updates $\widetilde G_{n-1}$ recursively upon receiving each new observation $Y_n$, according to the procedure of \citet{Smi(78)}, known as Newton's algorithm \citep{New(98),Mar(08)}. A sequential estimate of $f_G^{(X)}$ is then obtained from \eqref{eq:mixture} by replacing $G$ with  $\widetilde G_{n}$. This recursive procedure admits a quasi-Bayesian interpretation, providing a predictive construction of a Bayesian model that is asymptotically equivalent \citep{For(20)}. That is, for large $n$, it implicitly induces a prior on $G$ under which $f_{\widetilde G_n}^{(X)}$ is the posterior mean of $f_G^{(X)}$. Moreover, $f_{\widetilde G_n}^{(X)}$ is straightforward to evaluate and scalable to massive datasets, as its per-observation computational cost remains constant as new data arrive.

Interpreting Newton's algorithm as a quasi-Bayesian learning model makes it possible to quantify uncertainty in the estimate $f_{\widetilde{G}_n}^{(X)}$ of $f_G^{(X)}$. For density deconvolution at a fixed point $x\in\mathbb{R}$, we establish a local central limit theorem that enables the construction of a sequential asymptotic credible interval for $f_G^{(X)}(x)$. We then extend this result to density deconvolution over a bounded interval $I\subset\mathbb{R}$ by proving a uniform central limit theorem, which yields a sequential asymptotic credible band for $f_G^{(X)}$ on $I$. These new central limit theorems also apply to the direct density estimation problem, in which the $X_i$'s are observed, thus extending the local uncertainty quantification results of \citet{For(20)} for the mixing distribution $G$ to the induced density, both locally and uniformly.

We establish asymptotic guarantees for the estimate $f_{\widetilde G_n}^{(X)}$ under the frequentist modeling assumption that the $X_i$'s are i.i.d. from a mixture model with kernel $k$ and a ``true'' mixing distribution $G^{\ast}$. This frequentist framework was first considered by \citet{Tok(09)} for the direct density estimation problem; see also \citet{MarTok(09),Mar(12),Dix(23)}. First, under $G^{\ast}$, we prove that $f_{\widetilde G_n}^{(X)}$ is a consistent estimate of $f_{G^{\ast}}^{(X)}$ in $L^1$ and asymptotically agrees with the estimate that would be obtained if the $X_i$'s were directly observed. Second, still under $G^{\ast}$ and under additional regularity conditions, we derive a rate of convergence, in the $L^1$-Wasserstein distance, of the probability measure induced by $f_{\widetilde G_n}^{(X)}$ to that induced by $f_{G^{\ast}}^{(X)}$. We also establish that the associated distribution function $F_{\widetilde G_n}^{(X)}$ asymptotically merges, at an explicit rate, with the Bayesian nonparametric posterior mean estimate under a Dirichlet process mixture model \citep{Fer(73),Lo(84)}.

We illustrate the proposed methodology through synthetic- and real-data analyses. In the synthetic-data experiments, we consider Gaussian location-scale mixtures under Laplace (ordinary-smooth) and Gaussian (super-smooth) noise distributions. We compare the resulting estimates with those obtained by a Bayesian nonparametric procedure based on the Dirichlet process mixture model, implemented using sequential Monte Carlo and batch algorithms, finding comparable empirical performance but a substantial computational advantage over both. We then analyze flow-cytometry data on reporter fluorescence in differentiating mouse embryonic stem cells. By retaining the acquisition order and considering nested prefixes of the data stream, this application illustrates sequential learning as observations become available. The resulting estimates are compared with the aforementioned sequential Bayesian nonparametric estimates and batch ridge-regularized Fourier deconvolution estimates.

\subsection{Related work}

Our work is connected to a growing literature on predictive or recursive inference, referred to as the ``post-Bayes'' framework. This literature includes recent developments on predictive distributions and martingale posteriors \citep{Cap(18),Han(18),Fon(23),Fon(24),Hol(23),Agn(25),Bat(25),For(25),Yun(25)}, as well as on generalized Bayesian updates \citep{Bis(16),Kno(22)}. Although these contributions provide important insights and bring a broad range of procedures under a common quasi-Bayesian perspective, they have so far focused on direct problems, in which the $X_i$'s are observed. By contrast, our work addresses the inverse setting, in which only the contaminated $Y_i$'s are available, thereby extending the scope of predictive inference to a new class of problems.

\subsection{Organization of the paper}

Section~\ref{sec2} introduces the quasi-Bayesian learning model for density deconvolution, Section~\ref{sec3} develops sequential estimation and uncertainty quantification, and Section~\ref{sec5} establishes asymptotic frequentist guarantees for the resulting estimates. Sections~\ref{sec6} and~\ref{sec61} present, respectively, the synthetic- and real-data analyses. Section~\ref{sec7} concludes with a discussion of directions for future research. The Supplementary Material contains proofs and additional numerical illustrations.


\section{Quasi-Bayesian density deconvolution}\label{sec2}

\subsection{A learning process for density deconvolution}

In the Bayesian framework, a likelihood and a prior specify a model, with the resulting predictive distributions determining the learning process, that is, how inference is updated as data become available. The predictive framework reverses this perspective by specifying the learning process directly through a sequence of predictive distributions; see \citet{For(25)} and the references therein. Under suitable conditions, this sequence implicitly determines both a prior and the corresponding (random) parameter of interest. The predictive distributions describe a learning mechanism rather than a data-generating process, so the observations need not be independent. We adopt this framework for density deconvolution.

Consider a stream of real-valued random variables $(Y_{n})_{n\geq 1}$ that are generated according to
\begin{equation}\label{eq:Y}
Y_{n}=X_{n}+Z_{n}\qquad n\geq1,
\end{equation}
where the unobserved variables $(X_{n})_{n\geq1}$ are independent of the noise variables $(Z_{n})_{n\geq1}$, with the $Z_{n}$'s being i.i.d. with known density $f_{Z}$ with respect to a $\sigma$-finite measure $\lambda$. Denoting by $(\mathcal G_n)_{n\geq 0}$ the natural filtration generated by $(Y_n)_{n\geq1}$, for each $n\geq0$: i) $X_{n+1}$ is conditionally independent of $(X_{1:n},Z_{1:n})$, given $\mathcal G_n$; ii) the conditional distribution of $X_{n+1}$ given $\mathcal G_n$ is a mixture of a known positive kernel $k(\cdot\mid\theta)$, with $\theta\in\Theta\subset\mathbb{R}$, whose density with respect to $\lambda$ is
\begin{equation}\label{eq:predX}
f^{(X)}_{\widetilde{G}_{n}}(x)=\int_{\Theta}k(x\mid\theta)\widetilde{G}_{n}(\ddr\theta)\qquad x\in\mathbb{R},
\end{equation}
where the mixing distribution $\widetilde{G}_{n}$ on $\Theta$ is defined recursively by means of Newton's algorithm \citep{New(98)}, i.e.,
\begin{equation}\label{eq:newton}
\widetilde{G}_{n+1}(\ddr\theta)=(1-\widetilde{\alpha}_{n+1})\widetilde{G}_{n}(\ddr\theta)+\widetilde{\alpha}_{n+1}\frac{(f_{Z}\ast k(\cdot\mid\theta))(Y_{n+1})\widetilde{G}_{n}(\ddr\theta)}{\int_{\Theta}(f_{Z}\ast k(\cdot\mid\theta^{'}))(Y_{n+1})\widetilde{G}_{n}(\ddr\theta^{'})},
\end{equation}
given an initial guess $\widetilde{G}_{0}$ and a sequence $(\widetilde{\alpha}_{n})_{n\geq1}$ of real numbers in $(0,1)$ with $\sum_{n\geq1}\widetilde{\alpha}_{n}=+\infty$ and $\sum_{n\geq1}\widetilde{\alpha}_{n}^{2}<+\infty$.  The sequence $(\widetilde{\alpha}_{n})_{n\geq1}$ is referred to as the learning rate, for which a standard choice is $\widetilde{\alpha}_{n}=(\alpha+n)^{-\gamma}$ for $\alpha>0$ and $\gamma\in(1/2,1]$. See \citet{For(20)} for details.

From \eqref{eq:Y}--\eqref{eq:newton}, for  each $n\geq0$ the density of the conditional distribution of $Y_{n+1}$ given $\mathcal G_n$, with respect to $\lambda$, is
\begin{equation}\label{eq:predY}
f^{(Y)}_{\widetilde{G}_{n}}(y)=\int_{\Theta}\widetilde{k}(y\mid\theta)\widetilde{G}_{n}(\ddr\theta)\qquad y\in\mathbb{R},
\end{equation}
where
\begin{displaymath}
\widetilde{k}(y\mid\theta)=(f_{Z}\ast k(\cdot\mid\theta))(y).
\end{displaymath}
Since $f_{Z}$ is known, $f^{(Y)}_{\widetilde{G}_{n}}$ is itself a mixture over $\Theta$, with known convolution kernel $\widetilde{k}(y\mid\theta)$ and the same mixing distribution $\widetilde{G}_{n}$ defined in \eqref{eq:newton}. In particular, $\widetilde{G}_{n}$ is updated as the new observation $Y_{n+1}$ becomes available, by taking a weighted average, with respect to $\widetilde{\alpha}_{n+1}$, of  $\widetilde{G}_n$ itself and of its posterior distribution, based on $Y_{n+1}$. The sequence of predictive densities $p_n(y):=f_{\tilde G_n}(y)=\int_\Theta\tilde k(y\mid\theta)\tilde G_n(\ddr\theta)$ characterizes the predictive learning mechanism for the observable $Y_{n}$'s. The next proposition shows that this predictive specification induces a well-defined learning process for the $(X_n,Z_n,Y_n)$'s, and hence for the $(X_n)$'s; see \ref{sec:appprop1} for the proof.

\begin{prp}\label{prop:lawX}
There exists a probability space $(\Omega,\mathcal F,P)$ and three sequences of random variables $(X_n)_{n\geq1}$, $(Y_n)_{n\geq1}$ and $(Z_n)_{n\geq1}$, which are defined on $(\Omega,\mathcal F,P)$, such that the following statements hold: i) \eqref{eq:Y}, \eqref{eq:newton} and \eqref{eq:predY} hold true; ii) the $Z_n$'s are independent of the $X_n$'s and i.i.d. with known density $f_Z$ with respect to $\lambda$; iii) for every $n\geq1$, $X_{n+1}$ is conditionally independent of $(X_{1:n},Z_{1:n})$, given $\mathcal G_n$, and the conditional distribution of $X_{n+1}$, given $\mathcal G_n$ has density \eqref{eq:predX} with respect to $\lambda$; iv) the conditional distribution of $X_{n+1}$ given  $\{X_i,Y_i,Z_i:i\leq n\}$ is uniquely defined by \eqref{eq:Y}, \eqref{eq:newton} and \eqref{eq:predY}, as well as by the assumptions on $(Z_n)_{n\geq1}$. 
\end{prp}

\subsection{Quasi-Bayesian properties of the learning process \eqref{eq:Y}--\eqref{eq:newton}}

The next theorem, building on \citet{For(20)}, shows that the learning process \eqref{eq:Y}--\eqref{eq:newton} is quasi-Bayesian, providing a predictive construction of a Bayesian model that is asymptotically equivalent as $n\to+\infty$. Whereas \citet{For(20)} considered direct density estimation with the $X_n$'s observed, deconvolution requires an additional argument because only the contaminated $Y_n$'s are observed. The predictive structure of $(Y_n)_{n\geq1}$ fits the framework of \citet{For(20)}, since the conditional distribution of $Y_{n+1}$ given $\mathcal G_n$ is a mixture with mixing distribution $\widetilde G_n$. It remains to characterize the induced probabilistic structure of $(X_n)_{n\geq1}$ and establish the required properties relative to the filtration $(\mathcal G_n)_{n\geq0}$.

\begin{thm}\label{teo_cid}
Consider the learning process  \eqref{eq:Y}--\eqref{eq:newton}, and assume that $k(x\mid\theta)$ is bounded and continuous in $\theta\in\Theta$ for every $x\in\mathbb R$ and that $\int\sup_{\theta\in\Theta}k(x\mid\theta)\lambda(dx)<+\infty$. As $n\rightarrow+\infty$, $\widetilde G_n$ converges weakly $P$-a.s. to a random probability distribution $\widetilde G$ on $\Theta$,  such that, by defining
\begin{equation}\label{parx}
f^{(X)}_{\widetilde{G}}(x)=\int_{\Theta}{k}(x\mid\theta)\widetilde{G}(\ddr\theta)\qquad x\in\mathbb{R},
\end{equation}
the following statements hold: i) $P$-a.s. $f_{\widetilde G}^{(X)}$ is the limit pointwise and in $L^1(\mathbb R,\mathcal B(\mathbb R),\lambda)$ of $f_{\widetilde G_n}^{(X)}$; ii) $\int_A f^{(X)}_{\widetilde G}(x)\lambda(dx)=\lim_{n\rightarrow+\infty}\frac 1 n  \sum_{k=1}^{n}I(X_k\in A)$; iii) for every $n\geq 0$,  $k\geq 1$ and $A\in\mathcal B(\mathbb R)$, $P(X_{n+k}\in A\mid \mathcal G_n)=\int_A f^{(X)}_{\widetilde G_n}(x)\lambda(dx)=E\left( \int_A f^{(X)}_{\widetilde G}(x)\lambda(dx)\middle|\, \mathcal G_n\right)$ 
with $f_{\widetilde G_n}^{(X)}$ as in \eqref{eq:predX}; iv) for any $n\geq 1$, $P(X_n\in A)=\iint_A f^{(X)}_{G}(x)\lambda(dx)\pi(dG)$, where $\pi$ is the distribution of $\widetilde G$.
\end{thm}

See \ref{app4_0} for the proof of Theorem~\ref{teo_cid}. It follows from \eqref{eq:Y}--\eqref{eq:newton} that, for every $n\geq0$, the random variables $(X_{n+k})_{k\geq1}$ are conditionally identically distributed given $\mathcal G_n$, the information available after the $n$-th step. Since $(X_{n+k})_{k\geq1}$ and $(X_i,Z_i)_{i\leq n}$ are conditionally independent given $\mathcal G_n$, the sequence $(X_n)_{n\geq1}$ is conditionally identically distributed in the sense of \citet{Berti(04)}, with directing random probability measure $f_{\widetilde G}^{(X)}(x)\lambda(dx)$. Consequently, by \citet[Theorem~2.2]{Berti(04)}, conditionally on the random limit $\widetilde G$, the sequence $(X_n)_{n\geq1}$ is asymptotically i.i.d. with density $f_{\widetilde G}^{(X)}$ relative to $\lambda$, while its empirical distribution converges weakly, $P$-almost surely, to the random probability measure with density $f_{\widetilde G}^{(X)}(x)=\int_\Theta k(x\mid\theta)\widetilde G(\ddr\theta)$. This is precisely the quasi-Bayesian framework of \citet{For(20)}, in which $f_{\widetilde G}^{(X)}$ is the random mixture density assigned in the infinite-sample limit, making $\widetilde G$ the natural parameter of interest in the density deconvolution problem.


\section{Sequential estimation and uncertainty quantification}\label{sec3}
\subsection{Preliminaries}\label{sec3:preliminaries}

By Theorem~\ref{teo_cid}, the density $f_{\widetilde G}^{(X)}$ in \eqref{parx} is the inferential target of the quasi-Bayesian learning process \eqref{eq:Y}--\eqref{eq:newton}. Theorem~\ref{teo_cid} also shows that $f_{\widetilde G_n}^{(X)}$ in \eqref{eq:predX}, is a natural sequential estimate of $f_{\widetilde G}^{(X)}$. Unlike in a Bayesian model, however, uncertainty about  $f_{\widetilde G}^{(X)}$ cannot be quantified analytically, since $f_{\widetilde G}^{(X)}$ is defined implicitly through the limiting mixing distribution $\widetilde G$, rather than through a prior specification.  Within the quasi-Bayesian learning process \eqref{eq:Y}--\eqref{eq:newton}, the conditional distribution of $f_{\widetilde G}^{(X)}$, given $Y_{1:n}$, plays the role of the posterior distribution, with $f_{\widetilde G_n}^{(X)}$ as its center. We establish local and uniform approximations of this conditional distribution, leading to asymptotic credible intervals and bands for $f_{\widetilde G}^{(X)}$. Approximations arise from central limit theorems under the following assumptions: A1) $\mbox{Supp}(\widetilde G_0)=\Theta$, with $\mbox{Supp}$ denoting closed support; A2) $\lambda$ is the Lebesgue measure on $\mathbb R$; A3) $k(x\mid\cdot )>0$, bounded and continuous for every $x\in\mathbb R$, and  $\int\sup_{\theta\in\Theta}k(x\mid\theta)\lambda(dx)<+\infty$; A4) $\int_\Theta k(\cdot\mid\theta)G(\ddr\theta)$ and $\int_{\Theta} \widetilde k(y\mid\theta)G(\ddr\theta)$ identify $G$; A5) $f_Z$ is strictly positive on $\mathbb R$ and bounded.

The assumptions A1)--A5) imply that $\widetilde k(y\mid\theta)>0$ for every $y\in\mathbb R$ and $\theta\in\Theta$. For the kernel $\widetilde{k}(y\mid\theta)=(f_{Z}\ast k(\cdot\mid\theta))(y)$,  A3)--A5) hold, for instance, for the Gaussian kernel $k$ with bounded expectation and variance bounded away from zero, and with $f_Z$ the density of either a Laplace (ordinary-smooth) or Gaussian (super-smooth) distribution. A closed form for $\widetilde k$ is not required but improves computational efficiency; otherwise, it can be evaluated numerically once, before the sequential updates, so that no additional integration is required during the recursion.

\begin{remark}
The assumption $k(x\mid\theta)>0$ for all $x$ and $\theta$ can be relaxed to allow $\theta$-dependent supports. In this case, \eqref{eq:predX} and its local and uniform asymptotic properties remain valid on regions where $f_{\widetilde G}^{(X)}$ is strictly positive. We retain the positivity assumption to simplify the statements.
\end{remark}

The central limit theorems below exploit the martingale structure associated with the $X_{n}$'s being conditionally identically distributed, as established in Theorem~\ref{teo_cid}, together with a suitable almost-sure conditional convergence result for martingales; we refer to \ref{appB} for details.

 \subsection{Asymptotic credible intervals}
We establish an asymptotic Gaussian approximation of the conditional distribution of $f^{(X)}_{\widetilde{G}}$,  given $Y_{1:n}$, at a fixed point $x\in\mathbb{R}$, i.e. local  approximation. According to  Theorem~\ref{teo_cid}, under A1)--A5), $f_{\widetilde G_n}^{(X)}(x)$ converges $P$-a.s. to  $f_{\widetilde G}^{(X)}(x)$, and, for every $y\in\mathbb R$, $f_{\widetilde G_n}^{(Y)}(y)$ converges $P$-a.s. to 
\begin{equation}\label{eq:y}
f_{\widetilde{G}}^{(Y)}(y)=\int_\Theta \widetilde k(y\mid\theta)\widetilde G(\ddr\theta).
\end{equation}
Define the limiting conditional density of $(X_n)_{n\geq1}$, given $(Y_n)_{n\geq1}$, as
\begin{equation}\label{eq:xgiveny}
f_{\widetilde{G}}^{(X)}(x\mid y)=\int_\Theta k(x\mid\theta)\widetilde G(\ddr\theta\mid y),
\end{equation}
with
\begin{equation}\label{eq:gthetay}
\widetilde G(\ddr\theta\mid y)=\frac{\widetilde k(y\mid\theta)\widetilde G(\ddr \theta)}{f_{\widetilde{G}}^{(Y)}(y)}.
\end{equation}
By A3) and A5), $f_{\widetilde{G}}^{(X)}(x)$, $f_{\widetilde{G}}^{(Y)}(y)$ and $f_{\widetilde{G}}^{(X)}(x\mid y)$ are strictly positive for every $x$ and $y$. The next theorem provides a local central limit theorem, at a fixed point $x\in\mathbb{R}$, for $f^{(X)}_{\widetilde{G}}$ centered at $f_{\widetilde G_n}^{(X)}$.

\begin{thm}\label{th:clt}
For $x\in \mathbb R$, let $f^{(X)}_{{\widetilde{G}_{n}}}(x)$ and $f_{\widetilde{G}}^{(X)}(x)$ be defined as in \eqref{eq:predX} and \eqref{parx}, respectively. Under A1)--A5), if $(b_n)_{n\geq1}$ is a sequence of strictly positive numbers such that $b_n\uparrow+\infty$, $\sum_{n\geq1} \widetilde\alpha_n^4b_n^{2}<+\infty$, $b_n\sup_{k\geq n}\widetilde\alpha_k^2\rightarrow 0$ and $ b_n\sum_{k\geq n}\widetilde \alpha_k^2\rightarrow 1$, then, for $\lambda$-almost every $x\in\mathbb R$,  as $n\rightarrow+\infty$ the conditional distribution of $b_n^{1/2}( f_{\widetilde{G}}^{(X)}(x)- f^{(X)}_{\widetilde{G}_{n}}(x) )$, given $\mathcal G_n$, converges $P$-a.s.  to a centered Gaussian kernel with random variance function $v(x)$. In particular, there holds
\begin{equation}\label{eq:vx}
v(x)=\int \left(f_{\widetilde{G}}^{(X)}(x\mid y)-f_{\widetilde{G}}^{(X)}(x)\right)^2f_{\widetilde{G}}^{(Y)}(y)\lambda(\ddr y),
\end{equation}
with $f_{\widetilde{G}}^{(X)}$, $f_{\widetilde{G}}^{(Y)}$ and $f_{\widetilde{G}}^{(X)}(\cdot\mid y)$ defined in \eqref{parx}, \eqref{eq:y} and \eqref{eq:xgiveny}, respectively. In particular, if the learning rate is $\widetilde{\alpha}_n=(\alpha+n)^{-\gamma}$,  with $\alpha>0$ and $\gamma\in(1/2,1]$, then $b_n=(2\gamma-1)n^{2\gamma-1}$, for $n\geq1$.
\end{thm}

See \ref{app:prrofth2} for the proof of Theorem~\ref{th:clt}. If $(\widetilde\alpha_n)_{n\geq1}$ is non-increasing, then $b_n\sup_{k\geq n}\widetilde\alpha_k^2\rightarrow 0$ follows from $\sum_{n\geq1}\widetilde\alpha_n^4b_n^2<+\infty$. A typical choice for $(\widetilde\alpha_n)_{n\geq1}$ is given by $\widetilde\alpha_n=1/(1+n)$, for which $b_n=n$.

Let $\Phi_{0,\sigma^2}$ denote the cumulative distribution function of a Gaussian random variable with mean $0$ and variance $\sigma^2$, with the proviso that $\Phi_{0,0}$ corresponds to the case of a distribution degenerate at zero. Under the assumptions of Theorem~\ref{th:clt}, for every $z\neq 0$, as $n\rightarrow+\infty$
\begin{equation}\label{clt_unk}
P\left(b_n^{1/2}(f_{\widetilde{G}}^{(X)}(x)-f_{\widetilde{G}_{n}}^{(X)}(x))\leq z\, \middle|\, \mathcal G_n\right)\rightarrow\Phi_{0,v(x)}(z),\qquad P\text{-a.s.}
\end{equation}
Since the variance $v(x)$ depends on the unknown mixing distribution $\widetilde G$, \eqref{clt_unk} cannot be applied to constructing credible intervals for $f_{\widetilde{G}}^{(X)}(x)$. Then, we introduce the empirical variance $v_n(\cdot)$, which is defined as $v(\cdot)$ in \eqref{eq:vx} with $\widetilde{G}_{n}$ in place of $\widetilde{G}$; accordingly we set $f^{(X)}_{\widetilde{G}_{n}}(\cdot\mid y)$ and $\widetilde G_n(\cdot\mid y)$ as $f^{(X)}_{\widetilde{G}}(\cdot\mid y)$ in \eqref{eq:xgiveny} and $\widetilde G(\cdot\mid y)$ in \eqref{eq:gthetay}, respectively, with $\widetilde{G}_{n}$ in place of $\widetilde{G}$. In particular, as $n\rightarrow+\infty$, $v_{n}(x)$ converges to $v(x)$, $P$-a.s.; we refer to Lemma~\ref{lem:vn} in \ref{appE} for details.

Let $(b_n)_{n\geq1}$ be a sequence of strictly positive numbers as in Theorem~\ref{th:clt}. Then, for every $t_1,t_2\in\mathbb R$, as $n\rightarrow+\infty$
\begin{align*}
&E\left(e^{it_1 b_n^{1/2}\left(f_{\widetilde{G}}^{(X)}(x)-f_{\widetilde G_{n}}^{(X)}(x)\right)+it_2v_n(x)}\,\middle|\, \mathcal G_n\right)\rightarrow e^{it_1^2v(x)+it_2v(x)},\qquad P\text{-a.s.}
\end{align*}
Therefore, by an application of the Portmanteau theorem, for every $z>0$ and every $\epsilon>0$, it holds that
\begin{align*}
&\liminf_{n\rightarrow+\infty} P\left(b_n^{1/2}\left|f_{\widetilde{G}}^{(X)}(x)-f_{\widetilde G_{n}}^{(X)}(x)\right|<z (\max(v_n(x),\epsilon)
)^{1/2}\,\middle| \, \mathcal G_n\right)
\geq 2\Phi_{0,1}(z)-1.
\end{align*}
If $z_{1-\beta/2}$ is the $(1-\beta/2)$-quantile of the standard Gaussian distribution, with $\beta\in(0,1)$, then, as $n\rightarrow+\infty$
\begin{displaymath}
\liminf_{n\rightarrow+\infty}P\left(\left|f_{\widetilde G}^{(X)}(x)-f_{\widetilde G_n}^{(X)}(x)\right|\leq b_n^{-1/2}z_{1-\beta/2}\bigl(\max\{v_n(x),\epsilon\}\bigr)^{1/2}\,\middle|\,\mathcal G_n\right)\geq 1-\beta,
\end{displaymath}
i.e., for every $\epsilon>0$, a sequential asymptotic credible interval for $f_{\widetilde{G}}^{(X)}(x)$, with level at least $(1-\beta)$ is
\begin{equation}\label{asym_intervals}
f^{(X)}_{\widetilde G_n}(x)\pm b_n^{-1/2}z_{1-\beta/2}\bigl(\max\{v_n(x),\epsilon\}\bigr)^{1/2}.
\end{equation}

\subsection{Asymptotic credible bands}

We establish a Gaussian approximation of the conditional distribution of $f^{(X)}_{\widetilde{G}}$, given $Y_{1:n}$, on a bounded interval $I\subset\mathbb{R}$, i.e. uniform approximation. In addition to A1)--A5), we assume: A6) $k(x\mid\theta)$ is continuously differentiable with respect to $x$, with $\sup_{\theta\in\Theta,\,x\in \mathbb R}|k(x\mid\theta)|<+\infty$ and $\sup_{\theta\in\Theta,\,x\in \mathbb R}|\frac{\ddr}{\ddr x}k(x\mid\theta)|<+\infty$. Let $C(\mathbb R)$ be the space of continuous functions on $\mathbb{R}$, endowed with the topology of uniform convergence on compact sets. Under A1)--A6), $f_{\widetilde G_{n}}^{(X)}(\omega)$ is continuous for every $\omega\in\Omega$. Then, as $n\rightarrow+\infty$, $f_{\widetilde G_{n}}^{(X)}\rightarrow f_{\widetilde{G}}^{(X)}$ $P$-a.s., i.e. as $n\rightarrow+\infty$
\begin{displaymath}
\sup_{x\in I}\left|f_{\widetilde G_{n}}^{(X)}(x)-f_{\widetilde{G}}^{(X)}(x)\right|\rightarrow 0,\qquad P\text{-a.s.}
\end{displaymath}
The next theorem provides a uniform central limit theorem, on $ I \subset \mathbb{R} $, for $ f_{\widetilde G}^{(X)} $ centered at $ f_{\widetilde G_n}^{(X)} $. The limiting Gaussian process is characterized by a centered Gaussian process kernel.

\begin{thm}\label{th:uniformclt}
Let $f^{(X)}_{\widetilde{G_{n}}}$ and $f_{\widetilde{G}}^{(X)}$ be defined as in \eqref{eq:predX} and \eqref{parx}, respectively. Under A1)--A6), if $(b_n)_{n\geq1}$ is a sequence of strictly positive numbers such that $b_n\uparrow+\infty$, $\sum_{n\geq1} \widetilde\alpha_n^4b_n^{2}<+\infty$, $b_n\sup_{k\geq n}\widetilde\alpha_k^2\rightarrow 0$ and $ b_n\sum_{k\geq n}\widetilde \alpha_k^2\rightarrow 1$, then, as $n\rightarrow+\infty$ the conditional distribution of  $b_n^{1/2}(f_{\widetilde{G}}^{(X)}-f_{\widetilde G_{n}}^{(X)})$, given $\mathcal G_n$, converges $P$-a.s. in $C(\mathbb R)$ to a centered Gaussian process kernel with random covariance function $R(x_1,x_2)$, for every $(x_{1},x_{2})\in\mathbb{R}^{2}$. In particular, there holds
\begin{equation}\label{eq:covariance}
R(x_1,x_2)=\int (f_{\widetilde{G}}^{(X)}(x_1\mid y)-f_{\widetilde{G}}^{(X)}(x_1))(f_{\widetilde{G}}^{(X)}(x_2\mid y)-f_{\widetilde{G}}^{(X)}(x_2))f_{\widetilde{G}}^{(Y)}(y)\lambda(\ddr y),
\end{equation}
with $f_{\widetilde{G}}^{(X)}$, $f_{\widetilde{G}}^{(Y)}$ and $f_{\widetilde{G}}^{(X)}(\cdot\mid y)$ defined in \eqref{parx}, \eqref{eq:y} and \eqref{eq:xgiveny},  respectively. In particular, if the learning rate is $\widetilde\alpha_n=(\alpha+n)^{-\gamma}$, with $\alpha>0$ and $\gamma\in(1/2,1]$, then $b_n=(2\gamma-1)n^{2\gamma-1}$, for $n\geq1$.
\end{thm}

See \ref{app:proofth3} for the proof of Theorem~\ref{th:uniformclt}. The random covariance function $R$ in \eqref{eq:covariance} is measurable with respect to the $\sigma$-algebra generated by $\widetilde G$, and it is positive semi-definite for every $\widetilde G$. Indeed, for every $J\geq 1$ and for every $x_1,\dots,x_J\in\mathbb R$ and $u_1,\dots,u_J\in \mathbb R$, there holds
\begin{align*}
&\sum_{i,j=1}^J u_i u_j \int_{\mathbb{R}} \left(f_{\widetilde{G}}^{(X)}(x_i\mid y)-f_{\widetilde{G}}^{(X)}(x_i)\right)\left(f_{\widetilde{G}}^{(X)}(x_j\mid y)-f_{\widetilde{G}}^{(X)}(x_j)\right)f_{\widetilde{G}}^{(Y)}(y)\lambda(\ddr y)\\
&\quad=\int_{\mathbb{R}} \left(\sum_{i=1}^J u_i(f_{\widetilde{G}}^{(X)}(x_i\mid y)-f_{\widetilde{G}}^{(X)}(x_i))\right)^2f_{\widetilde{G}}^{(Y)}(y)\lambda(\ddr y)\geq 0.
\end{align*}
For each possible realization $G$ of $\widetilde G$, let $\mathbb G_G$ be a  Gaussian process in $C(\mathbb R)$, independent of $\widetilde G$, and with covariance function as in \eqref{eq:covariance} with $G$ in the place of $\widetilde G$. The next theorem provides an upper bound  for the fluctuations of $|\mathbb G_{\widetilde G}|$ on a bounded interval $I$.  Let $\sigma(I)$ and $\psi(z)$ be defined as:
\begin{equation} \label{eq:sigmaI}
\sigma(I)=\left(\sup_{x\in I}\int (f_{\widetilde{G}}^{(X)}(x\mid y)-f_{\widetilde{G}}^{(X)}(x))^2f_{\widetilde{G}}^{(Y)}(y)\lambda(\ddr y)\right)^{1/2}
\end{equation}
and
\begin{equation}\label{eq:psi}
\psi(z)=\!\!\!\!\!\!\!\sup_{\small{\begin{array}{l}x_1,x_2\in I,\\|x_1-x_2|<z\end{array}}}\!\!\!\!\!\!\!\!\left(\int (f_{\widetilde{G}}^{(X)}(x_1\mid y)-f_{\widetilde{G}}^{(X)}(x_2\mid y))^2f_{\widetilde{G}}^{(Y)}(y)\lambda(\ddr y)-(f_{\widetilde{G}}^{(X)}(x_1)-f_{\widetilde{G}}^{(X)}(x_2))^2\right)^{1/2}\!\!\!\!. 
\end{equation}
The upper bound for the fluctuations of $|\mathbb G_{\widetilde G}|$ follows by bounding $E(\sup_{x\in I} \mathbb G_{\widetilde G}(x)\mid \widetilde G)$ and then by applying a concentration inequality for Gaussian processes to bound $\mathbb G_{\widetilde G}(x)$, uniformly in $x\in I$, with respect to its conditional distribution, given $\widetilde G$. We refer to \ref{appC} for details.

\begin{thm}\label{th:sup}
Let $\psi^{-1}(t)=\inf\{z:\psi(z)>t\}$ denote the inverse of the function $\psi$ in \eqref{eq:psi}. For every $\beta\in (0,1)$,
\begin{displaymath}
P\left(\sup_{x\in I}|\mathbb G_{\widetilde G}(x)|\leq s_{\widetilde G}(I,\beta)\,\middle|\, \widetilde G\right)\geq 1-\beta,
\end{displaymath}
with
\begin{equation}\label{eq_v}
s_{\widetilde G}(I,\beta)=12\int_0^{\sigma(I)}\left(\log\left(1+\frac{\lambda(I)}{2\psi^{-1}(z/2)}\right)\right)^{1/2}dz+\sigma(I)\sqrt{2|\log(\beta/2)|}.
\end{equation}
\end{thm}

See \ref{app:proofth4} for the proof of Theorem~\ref{th:sup}. We now  provide an asymptotic credible band for $f_{\widetilde G}^{(X)}$. From Theorem~\ref{th:uniformclt}, 
\begin{align*}
&P\left(\sup_{x\in I}b_n^{1/2}\left|f_{\widetilde{G}}^{(X)}(x)-f_{\widetilde{G}_{n}}^{(X)}(x)\right|\in \cdot\,\,\middle| \, \,\mathcal G_n\right)\stackrel{w}{\longrightarrow} P\left(\sup_{x\in I}|\mathbb G_{\widetilde G}(x)| \in\cdot \,\,\middle|\, \widetilde G\right)\qquad P\text{-a.s.},
\end{align*}
where, from Theorem~\ref{th:sup}, $P\left(\sup_{x\in I}|\mathbb G_{\widetilde G}(x)|\leq s_{\widetilde G}(I,\beta)\, \, \middle| \,\,\widetilde G\right)\geq 1-\beta$ for every $\beta\in (0,1)$, with $s_{\widetilde G}(I,\beta)$ as in \eqref{eq_v}.
Now, consider the empirical version of $s_{\widetilde G}(I,\beta)$, which is defined as 
\begin{equation}\label{eq_vn}
s_n(I,\beta)=12\int_0^{\sigma_n(I)}\left(\log\left(1+\frac{\lambda(I)}{2\psi_n^{-1}(z/2)}\right)\right)^{1/2}dz+\sigma_n(I)\sqrt{2|\log(\beta/2)|},
\end{equation}
where we set $\sigma_n(\cdot)$ and $\psi_{n}(\cdot)$ as $\sigma(\cdot)$ in \eqref{eq:sigmaI} and $\psi(\cdot)$ in \eqref{eq:psi}, respectively, with $\widetilde{G}_{n}$ in place of $\widetilde{G}$.

\begin{thm}\label{th:convv}
If $s_{\widetilde G}(I,\beta)$ and $s_n(I,\beta)$ are defined as in \eqref{eq_v} and \eqref{eq_vn}, respectively, then as $n\rightarrow+\infty$ 
\begin{displaymath}
s_n(I,\beta)\rightarrow s_{\widetilde G}(I,\beta)\qquad P\text{-a.s.}
\end{displaymath}
\end{thm}

See \ref{app:proofth5} for the proof of Theorem~\ref{th:convv}. Let $(b_n)_{n\geq1}$ be a sequence of numbers as defined in Theorem~\ref{th:uniformclt}. Because of the limiting behaviour of $s_n(I,\beta)$ in  Theorem~\ref{th:convv}, for every  $\epsilon >0$ and every $\beta\in (0,1)$,
\begin{equation}\label{th:confband}
\liminf_{n\rightarrow+\infty}P\left(\sup_{x\in I}\left|f_{\widetilde{G}}^{(X)}(x)-f_{\widetilde{G}_{n}}^{(X)}(x)\right|< b_n^{-1/2}\max(s_n(I,\beta),\epsilon)\,\middle|\, \mathcal G_n\right)\geq 1-\beta.
\end{equation}
See \ref{app:proofeq} for details on Equation \eqref{th:confband}. It follows by an application of \eqref{th:confband} that, for every $\epsilon>0$ and every $\beta\in (0,1)$, a sequential asymptotic credible band for $f_{\widetilde{G}}^{(X)}$, on $I$, with level at least $(1-\beta)$ is
\begin{equation}\label{asym_bands}
f_{\widetilde{G}_{n}}^{(X)}\pm b_n^{-1/2}\max(s_n(I,\beta),\epsilon).
\end{equation}


\section{Frequentist  properties}\label{sec5}

\subsection{Preliminaries}

We establish asymptotic guarantees for the estimate $f_{\widetilde G_n}^{(X)}$ under the frequentist modeling assumption  that the $X_{n}$'s are i.i.d. from a mixture model with kernel $k$ and a ``true'' mixing distribution $G^{\ast}$, i.e.,
\begin{equation}\label{eq:fg*}
f_{G^{\ast}}^{(X)}(x)=\int_{\Theta} k(x\mid\theta)G^{\ast}(\ddr\theta)\qquad x\in\mathbb{R}.
\end{equation}
To underline that the $X_n$'s are i.i.d., we denote by $P^*$ the  probability measure on $(\Omega,\mathcal F)$ under which the $X_{n}$'s are defined. To establish asymptotic frequentist guarantees we assume  A1)--A5) from Section~\ref{sec3:preliminaries}, together with these assumptions: F1) $\widetilde G_0$ is absolutely continuous with respect to a $\sigma$-finite measure $\mu$ and $G^{\ast}$ belongs to the weak closure of the class of probability measures on $\Theta$ that are absolutely continuous with respect to $\mu$; F2) $\Theta$ is compact; F3) $\int\sup_{\theta_1,\theta_2 \in \Theta} \left( \frac{\widetilde{k}(y \mid \theta_1)}{\widetilde{k}(y\mid \theta_2)} \right)^2f_{G^{\ast}}^{(Y)}(y)\lambda(\ddr y)<+\infty$, with $f_{G^{\ast}}^{(Y)}=f_Z*f_{G^{\ast}}^{(X)}$ and $f_{G^{\ast}}^{(X)}$ defined as in \eqref{eq:fg*}.

\subsection{Consistency and agreement with the direct density estimation problem}\label{sec:consistency}

The next proposition provides consistency of $f_{\widetilde G_n}^{(X)}$ in $L^1$, under the ``true'' $G^{\ast}$. This result is an immediate consequence of \citet[Corollary~4.7]{MarTok(09)}, where consistency  was first established for the direct density estimation problem, in which the $X_{n}$'s are observed.

\begin{prp}\label{teo_cons}
Let $(X_n)_{n \geq 1}$, $(Y_n)_{n \geq 1}$, and $(Z_n)_{n \geq 1}$ be three sequences of random variables defined on a probability space $(\Omega, \mathcal{F}, P^*)$, such that the $X_n$'s are i.i.d. with density $f_{G^{\ast}}^{(X)}$ in \eqref{eq:fg*}, the $Z_n$'s are i.i.d. with density $f_Z$ and independent of the $X_n$, and $Y_n = X_n + Z_n$ for all $n \geq 1$. Let $f_{\widetilde G_n}^{(X)}$ be defined as in \eqref{eq:predX},  where $\widetilde G_n$ satisfies Newton's recursion \eqref{eq:newton}. If A1)--A5) and F1)--F3) hold true, then $\widetilde G_n$ converges weakly to $G^{\ast}$, and $f_{\widetilde G_n}^{(X)}$ converges to $f_{G^{\ast}}^{(X)}$ in $L^1$, $P^*$-a.s.
\end{prp}

See \ref{app:cons} for the proof of Proposition~\ref{teo_cons}. \citet{MarTok(09)} also consider model misspecification, proving that when the ``true'' density lies outside the mixture model with kernel $k$, the recursive estimate converges almost surely in $L^1(\mathbb R)$ to the mixture minimizing its Kullback--Leibler divergence from that density. Extending this result to deconvolution would require relating the Kullback--Leibler divergences between observed convolution densities and their latent mixture densities. Such a relationship is generally unavailable because convolution is a smoothing operator: substantially different latent densities may induce very similar convolution densities. Accordingly, we restrict our analysis to the correctly specified setting.

Proposition~\ref{teo_cons} makes it possible to establish that the estimates obtained from the observed $Y_{n}$'s agree asymptotically with those that would be obtained from the unobserved $X_{n}$'s. In particular, the quasi-Bayesian learning process on the $X_n$'s is formulated, for $n\geq0$, as
\begin{equation}\label{eq:X}
X_{n+1} \mid \mathcal G_n^a \sim \int_\Theta k(\cdot\mid\theta)\,G_n(\ddr\theta),
\end{equation}
where $\mathcal G_n^a$ is the \lq\lq augmented\rq\rq  $\sigma$-algebra generated by $((X_i,Y_i,Z_i))_{i\leq n}$, and $G_{n}$ is updated recursively by Newton's algorithm, i.e., 
\begin{equation}\label{eq:gn}
G_{n+1}(\ddr\theta) = (1 - \alpha_{n+1}) G_n(\ddr\theta) + \alpha_{n+1} \frac{k(X_{n+1} \mid \theta)\,G_n(\ddr\theta)}{\int_\Theta k(X_{n+1} \mid \theta')\,G_n(\ddr\theta')},
\end{equation}
given an initial guess $G_{0}$ and a learning rate $(\alpha_{n})_{n\geq1}$ with $\sum_{n \geq 1} \alpha_n = +\infty$ and $\sum_{n \geq 1} \alpha_n^2 < +\infty$. Note that $\alpha_{n}$ in \eqref{eq:gn} may differ from $\widetilde{\alpha}_{n}$ in \eqref{eq:newton}, and that $k$ is used in place of the convolution kernel $\widetilde{k}$.

Equations \eqref{eq:X} and \eqref{eq:gn} imply that, for every $n\geq0$, $X_{n+1}$ is conditionally independent of $((Z_i,Y_i))_{i\leq n}$ given $(X_i)_{i\leq n}$. This reflects the natural assumption that, once $X_1,\ldots,X_n$ are observed, the corresponding noise-contaminated observations provide no additional information for learning. Under the ``true'' $G^{\ast}$, the next corollary shows that the estimate $f_{\widetilde G_n}^{(X)}$ asymptotically agrees with $f_{G_n}^{(X)}$, obtained from the quasi-Bayesian learning process \eqref{eq:X}--\eqref{eq:gn}.

\begin{cor}\label{teo_merg}
Consider the quasi-Bayesian learning processes in \eqref{eq:predX}–\eqref{eq:newton} and \eqref{eq:X}–\eqref{eq:gn}. Let the assumptions of Proposition~\ref{teo_cons} hold true, let $\mbox{\rm Supp}(G_0)=\Theta$ and $G_0$ be absolutely continuous with respect to $\mu$, and let $\int \sup_{\theta_1,\theta_2 \in \Theta} \left( \frac{k(x \mid \theta_1)}{k(x \mid \theta_2)} \right)^2 f_{G^{\ast}}^{(X)}(x)\, \lambda(\mathrm{d}x) < +\infty$. Then, as $n\rightarrow+\infty$
\begin{equation*}
\int \left| f_{G_n}^{(X)}(x) - f_{\widetilde G_n}^{(X)}(x) \right| \lambda(\mathrm{d}x) \rightarrow 0 \qquad P^*\text{-a.s.}
\end{equation*}
\end{cor}
See \ref{app:direct} for the proof of Corollary~\ref{teo_merg}. Corollary~\ref{teo_merg} may also be used to motivate an optional Monte
Carlo simulation-based calibration of the learning rate $\widetilde\alpha_n$; we refer to \ref{sec:lrates} for details.

\subsection{Merging with the Dirichlet process mixture model}\label{sec:merg}

Under the ``true'' $G^{\ast}$, and assuming a Laplace noise distribution, the next proposition provides a rate of convergence of the distribution induced by the estimate $f_{\widetilde G_n}^{(X)}$, to the one induced by $f_{G^{\ast}}^{(X)}$, in the $L^1$-Wasserstein metric. This metric is motivated by the inversion inequality of \citet[Theorem~3.1]{Rou(23)}, which relates the $L^1$ error between $f_{\widetilde G_n}^{(Y)}$ and $f_{G^{\ast}}^{(Y)}$ to the $L^1$-Wasserstein distance between the distributions induced by $f_{\widetilde G_n}^{(X)}$ and $f_{G^{\ast}}^{(X)}$. We write $W_1(F_1,F_2)$ for the Wasserstein distance between the probability measures on $\mathcal B(\mathbb R)$ associated with the distributions $F_1$ and $F_2$; it holds  $W_1(F_1,F_2)=\int_{\mathbb R}|F_1(x)-F_2(x)|\,dx$. 

\begin{prp}\label{th:rateNewton}
Let $(X_n)_{n \geq 1}$, $(Y_n)_{n \geq 1}$, and $(Z_n)_{n \geq 1}$ be three sequences of random variables defined on a probability space $(\Omega, \mathcal{F}, P^*)$, such that the $X_n$ are i.i.d. with density $f_{G^{\ast}}^{(X)}$ in \eqref{eq:fg*}, the $Z_n$ are i.i.d. with  Laplace distribution and independent of the $X_n$, and $Y_n = X_n + Z_n$ for all $n \geq 1$. Let $f_{\widetilde G_n}^{(X)}$ be defined as in \eqref{eq:predX},  where $\widetilde G_n$ satisfies Newton's recursion \eqref{eq:newton} with $\widetilde \alpha_n\sim n^{-\gamma}$ such that $\gamma\in (2/3,1)$. Let A1)--A5) and F1)--F3) hold true and, in addition, assume that:
\begin{enumerate}
\item[(i)] there exists $\kappa>0$ and $M_\kappa\in L^1(\mathbb R)$ such that
\begin{equation}
\sup_{\theta\in\Theta}\left|\partial_x^\ell k(x+\delta\mid\theta)-\partial_x^\ell k(x\mid\theta)\right|\leq M_\kappa(x)|\delta|^{\kappa-\ell},\qquad x,\delta\in\mathbb R,\label{eq:M},
\end{equation}
where $\ell:=\max\{j\in\mathbb N_0:j<\kappa\}$;
\item[(ii)] there exists $a\in(0,1)$ such that
\begin{equation}\label{eq:theorem-exp-moment}
\sup_{\theta\in\Theta}\int_{\mathbb R}e^{a|x|}k(x\mid\theta)\,dx<+\infty.
\end{equation}
\end{enumerate}
Then, denoting by $F_{\widetilde G_n}^{(X)}$ and  $F_{G^{\ast}}^{(X)}$ the distribution functions induced by $f_{\widetilde G_n}^{(X)}$ and $f_{G^{\ast}}^{(X)}$ respectively, it holds
\begin{equation}\label{eq:inverse-W1-rate}
W_1\!\left(F_{\widetilde G_n}^{(X)},F_{G^{\ast}}^{(X)}\right)=o_{{P^*}}\!\left(n^{-\frac{(1-\gamma)(\kappa+1)}{2(\kappa+2)}}\right).
\end{equation}
\end{prp}

See \ref{sec:proofrateNewtron} for the proof of Proposition~\ref{th:rateNewton}. Besides being of independent interest, Proposition~\ref{th:rateNewton} provides the key ingredient for establishing that, still under the ``true'' $G^{\ast}$, the quasi-Bayesian distribution function (estimate) $F_{\widetilde G_n}^{(X)}$ asymptotically merges, at an explicit rate in the $L^1$-Wasserstein metric, with the corresponding Bayesian nonparametric posterior mean estimate under a Dirichlet process mixture model \citep{Fer(73),Lo(84)}, as shown next.
To establish this merging result, we consider the Dirichlet process Gaussian-mixture model
\begin{align}
Y_i&=X_i+Z_i,\quad
X_i\mid\mu_i,\sigma^2\stackrel{\mathrm{ind}}{\sim}
\phi(\cdot\mid\mu_i,\sigma^2),\quad
Z_i\stackrel{\mathrm{iid}}{\sim}\mathrm{Laplace}(0,b),
\qquad i=1,\ldots,n,
\nonumber\\
\mu_i\mid G&\stackrel{\mathrm{iid}}{\sim}G,\quad
G\sim\mathrm{DP}(H_0),\quad
\sigma\sim\Pi_\sigma.
\label{eq:Rousseau}
\end{align}
where $\phi(\cdot\mid\mu,\sigma^{2})$ is the Gaussian density with mean $\mu$ and variance $\sigma^2$, $\sigma$ is independent of $(\mu,G)$, $DP(\cdot)$ is the law of a Dirichlet process with non-atomic (diffuse) base measure $H_{0}$, and $\Pi_\sigma$ is the prior distribution for $\sigma$. Under the model \eqref{eq:Rousseau}, we denote by $\widehat F_n^B$ the Bayesian nonparametric posterior mean estimate, i.e. $\widehat F_n^B(x)=E_{\Pi}(F^{(X)}_{G,\sigma}(x)\mid Y_{1:n})$, where $F^{(X)}_{G,\sigma}(x)=P_\Pi(X_i\leq x\mid G,\sigma)$. The posterior contraction rates of \citet[Theorem~4.2 and Theorem~4.4.]{Rou(23)} do not require the ``true'' density to admit a mixture representation. In contrast, the merging result in the next corollary relies on Proposition~\ref{th:rateNewton}, whose proof assumes that the ``true'' density admits a mixture representation, i.e. it is correctly specified. 

\begin{cor}
\label{th:Newton-Bayes-merging}
Suppose that the assumptions of Proposition~\ref{th:rateNewton} hold true, and, in addition, assume that:
\begin{enumerate}
\item[(i)] the density $f_{G^{\ast}}^{(X)}$ satisfies, for some constants $C_0,C_0'>0$ and $\iota>1$,
\begin{equation}\label{eq:true-density-tail-Bayes}
f_{G^{\ast}}^{(X)}(x)\leq C_0'\exp\!\left\{-(1+C_0)|x|^\iota\right\},\qquad x\in\mathbb R;
\end{equation}
\item[(ii)] the prior $\Pi_\sigma$ has a continuous density $\pi_\sigma$ on $(0,+\infty)$ such that, for some constants
$D_1,D_1',D_2,D_2'>0$ and $s_1,s_2,t_1,t_2\geq0$,
\begin{equation}\label{eq:prior-sigma-zero}
D_1'\sigma^{-s_1}\exp\!\left\{-D_1\sigma^{-1}|\log\sigma|^{t_1}\right\}\leq\pi_\sigma(\sigma)\leq D_2'\sigma^{-s_2}\exp\!\left\{-D_2\sigma^{-1}|\log\sigma|^{t_2}\right\}
\end{equation}
for all sufficiently small $\sigma>0$. Moreover, for some constants $D_3,D_3'>0$ and $\varrho>1$,
\begin{equation}\label{eq:prior-sigma-infinity}
\Pi_\sigma([\bar\sigma,+\infty))\leq D_3'\exp\!\left\{-D_3\bar\sigma^\varrho\right\},\qquad \bar\sigma\rightarrow+\infty;
\end{equation}
\item[(iii)] the base measure $H_0$ has a continuous and strictly positive density $h_0$ such that, for some constants
$b_0,b_0',c_0,c_0'>0$,
\begin{equation}\label{eq:DP-base-tail}
c_0\exp\!\left\{-b_0|u|^\iota\right\}\leq h_0(u)\leq c_0'\exp\!\left\{-b_0'|u|^\iota\right\},\qquad u\in\mathbb R,
\end{equation}
where $\iota>1$ is as in \eqref{eq:true-density-tail-Bayes}.
\end{enumerate}
Then,
\begin{equation}
W_1\!\left(
F_{\widetilde G_n}^{(X)},
\widehat{F}_n^{\,B}
\right)
=
o_{P^*}\!\left(
n^{-\frac{(1-\gamma)(\kappa+1)}
{2(\kappa+2)}}
\right).
\label{eq:Newton-Bayes-merging}
\end{equation}
\end{cor}
See \ref{app:Newton-Bayes} for the proof of Corollary~\ref{th:Newton-Bayes-merging}.


\section{Synthetic-data analysis}\label{sec6}

\subsection{Unimodal example}\label{sec:unimodal-main}

For each $n\in\{250;\, 500;\, 750;\,1,000\}$ we generate random variables $X_1,\ldots,X_n$ i.i.d. according to the density $f^{(X)}(\cdot)=\phi(\cdot\mid 2,2)$, where $\phi(\cdot\mid\mu,\sigma^2)$ denotes the Gaussian density with mean $\mu$ and variance $\sigma^2$.  Using the mixture representation \eqref{eq:mixture}, we write $f^{(X)}$ as $f_{G^{\ast}}^{(X)}$, with $G^{\ast}=\delta_{(2,2)}$ the true mixing distribution. We consider the Laplace noise distribution for the random variables $Z_{i}$'s, namely $Z_i\stackrel{\mathrm{iid}}{\sim}\mathrm{Laplace}(0,b_l)$ with $b_l=\frac{\sigma_l}{\sqrt{2}}$ and $\sigma_l\in\{0.25,0.50\}$, where $\sigma_l$ is the standard deviation of the Laplace noise. The  $Z_{i}$'s are independent of the $X_{i}$'s, and the observations are modeled as follows: $Y_i=X_i+Z_i$, for $i=1,\ldots,n$. For each $n$, the same  $Y_{1:n}=(Y_{1},\ldots,Y_{n})$ is used for all the methods under comparison in this section.

We implement Newton's algorithm~\eqref{eq:newton} with $k(\cdot\mid\theta)=\phi(\cdot\mid\mu,\sigma^2)$, where  $\theta=(\mu,\sigma^2)\in\Theta=[-20,20]\times[0.1,5]$.  We use the learning rate $\widetilde\alpha_i=(1+i)^{-1}$, for $i\geq1$, and set $\widetilde G_0$ to the Uniform distribution on $\Theta$. Numerical integrals are evaluated using trapezoidal quadrature, and the mixing distribution is renormalized after each update. Since Newton's recursion depends on the order in which the $Y_{i}$'s are processed, following \cite{Tok(09)}, for each $n$ we repeat it over $R_n=n/10$ random permutations of $Y_{1:n}$. If $\widetilde G_n^{(r)}$ denotes Newton's mixing-distribution estimate obtained from the $r$th permutation, $r=1,\ldots,R_{n}$, then we form the order-averaged estimate $\overline G_n=\frac{1}{R_n}\sum_{r=1}^{R_n}\widetilde G_n^{(r)}$. Based on \eqref{eq:predX}, we define the quasi-Bayes estimate of $f^{(X)}$ as
 \begin{equation} \label{eq:main-fGbar}
 f_{\overline G_n}^{(X)}(x)=\int_\Theta\phi(x\mid\mu,\sigma^2) \overline G_n(\ddr\mu,\ddr\sigma^2).
\end{equation}
Similarly, based on \eqref{asym_intervals} and \eqref{asym_bands}, we define the quasi-Bayes credible intervals and bands for $f_G^{(X)}$, respectively, using $\overline G_n$. For numerical implementation, the parameter space $\Theta$ is discretized  using increments of $0.1$ in both coordinates, giving $401\times50=20,050$ grid points. Such a discretization imposes no modeling restrictions. Further details are given in \ref{sec:supp-unimodal-newton}. In addition, we refer to \ref{sec:supp-unimodal-mc} for a comparison with Monte Carlo credible intervals and bands, showing that the quasi-Bayesian credible intervals closely agree with their Monte Carlo counterparts, whereas the quasi-Bayesian credible bands are wider and thus more conservative, due to the metric-entropy and Gaussian concentration bounds used in their construction.

Figure~\ref{fig:main-unimodal-laplace} reports the quasi-Bayes estimate of $f_{G}^{(X)}$ in \eqref{eq:main-fGbar}, and the asymptotic credible intervals and bands. The figure consists of two columns, which correspond to the Laplace noise standard deviations $\sigma_{l}=0.25$ and $\sigma_{l}=0.50$, respectively.  The first row of the figure compares the true density $f_{G^{\ast}}^{(X)}$ with the quasi-Bayes density estimates, while the second row displays the credible intervals, and the third row displays the credible bands on the interval $I=[-4,6]$.

\begin{figure}[t]
 \centering
 \includegraphics[
   width=\textwidth,
   trim=20 15 20 0,
   clip
 ]{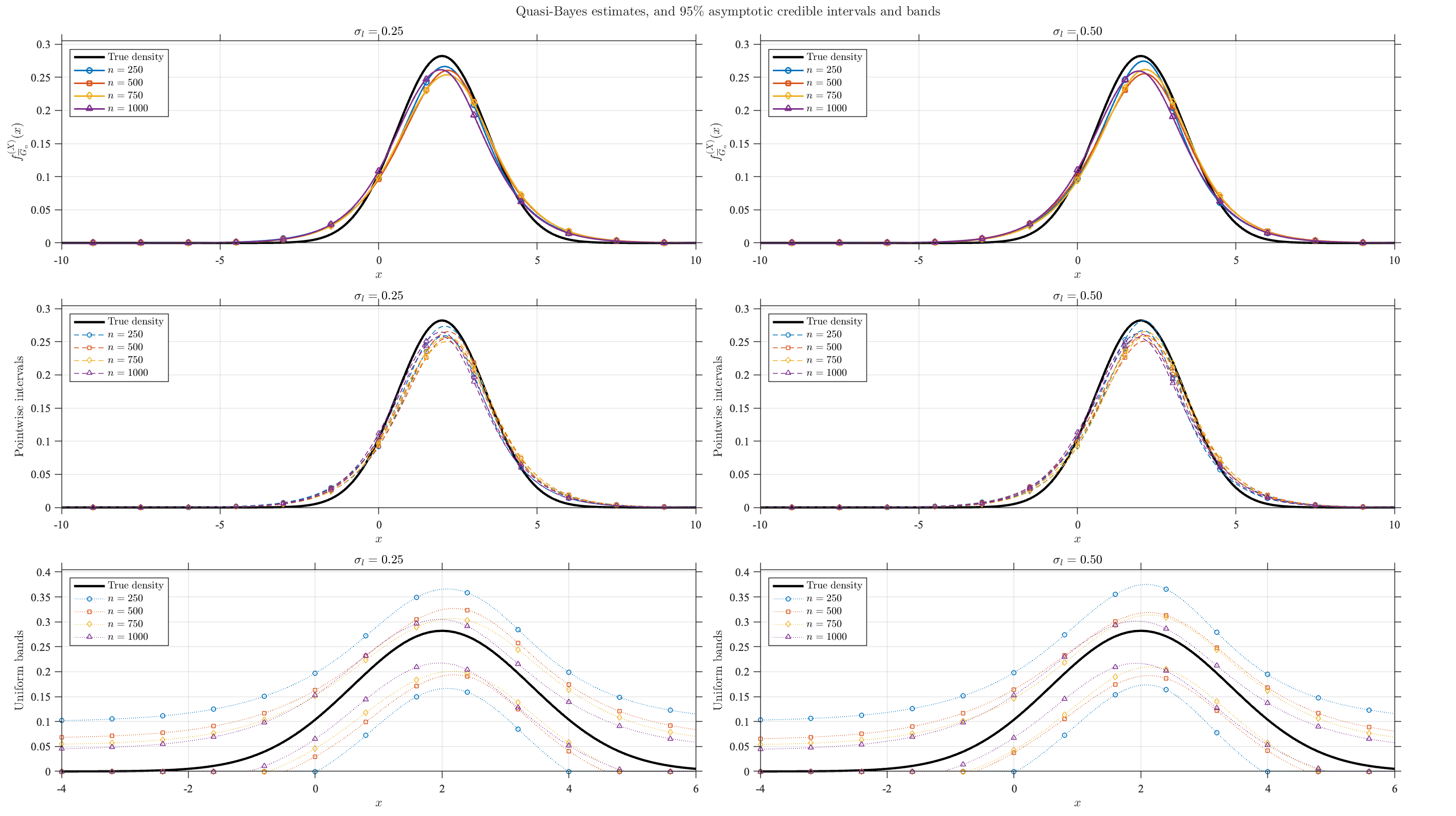}
\caption{\footnotesize{Unimodal example, with Laplace noise: estimates, credible intervals and bands.}}
 \label{fig:main-unimodal-laplace}
\end{figure}

We compare the quasi-Bayesian approach with a Bayesian approach in which the mixing distribution $G$ is endowed with a Dirichlet process prior \citep{Fer(73),Lo(84)}. Consider the Gaussian kernel $k(\cdot\mid\theta)=\phi(\cdot\mid\mu,\sigma^2)$, where $\theta=(\mu,\sigma^2)\in\mathbb{R}\times\mathbb{R}^{+}$, and assume the hierarchical model
\begin{align}
Y_i&=X_i+Z_i,\quad
X_i\mid\theta_i\stackrel{\mathrm{ind}}{\sim}
k(\cdot\mid\theta_i),\quad
Z_i\stackrel{\mathrm{iid}}{\sim}\mathrm{Laplace}(0,b_l),
\qquad i=1,\ldots,n,
\nonumber\\
\theta_i\mid G&\stackrel{\mathrm{iid}}{\sim}G,\quad
G\mid M\sim\mathrm{DP}(M,H),\quad
M\sim\mathrm{Gamma}(1,1).
\label{eq:main-unimodal-dp-model}
\end{align}
where $\text{Gamma}(\cdot,\,\cdot)$ is the Gamma distribution, and $DP(\cdot,\cdot)$ is the law of a Dirichlet process with strength parameter $M>0$ and non-atomic (diffuse) base probability measure $H$ \citep{Fer(73)}. We set $H$ to be the Uniform distribution on $\Theta=[-20,20]\times[0.1,5]$. Posterior inference is performed both
sequentially, using the sequential Monte Carlo (SMC) algorithm of~\citet{Fea(04)} and in batch, using Markov chain Monte Carlo (MCMC) Algorithm~8 of~\citet{Nea(00)}. Implementation details for the SMC algorithm and Algorithm~8  are given in \ref{sec:supp-unimodal-bayes}. Newton's and SMC algorithms are averaged over the same $R_n$ permutations, whereas Algorithm~8 is applied once to each complete dataset. Figure~\ref{fig:main-unimodal-dp} summarizes the comparison between the quasi-Bayesian and Bayesian approaches. The first row of the figure displays the estimates of $f_{G}^{(X)}$ obtained with Newton's algorithm, the SMC algorithm, and Algorithm~8; the second and third rows report total work and total CPU time, respectively.

\begin{figure}[t]
 \centering
 \includegraphics[
   width=\textwidth,
   trim=20 15 20 0,
   clip
 ]{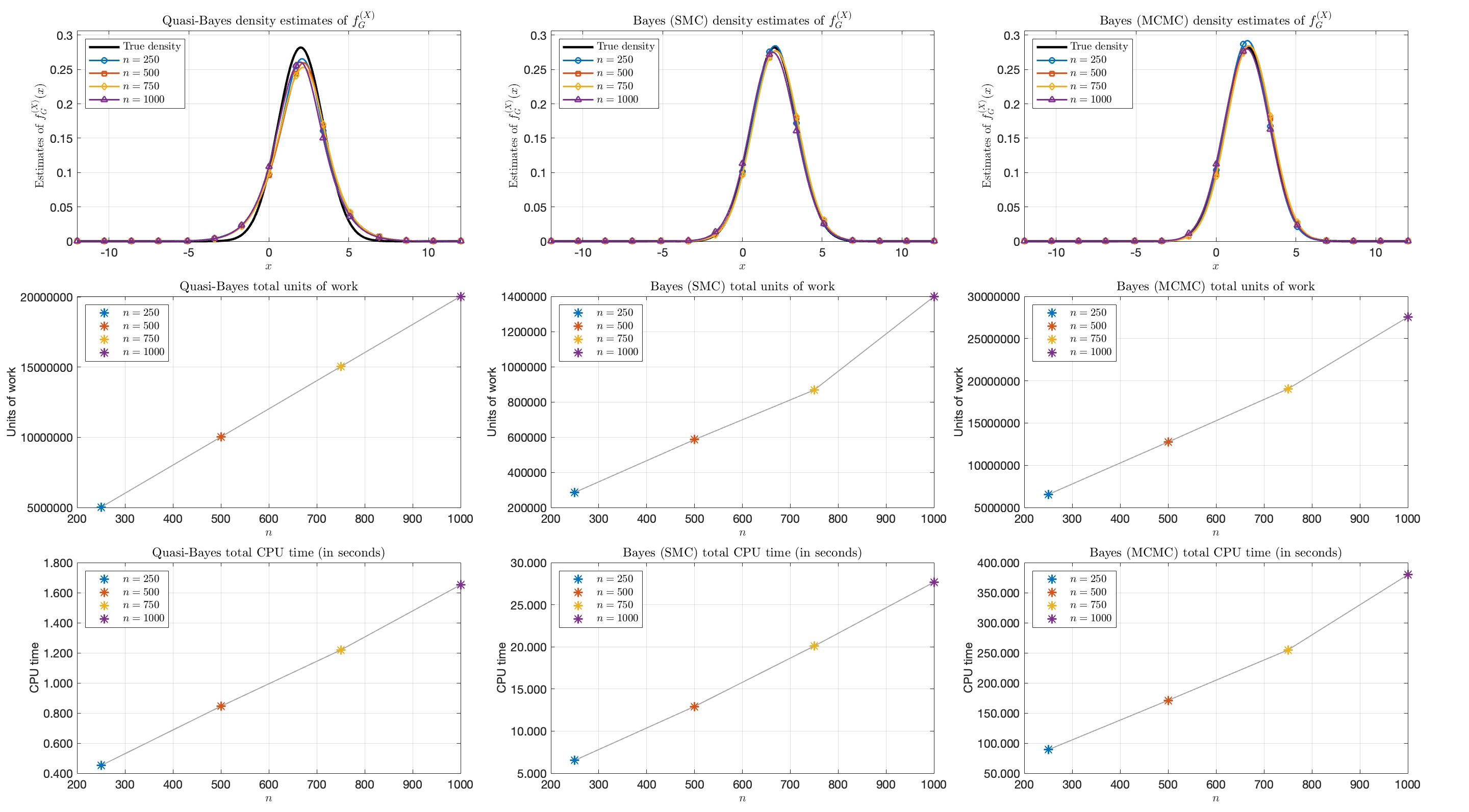}
 \caption{\footnotesize{Unimodal example, with Laplace noise ($\sigma_l=0.25$): comparison between quasi-Bayesian and Bayesian approaches.}} \label{fig:main-unimodal-dp}
\end{figure}

Figure~\ref{fig:main-unimodal-dp} shows that Newton's algorithm provides estimates that are comparable with those obtained from the SMC algorithm and Algorithm~8. The main differences concern computational efficiency, which is expressed in terms of CPU time. Although Newton's and SMC algorithms process the observations sequentially, the SMC algorithm requires the propagation and possible resampling of particles, whereas Newton's algorithm relies on a single deterministic recursive update. Algorithm~8, instead, requires a complete batch analysis for each sample size. Consequently, Newton's algorithm combines competitive estimation accuracy with a substantially smaller CPU time, particularly relative to Algorithm~8.

We refer to \ref{sec:supp-unimodal} for a more detailed analysis of the unimodal example, also considering the Gaussian noise distribution for the $Z_{i}$'s, namely $Z_i\stackrel{\mathrm{iid}}{\sim}N(0,\sigma_g^2)$ with standard deviation $\sigma_g\in\{0.25,0.50\}$.

\subsection{Bimodal example}\label{sec:bimodal1-main}

For each $n\in\{250,500,750,1000\}$, we generate random variables $X_1,\ldots,X_n$ i.i.d. according to the density $f^{(X)}(\cdot)=0.3\,\phi(\cdot\mid-1,2)+0.7\,\phi(\cdot\mid3,1.5)$, where $\phi(\cdot\mid\mu,\sigma^2)$ denotes the Gaussian density with mean $\mu$ and variance $\sigma^2$. Using the mixture representation \eqref{eq:mixture}, we write $f^{(X)}$ as $f_{G^{\ast}}^{(X)}$, with $G^{\ast}=0.3\,\delta_{(-1,2)}+0.7\,\delta_{(3,1.5)}$ the true mixing distribution. We consider the Laplace noise distribution for the random variables $Z_{i}$'s, namely $Z_i\stackrel{\mathrm{iid}}{\sim}\mathrm{Laplace}(0,b_l)$ with $b_l=\frac{\sigma_l}{\sqrt{2}}$ and $\sigma_l\in\{0.25,0.50\}$, where $\sigma_l$ is the standard deviation of the Laplace noise; this is arguably the most common example of ordinary-smooth noise distribution. The  $Z_{i}$'s are independent of the $X_{i}$'s, and the observations are modeled as follows: $Y_i=X_i+Z_i$, for $i=1,\ldots,n$.

We use the same implementation of Newton's algorithm as described in Section~\ref{sec:unimodal-main} for the unimodal example: the parameter space, numerical grids, initial distribution, learning rate, trapezoidal quadrature, renormalization, and number of permutations are unchanged. Figure~\ref{fig:main-bimodal1-laplace} reports the quasi-Bayes estimate of $f_{G}^{(X)}$, and the asymptotic credible intervals and bands.

\begin{figure}[t]
 \centering
 \includegraphics[
   width=\textwidth,
   trim=20 15 20 0,
   clip
 ]{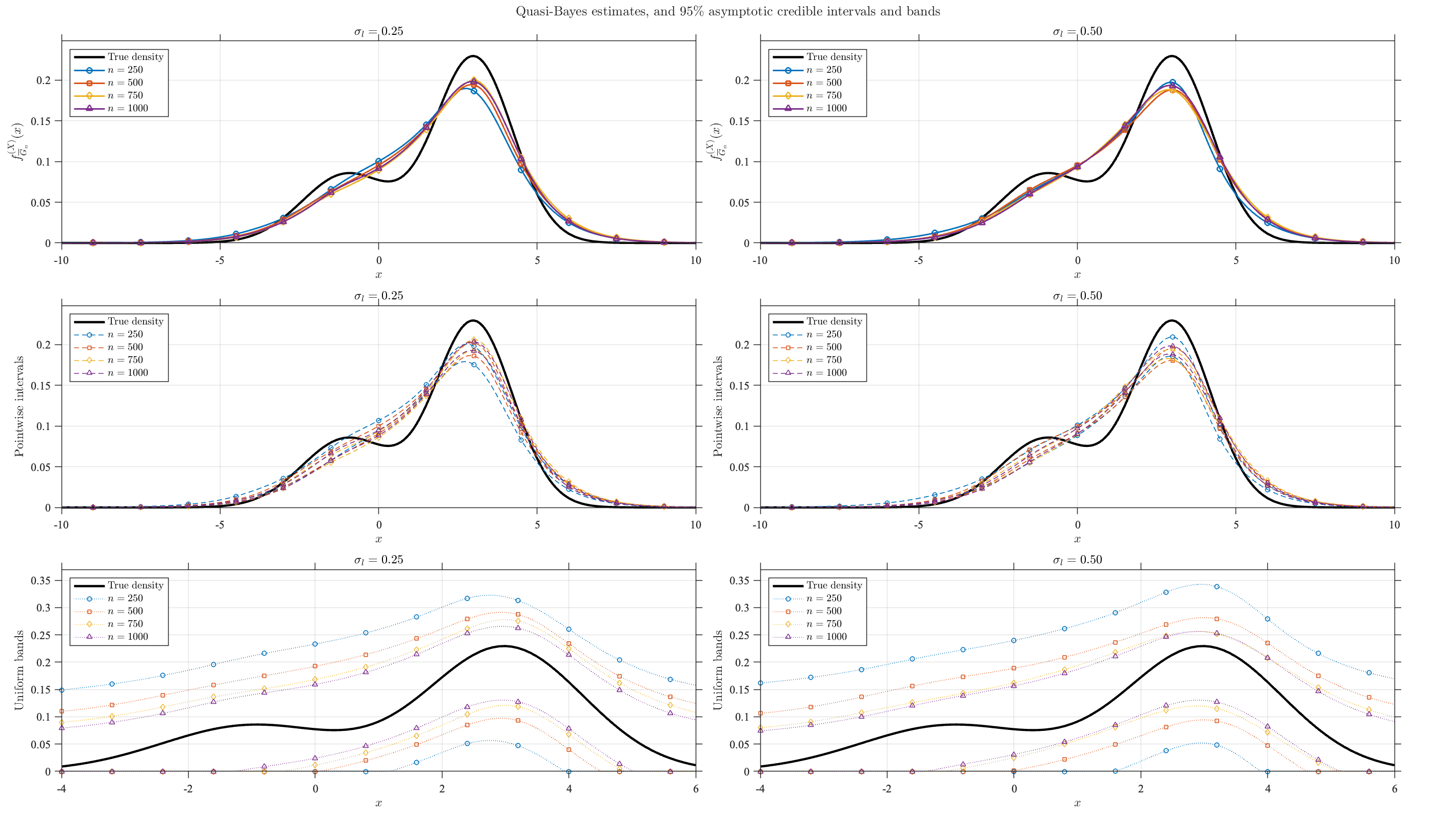}
  \caption{\footnotesize{Bimodal example, with Laplace noise: estimates, credible intervals and bands.}} \label{fig:main-bimodal1-laplace}
\end{figure}

We refer to \ref{sec:supp-bimodal1} for a more detailed analysis of the bimodal example, including Monte Carlo credible intervals and bands and a comparison between the quasi-Bayesian and Bayesian approaches. A second synthetic-data bimodal example is investigated in \ref{sec:supp-bimodal2}. The analyses in \ref{sec:supp-bimodal1}--\ref{sec:supp-bimodal2} are conducted under the Laplace and Gaussian noise distributions.


\section{Real-data analysis}\label{sec61}

We consider flow-cytometry data from \citet{Tor(23)} on the expression of the mesodermal transcription factor Brachyury in differentiating mouse embryonic stem cells. In flow cytometry, cells pass individually through laser beams, and the emitted and scattered light is converted into digital intensity measurements for each detector channel. The fluorescence from a labeled cell contains both the reporter signal of interest and an intrinsic background component, known as autofluorescence, which may obscure low
expression levels. The experiment therefore includes separate labeled and unlabeled cell populations, the latter characterizing the background distribution \citep{Tor(23)}. This yields a deconvolution problem aimed at recovering the population distribution of the latent reporter signal. Since the cells are recorded successively, we retain their acquisition order and process the observations sequentially rather than as an arbitrarily ordered static sample. Further details on the experimental setting  are provided in \ref{sec:supp-flow-cytometry}.

The publicly available data consist of fluorescence intensities and acquisition-time measurements; we refer to \url{https://github.com/dsb-lab/scBayesDeconv.jl}. Following the organization of the dataset in \citet{Tor(23)}, we use the first column as the sample from the labeled cell population and the sixth as the control sample describing autofluorescence. Since the two columns correspond to distinct cellular populations, their observations are not paired event-by-event. We order the labeled-cell measurements by acquisition time, and then examine the estimates over the nested prefixes corresponding to the first $n\in\{250,\,1{,}000,\,5{,}000,\,10{,}000\}$ events. Retaining the acquisition order does not amount to assuming a time-series dependence structure; it only reflects the order in which the observations become available to the sequential procedure. Let $Y_i$ denote the fluorescence intensity of the $i$-th labeled cell in acquisition order and assume the additive model $Y_i=X_i+Z_i$ for $i\geq1$, where $X_i$ is the reporter fluorescence and $Z_i$ is the autofluorescence and background contribution. Further, the sixth column provides $m\geq1$ independent control samples, say $\widetilde Z_1,\ldots,\widetilde Z_m$, measured on unlabeled cells and not paired with the $Y_i$'s.

We assume that the random variables $X_{i}$'s are i.i.d. according to an unknown density $f_X$, that the noise random variables $Z_{i}$'s are i.i.d. according to a density $f_Z$, and that the random variables $X_{i}$'s and $Z_{i}$'s are mutually independent and independent of the control random variables $\widetilde Z_1,\ldots,\widetilde Z_m$. Further, $\widetilde Z_1,\ldots,\widetilde Z_m$ are i.i.d. with the same distribution as the $Z_i$'s. The inferential target is the density $f_X$, which is represented as the Gaussian location-scale mixture $f_{G}^{(X)}(x)= \int_{\mathbb R\times\mathbb R^{+}}\phi(x\mid\mu,\sigma^2)\,G(d\mu,d\sigma^2)$, for $x\in\mathbb{R}$, with $G$ being the unknown mixing distribution. For the noise distribution with density $f_{Z}$ we consider a Gaussian distribution, i.e. $f_Z(\cdot)=\phi(\cdot\mid\mu_Z,\sigma_Z^2)$. In particular, the density $f_Z$ is estimated once from the control samples $\widetilde Z_1,\ldots,\widetilde Z_m$ and it is subsequently kept fixed for all sample sizes.

To implement the quasi-Bayesian approach for estimating $f_G^{(X)}$, we apply Newton's algorithm~\eqref{eq:newton}, processing the observations $Y_i$'s in acquisition order, with Gaussian kernel $k(\cdot\mid\theta)=\phi(\cdot\mid\mu,\sigma^2)$, where $\theta=(\mu,\sigma^2)\in\mathbb R\times\mathbb R^+$. On the standardized working scale, the parameter space $\mathbb R\times\mathbb R^+$ is restricted to $\Theta=[\mu_{\min},\mu_{\max}]\times[0.01,5]$, where the bounds $\mu_{\min}$ and $\mu_{\max}$ are determined from the observed and control samples. We use the learning rate $\widetilde\alpha_i=(1+i)^{-1}$, for $i\geq1$, and set $\widetilde G_0$ equal to the Uniform distribution on $\Theta$; this is precisely as in the synthetic-data analysis. All integrals with respect to the mixing distribution are evaluated using the two-dimensional trapezoidal rule, and $\widetilde G_i$ is numerically renormalized after each update. For numerical implementation, the location and variance coordinates are discretized using, respectively, $121$ equally spaced points and $45$ logarithmically spaced points, giving $5{,}445$ grid points. 

As additional benchmarks, we also consider a Bayesian approach under the Dirichlet-process Gaussian-mixture model \eqref{eq:supp-unimodal-dp-model}, with posterior inference performed sequentially using the SMC algorithm, and the ridge-regularized Fourier deconvolution approach \citep{Ste(90),Hal(07)}. We omit further implementation details for these competing procedures, which all apply the same fitted noise distribution. For the nested sample sizes $n\in\{250,\,1{,}000,\,5{,}000,\,10{,}000\}$, Figure~\ref{fig:main-fluo-gauss} reports the estimates of $f_G^{(X)}$ obtained by the quasi-Bayesian, Bayesian, and ridge-regularized Fourier deconvolution approaches.

\begin{figure}[t]
 \centering
 \includegraphics[
   width=\textwidth,
   trim=20 15 20 0,
   clip
 ]{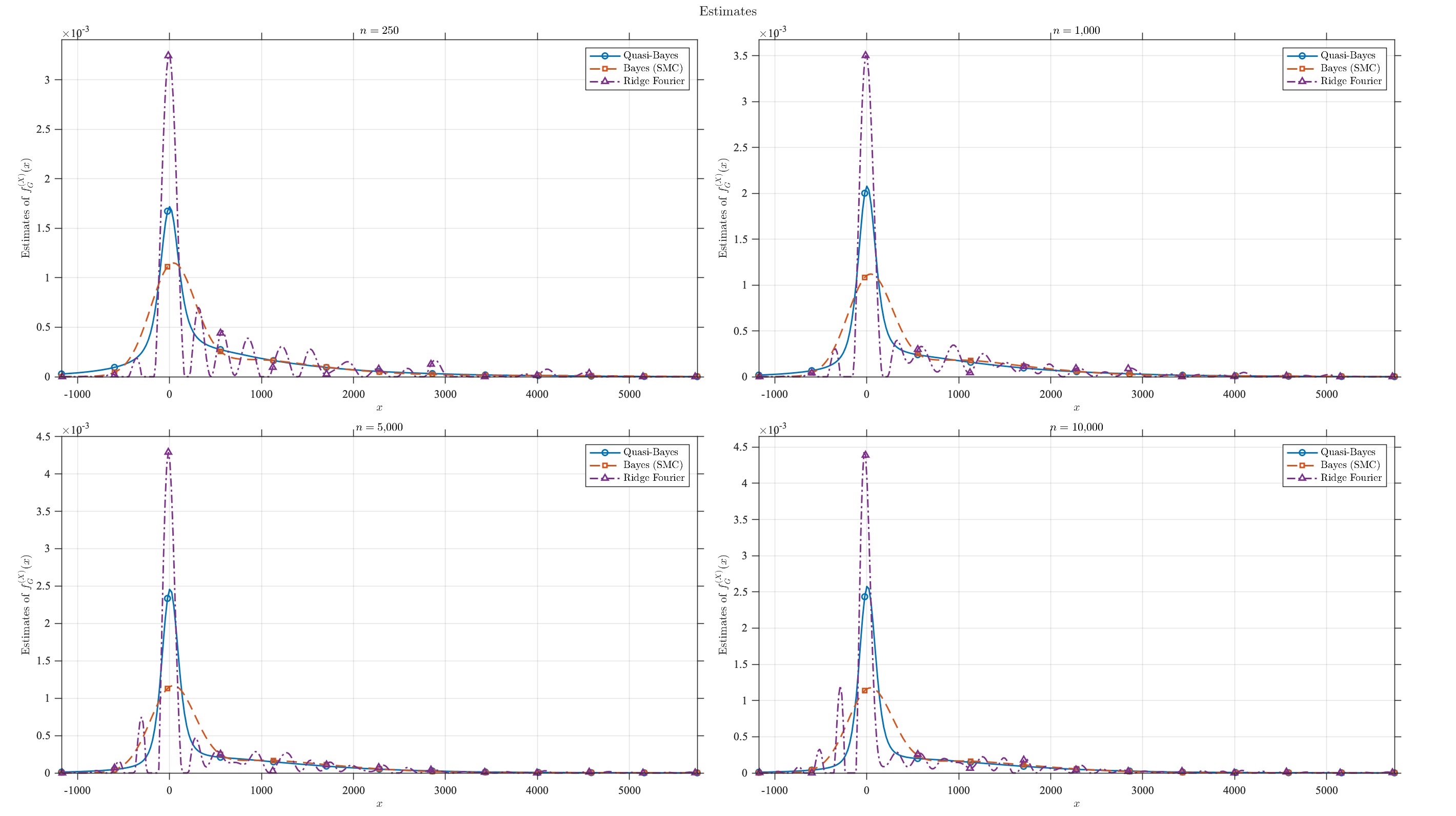}
  \caption{\footnotesize{Flow-cytometry data, with Gaussian noise: comparison between quasi-Bayesian, Bayesian and ridge-regularized Fourier deconvolution approaches.}} \label{fig:main-fluo-gauss}
\end{figure}

We refer to \ref{sec:supp-flow-cytometry} for a more detailed analysis of the flow-cytometry data, including a sensitivity analysis with respect to Gaussian, Laplace, and four-component Gaussian-mixture noise distributions, with the latter being originally proposed by \citet{Tor(23)}; for these analyses, we also report figures augmented with histograms of a reconstructed (proxy) version of the reporter fluorescence samples. Two additional real-data analyses on stellar metallicity data and active power-output data are presented in \ref{sec:supp-stellar} and \ref{sec:supp-power}.


\section{Discussion}\label{sec7}

Although our analysis has focused on univariate observations, the proposed quasi-Bayesian sequential methodology naturally extends to multivariate settings. The structure of Newton's recursive update remains unchanged, though the computational cost per update increases with dimension, primarily due to the greater complexity of numerical integration. We expect local central limit theorems and the associated credible intervals to continue to hold in higher dimensions, along with asymptotic frequentist properties such as consistency and merging. However, uncertainty quantification in the multivariate setting presents additional challenges: in particular, constructing uniform credible bands becomes more delicate, as their width typically increases with dimension, potentially limiting interpretability and coverage accuracy.

While we establish the consistency of the estimate $f_{\widetilde G_n}^{(X)}$ in $L^1$, its rate of convergence in the same norm remains an open problem. Deriving such a rate for quasi-Bayesian sequential deconvolution is likely to be technically challenging, but would provide valuable insight into the statistical efficiency of the method. As a preliminary step, we present in \ref{app:frequentist} a convergence rate for finite mixture models, which may provide a foundation for future work.

A natural extension concerns unknown noise distributions. Recent developments on joint estimation of signal and noise in static or batch settings suggest that recursive strategies could be adapted to this scenario. A quasi-Bayesian sequential implementation is an important direction for future research, also related to privacy-preserving inference in streaming settings, where noise is intentionally introduced through a known or partially known perturbation mechanism \citep{Dwo(10),Car(22),He(24)}.


\section*{Acknowledgement}

Stefano Favaro wishes to thank Mario Beraha for the stimulating discussions on Bayesian deconvolution, in connection with differential privacy, and for help with the code. Stefano Favaro received funding from the European Research Council (ERC) under the European Union's Horizon 2020 research and innovation programme under grant agreement No 817257.



\section*{Appendix}
\renewcommand{\thesection}{\Alph{section}}
\renewcommand{\theequation}{\thesection.\arabic{equation}}
\setcounter{section}{1}
\setcounter{equation}{0}
\setcounter{thm}{0}

\section{Proofs}
\subsection{Preliminary results}
\subsubsection{Almost-sure conditional convergence}\label{appB}
In this Section, we report some definitions and results, that are critical to Section \ref{sec3}. Given a probability space $(\Omega,\mathcal F,P)$ and a Polish space $\mathbb X$ with its Borel sigma-algebra $\mathcal X$, a  kernel on $\mathbb X$ is a function $K:\Omega\times\mathcal X$ satisfying:
\begin{itemize}
    \item[(i)] for every $\omega\in\Omega$, $K(\omega,\cdot)$ is a probability measure on $\mathcal  X$;
    \item[(ii)] for each $B\in\mathcal X$, the function $K(\cdot,B)$ is measurable with respect to $\mathcal F$.
\end{itemize}
A Kernel $K$ on $\mathbb R$ is called a Gaussian kernel if $K(\omega,\cdot)$ is a Gaussian distribution with mean $\mu(\omega)$ and variance $\sigma^2(\omega)$, where $\mu$ and $\sigma^2$ are random variables defined on $(\Omega,\mathcal F,P)$. We denote the Gaussian kernel by $\mathcal N(\mu,\sigma^2)$ and interpret the Gaussian distribution with zero variance as the degenerate law centered on the mean. 
A kernel $K$ on the Polish space $C(\mathbb R)$, with the topology of uniform convergence on compact sets,  is called a Gaussian process kernel if $K(\omega,\cdot)$ is the distribution of a Gaussian process in $C(\mathbb R)$, with mean function $\mu(\omega)\in C(\mathbb R)$ and covariance function $R(\omega)\in C(\mathbb R\times\mathbb R)$.
If $\mu(\omega)=0$ for all $\omega$, then the Gaussian process kernel is said to be centered.

We now extend the definition of almost sure conditional convergence, given by \cite{crimaldi2009} for real sequences, to sequences taking values in Polish spaces.

\begin{defi}
    Let $(X_n)$ be a sequence of random variables  adapted to a filtration $(\mathcal G_n)$ and taking values in a Polish space $\mathbb X$ with its Borel sigma-algebra $\mathcal X$. For every $n\geq 0$, let $K_n$ denote a regular version of the conditional distribution of $X_{n+1}$, given $\mathcal G_n$. If there exists a kernel $K$ such that the sequence $(K_n(\omega,\cdot))_{n}$ converges weakly to $K(\omega,\cdot)$ for almost every $\omega\in\Omega$, then we say that the sequence $(X_n)$ converges to $K$ in the sense of almost-sure conditional convergence.
\end{defi}

Next results are the main tool in the proof of Theorem \ref{th:clt}.

\begin{thm}[\cite{crimaldi2009}, Theorem A.1]
\label{th:app:almost sure conditional}
For each $n\geq 1$, let $(M_{n,j})_{j\geq 1}$ be a real valued martingale with respect to the filtration $(\mathcal F_{n,j})_{j\geq 1}$, satisfying $M_{n,0}=0$, and converging in $L^1$ to a random variable $M_{n,\infty}$. Set
$$
X_{n,j}=M_{n,j}-M_{n,j-1}\mbox{ for }j\geq 1, \quad U_n=\sum_{j\geq 1}X_{n,j}^2,\quad X_n^*=\sup_{j\geq 1}|X_{n,j}|.
$$
Assume that
\begin{itemize}
    \item[(a)] $(X_n^*)_n$ is dominated in $L^1$ and converges to zero a.s.
    \item[(b)] $(U_n)_n$ converges a.s. to a non-negative random variable $U$.
\end{itemize}
Then the the sequence $(M_{n,\infty})$ 
                    converges to the Gaussian kernel $\mathcal N(0,U)$ in the sense of almost-sure conditional convergence.
\end{thm}

\begin{remark}
   Requesting that the almost sure limit of $(U_n)$ is non-negative guarantees the well-definedness of the kernel $\mathcal N(0,U)$ for every $\omega \in \Omega$. Theorem A.1 in {\rm \cite{crimaldi2009}}  specifies that $U$ should be a \lq\lq positive random variable''. Although not explicitly stated, this requirement should be understood as $U \geq 0$, as becomes evident upon a careful examination of the proof.
\end{remark}

\begin{thm}[\cite{crimaldi2016}, Lemma 4.1]
\label{th:app:almost sure2}
    Let $(Z_n)$ be a sequence of real valued random variables adapted to the filtration $(\mathcal G_n)$ and such that $E(Z_{n+1}\mid \mathcal G_n)\rightarrow Z$ a.s. for some random variable $Z$. Moreover, let $(a_n)$ and $(b_n)$ be sequences of real numbers such that
    $$
    b_n\uparrow +\infty,\quad \sum_{k=1}^\infty \frac{E(Z_k^2)}{a_k^2b_k^2}<+\infty.
    $$
Then we have:
\begin{itemize}
    \item[(a)] If $\frac{1}{b_n}\sum_{k=1}^n \frac{1}{a_k}\rightarrow c$ for some constant $c$, then $\frac{1}{b_n}\sum_{k=1}^n \frac{Z_k}{a_k}\rightarrow cZ$ a.s.
    \item[(b)] If $b_n\sum_{k\geq n}\frac{1}{a_kb_k^2}\rightarrow c$ for some constant $c$, then ${b_n}\sum_{k\geq n} \frac{Z_k}{a_kb_k^2}\rightarrow c Z$ a.s.
\end{itemize}\end{thm}

\subsubsection{A concentration inequality for Gaussian processes}\label{appC}

Let $(T,d)$ be a totally bounded pseudometric space, i.e. $d$ satisfies the axioms of a metric except that does not necessarily separate points.
For every $\delta>0$, a $\delta$-packing of $T$ is a finite subset $T_0$ of $T$ such that for every $t_1,t_2\in T_0$, $d(t_1,t_2)>\delta$.
A $\delta$-covering of $T$ is a finite subset $T_1$ of $T$ whose covering radius $r$ satisfies $r\leq \delta$, where the covering radius of $T_1$ is defined as
$$
r=\inf\{s>0:\mbox{ for every }t\in T\mbox{ there exists }t_1\in T_1 \mbox{ with }d(t,t_1)\leq s\}.
$$
The maximal cardinality of a $\delta$-packing of $T$ is called the $\delta$-packing number of $T$ and denoted by $N(\delta,T)$. The minimal cardinality of a $\delta$-covering of $T$ is called the $\delta$-covering number of $T$, and denoted by $N'(\delta,T)$. It follows from the definitions that
$$
N'(\delta,T)\leq N(\delta,T)\leq N'(\delta/2,T).
$$
The $\delta$-entropy number of $T$ is defined as
$$
H(\delta,T)=\log(N(\delta,T)).
$$
The function $H(\cdot,T)$ is called the metric entropy of $(T,d)$. The covariance pseudometric of a centered Gaussian process $(X(t))_{t\in T}$ is defined as 
$$
d(s,t)=E((X(s)-X(t))^2)^{1/2}=(R(s,s)-2R(s,t)+R(t,t))^{1/2},
$$
where $R$ is the covariance function of $(X(t))_{t\in T}$.

\begin{thm}[\cite{massart2007}, Theorem 3.18]
\label{th:supexp}
Let $(X(t))_{t\in T}$ be a centered Gaussian process  and let $d$ be the covariance pseudometric of $(X(t))_{t\in T}$. Assume that $(T,d)$ is totally bounded and denote by $H(\delta,T)$ the $\delta$-entropy number of $(T,d)$, for all positive $\delta$. If $\sqrt{H(\cdot,T)}$ is integrable at $0$, then $(X(t))_{t\in T}$ admits a version which is almost surely uniformly continuous on $(T,d)$. Moreover, if $(X(t))_{t\in T}$ is almost surely continuous on $(T,d)$, then
$$
E(\sup_{t\in T}X(t))\leq 12\int_0^\sigma \sqrt{H(z,T)}dz,
$$
where $\sigma=(\sup_{t\in T}E(X^2(t)))^{1/2}$.
\end{thm}

The following concentration inequality allows to give an upper bound in probability for $  \sup_{t\in T}X(t)$.

  \begin{thm}[\cite{massart2007}, Proposition 3.19]
  \label{th:supX}
Let $(X(t))_{t\in T}$ be an almost surely continuous centered Gaussian process on the totally bounded set $(T,d)$, and let
$$
\sigma=\left(\sup_{t\in T}E(X^2(t))\right)^{1/2}.
$$
If $\sigma>0$, then, for every $z>0$,
    $$
    P(\sup_{t\in T}X(t)\geq E(\sup_{t\in T} X(t))+\sigma\sqrt{2z})\leq e^{-z}.
    $$
\end{thm}
If $\sigma=0$, then $\sup_{t\in T}X(t)=E(\sup_{t\in T}X(t))=0$. Hence the above inequality does not hold. However, in such a case, for every $\epsilon>0$, $P(\sup_{t\in T}X(t)\geq \epsilon)=0$. Therefore, whatever $\sigma\geq 0$, we can write that 
$$
  P(\sup_{t\in T}X(t)\geq \max(\epsilon,E(\sup_{t\in T} X(t))+\sigma\sqrt{2z}))\leq e^{-z},
$$
for every $\epsilon >0$ and $z>0$.

\subsection{Proofs of Section \ref{sec2}}\label{appD}
We introduce the following notation: for a filtered probability space $(\Omega,\mathcal F,P,(\mathcal G_n))$, we denote by $P_n(\cdot)=P(\cdot\mid\mathcal G_n)$ and by $E_n(\cdot)=E(\cdot\mid\mathcal G_n)$. We also denote by $X_{1:n}=(X_1,\dots,X_n)$.

\subsubsection{Proof of Proposition \ref{prop:lawX}}\label{sec:appprop1}
  Define recursively the law of $(X_n)$ and $(Z_n)$ for $n\geq 1$ by setting the conditional density of $X_{n+1}$, given $X_{1:n},Z_{1:n},Y_{1,n}$ as in \eqref{eq:predX}, the conditional density of $Z_{n+1}$, given $X_{1:n+1},Z_{1:n},Y_{1,n}$, as $f_Z$. Then define $Y_{n+1}=X_{n+1}+Z_{n+1}$ and $\widetilde G_{n+1}$ as in \eqref{eq:newton}. Let $(\Omega,\mathcal F,P)$ be a probability space on which all the random variables $(X_n,Y_n,Z_n)$ are defined. Denoting by $E_{n-1}$ the conditional expectation given $(X_{1:n-1},Y_{1:n-1},Z_{1:n-1})$, we can write that
  \begin{align*}
E_{n-1}(e^{isY_n})&=\varphi_{f_Z}(s)\int\varphi_{k}(s\mid\theta)\widetilde G_{n-1}(\ddr\theta)\\
&=\int \varphi_{k}(s\mid\theta) \varphi_{f_Z}(s)\widetilde G_{n-1}(\ddr\theta)\\
&=\int \varphi_{\widetilde k}(s\mid\theta)\widetilde G_{n-1}(\ddr\theta),
\end{align*}
where $\varphi_{f_Z}$ and  $\varphi_{\widetilde k}(s\mid\theta)$ denote the Fourier transforms of $f_Z$ and of  $\widetilde k(\cdot\mid\theta)=k(\cdot\mid\theta)*f_Z$, respectively. Hence \eqref{eq:predY} holds. 
Moreover, the conditional distribution of $X_{n+1}$, given $X_{1:n},Z_{1:n},Y_{1,n}$ is uniquely identified by
\begin{align*}
E_{n}(e^{isY_{n+1}})&=E_n(e^{isZ_{n+1}}e^{isX_{n+1}})\\
&=\varphi_{f_Z}(s)E_{n}(e^{isX_{n+1}}).
\end{align*}

\subsubsection{Proof of Theorem \ref{teo_cid}}\label{app4_0}
The convergence of the sequence $(\widetilde G_n)$ to a random probability measure $\widetilde G$ follows directly from Theorem~1 in \cite{For(20)}. For completeness, and to highlight the role of the filtration $(\mathcal G_n)$ in our setting, we briefly outline the argument.
For every $B\in \mathcal B(\Theta)$, the sequence $(\widetilde G_n(B))_{n\geq 0}$ is adapted to the filtration $(\mathcal G_n)_{n\geq 0}$. From \eqref{eq:predY} the conditional distribution of $Y_{n+1}$, given $\mathcal G_n$ has density $f_{\widetilde G_n}^{(Y)}(y)=\int_\Theta \widetilde k(y\mid\theta)\widetilde G_n(\ddr\theta)$ with respect to $\lambda$. Thus, for every $n\geq 0$,
\begin{align*}
&E(\widetilde G_{n+1}(B)\mid\mathcal G_n)\\
&=(1-\widetilde\alpha_{n+1})\widetilde G_n(B)+\widetilde\alpha_{n+1}\int_{B}\int \frac{\widetilde k(y\mid\theta)}{\int \widetilde k(y\mid\theta')\widetilde G_n(\ddr\theta')}f_{\widetilde G_n}^{(Y)}(y)\lambda(dy)\widetilde G_n(\ddr\theta)\\
&=\widetilde G_n(B).
\end{align*}
Thus, for every $B\in\mathcal B(\Theta)$, the sequence $(\widetilde G_n(B))_{n\geq0}$ is a bounded martingale with respect to the filtration $(\mathcal G_n)_{n\geq 0}$. By Lemma 7.14 in \cite{aldous1985}, there exists a random probability measure $\widetilde G$ on $\mathcal B(\Theta)$ such that $\widetilde G_n$ converges to $\widetilde G$ weakly $P$-a.s. and for every $n,k\in\mathbb N$ and every bounded continuous function $h$ on $\Theta$, 
$$\int_\Theta h(\theta)\widetilde G_n(\ddr\theta)=E\left(\int_\Theta h(\theta)\widetilde G_{n+k}(\ddr\theta)\mid\mathcal G_n\right)=E\left(\int_\Theta h(\theta)\widetilde G(\ddr\theta)\mid\mathcal G_n\right)\quad P\mbox{-a.s.}$$
We now turn to the latent sequence $(X_n)$, which is not directly covered by \cite{For(20)}. Since $k(x\mid\theta)$ is continuous in $\theta$,
then
$$
f_{\widetilde G_n}(x)=\int_\Theta k(x\mid\theta)\widetilde G_n(\ddr\theta)\rightarrow \int_\Theta k(x\mid\theta)\widetilde G(\ddr\theta)
$$
$P$-a.s., for $\lambda$-almost all $x$, as $n\rightarrow+\infty$.
 Moreover, for every $n\geq 0$, $k\geq 1$ and $A\in\mathcal B(\mathbb R)$,
\begin{align}\label{eq:martpred}
P(X_{n+k}\in A\mid\mathcal G_n)&=E(P(X_{n+k}\in A\mid\mathcal G_{n+k-1})\mid\mathcal G_n)\notag\\
&=E\left(\int_A\int_\Theta k(x\mid\theta)\widetilde G_{n+k-1}(\ddr\theta)\lambda(dx)\mid\mathcal G_n\right)\notag\\
&=\int_A\int_\Theta k(x\mid\theta)\widetilde G_{n}(\ddr\theta)\lambda(dx)\notag\\
&=P(X_{n+1}\in A\mid\mathcal G_n).
\end{align}
We show that the sequence $(X_n)$ is conditionally identically distributed in the sense of \cite{Berti(04)}. The assumptions of the learning model include the conditional independence of $X_{n+1}$ from $X_1,\dots,X_n$, given $Y_1,\dots,Y_n$ and the independence of $Z_{n+1}$ from $X_1,\dots,X_n,Y_1,\dots,Y_n$, for every $n\geq 0$ with the convention that a variable with a $0$ subscript is omitted. This implies that for every $k\geq 1$, $X_{n+k}$ is conditionally independent of $X_1,\dots,X_n$, given $Y_1,\dots,Y_n$. Indeed in a directed acyclic graph representing the joint probability distribution of $(X_1,Y_1,Z_1,\dots, X_{n+k},Y_{n+k},Z_{n+k})$, the random blocks $X_{n+k}$ and $(X_1,Z_1,\dots,X_n,Z_n)$ are $d$-separated by $Y_n$. Let $(\mathcal G_n^{a})_{n\geq 0}$ denote the natural filtration of the random  sequence $(X_n,Y_n,Z_n)_{n\geq 1}$. This enlargement is needed to ensure that $(X_n)$ is adapted. It follows by \eqref{eq:martpred} and the conditional independence of $X_{n+k}$ that, for every $A\in\mathcal B(\R)$,
$$
P(X_{n+k}\in A\mid \mathcal G_n^{a})=P(X_{n+k}\in A\mid\mathcal G_n)=P(X_{n+1}\in A\mid\mathcal G_n)=P(X_{n+1}\in A\mid\mathcal G_n^{a}).
$$
Thus, the sequence $(X_n)_{n\geq1}$ is adapted to the filtration $(\mathcal G_n^a)_n$ and conditionally identically distributed with respect to it (\cite{Berti(04)}).  Its directing random measure has density $f_{\widetilde G}^{(X)}$ with respect to $\lambda$, since for every $A$
\begin{align*}
\lim_{n\rightarrow+\infty}P(X_{n+1}\in A\mid\mathcal G_n^a)
&=\lim_{n\rightarrow+\infty}P(X_{n+1}\in A\mid\mathcal G_n)\\&=\lim_{n\rightarrow+\infty}\int_A\int_{\Theta}k(x\mid\theta)\widetilde G_n(\ddr\theta)\lambda(dx)\\
&=\int_A\int_{\Theta}k(x\mid\theta)\widetilde G(\ddr\theta)\lambda(dx)\\&=\int_Af_{\widetilde G}^{(X)}(x)\lambda(dx).
\end{align*}
By Lemma 2.1 and Theorem 2.2 in \cite{Berti(04)} applied to $f(X_n)=I(X_n\in A)$, we deduce that $\int_A f_{\widetilde G}(x)\lambda(dx)$ is the $P$-a.s. limit of $\frac 1 n \sum_{k=1}^n I(X_k\in A)$ for every $A\in\mathcal B(\mathbb R)$. Moreover,
\begin{align*}
P(X_{n+k}\in A\mid\mathcal G_n^a)&=
P(X_{n+k}\in A\mid\mathcal G_n)\\&=E\left(\int_A \int_\Theta k(x\mid\theta)\widetilde G_{n+k-1}(\ddr\theta)\lambda(dx)\;\middle|\;\mathcal G_n\right)\\
&=E\left(\int_A \int_\Theta k(x\mid\theta)\widetilde G(\ddr\theta)\lambda(dx)\;\middle|\;\mathcal G_n\right)\\
&=E\left(\int_A f_{\widetilde G}^{(X)}(x)\lambda(dx)\;\middle|\; \mathcal G_n\right).
\end{align*}
It follows that
$$
P(X_n\in A)=E\left(\int_Af_{\widetilde G}^{(X)}(x)\lambda(dx)\right)=\iint_A f_G^{(X)}(x)\lambda(dx)\pi(dG),
$$
where $\pi$ is the probability distribution of $\widetilde G$.\\
To prove that $f_{\widetilde G_n}^{(X)}$ converges to $f_{\widetilde G}^{(X)}$ in $L^1$, we can apply the dominated convergence theorem to the sequence of functions $(|f_{\widetilde G_n}^{(X)}-f_{\widetilde G}^{(X)}|)_{n\geq 0}$, which converges to zero $\lambda$-a.e. ($P$-a.s.) and is dominated by $\sup_{\theta\in\Theta}k(x\mid \theta)$ that is integrable with respect to $\lambda$ by assumption.

\subsection{Proofs of Section \ref{sec3}}\label{appE}

\subsubsection{Proof of Theorem \ref{th:clt}}\label{app:prrofth2}
\begin{lem}
   \label{lem:vn}
    Let $v(x)$ be defined as in \eqref{eq:vx} and let $v_n(x)$ be defined as $v(x)$ with $\widetilde{G}_{n}$ in place of $\widetilde{G}$. Then, for every $x\in\mathbb R$, $v_n(x)$ is strictly positive and converges to $v(x)$, $P$-a.s. 
     \end{lem}
     \begin{proof} Let us denote by $k_x(\cdot)=k(x\mid\cdot)$ and $\widetilde k_y(\cdot)=\widetilde k(y\mid\cdot)=(f_Z*k(\cdot\mid\cdot))(y)$. \\
     To show that $v_n(x)\neq 0$ $P$-a.s. we proceed by contradiction. If $v_n(x)=0$, then 
     $$\int_\Theta k(x\mid\theta)\left(\widetilde G_n(\ddr\theta\mid y)-\widetilde G(\ddr\theta)\right)=0$$
     for $\lambda$-almost every $y$. By the assumption A4), the last equality  entails $\widetilde G_n(\cdot \mid y)=\widetilde G_n(\cdot)$ for  $\lambda$-almost every $y$. On the other hand, by the assumptions A1), A2) and A5), $\mbox{supp}(\widetilde G_n)=\Theta$, which implies $\widetilde k(\cdot\mid\theta)\neq \int \widetilde k(\cdot\mid\theta')\widetilde G_n(\ddr\theta')$, leading to $\widetilde G_n(\cdot\mid y)\neq \widetilde G_n(\cdot)$ for $y$ in a set of $\lambda$-positive measure. 

     We now prove that $v_n(x)$ converges to $v(x)$, $P$-a.s. First, notice that, since we have $\int \sup_{\theta\in\Theta}k(x\mid\theta)\lambda(\ddr x)<+\infty$, then there exists a density function $h$ with respect to $\lambda$ such that
     $$
     h(x)=\frac{\sup_{\theta\in\Theta}k(x\mid\theta)}{\int \sup_{\theta\in\Theta}k(x\mid\theta)\lambda(\ddr x)}.
     $$
    Thus, we can write that
     \begin{align*}
     \int \sup_{\theta\in\Theta}\widetilde k(y\mid\theta)\lambda(\ddr y)&=\int \sup_{\theta\in\Theta}\int k(x\mid\theta)f_Z(y-x)\lambda(\ddr x)\lambda(\ddr y)\\
     &\leq C \iint h(x)f_Z(y-x)\lambda(\ddr x)\lambda(\ddr y)<+\infty,
     \end{align*}
     where $C=\int \sup_{\theta\in\Theta}k(x\mid\theta)\lambda(\ddr x)$.
     Moreover,
      \begin{align*}
                \int_\Theta \widetilde k_y \ddr \widetilde G_n&=\int \int_\Theta
k(x\mid\theta)f_Z(y-x)\lambda(\ddr x)\widetilde G_n(\ddr\theta)\leq C  (f_Z*h)(y),
\end{align*}
which is a function integrable with respect to $\lambda$.

      
       By assumption A3) and Theorem \ref{teo_cid},
      $$
      \int_\Theta k(x\mid\theta)\left(\widetilde G_n(\ddr\theta\mid y)-\widetilde G_n(\ddr\theta)\right) 
      \rightarrow
      \int_\Theta k(x\mid\theta)\left(\widetilde G(\ddr\theta\mid y)-\widetilde G(\ddr\theta)\right)\quad P\mbox{-a.s.}
      $$
      Since $\left[\int_\Theta k(x\mid\theta)\left(\widetilde G(\ddr\theta\mid y)-\widetilde G(\ddr\theta)\right)\right]^2\leq \sup_\theta k(x\mid\theta)^2 $ and $\int_\Theta \widetilde k_y\ddr \widetilde G_n\leq C(f_Z*h)(y)$ which is integrable with respect to $\lambda$, then by dominated convergence theorem,
   \begin{align*}
       v_n(x)
   \stackrel{P-a.s.}{\longrightarrow }&
   \int\left[\int_\Theta k(x\mid\theta)\left(\widetilde G(\ddr\theta\mid y)-\widetilde G(\ddr\theta)\right)\right]^2\left(\int_\Theta \widetilde k_y\ddr \widetilde G\right)\lambda(dy)\\
&=\int\left[f_{\widetilde G}^{(X)}(x\mid y)-f_{\widetilde G}^{(X)}(x)\right]^2f_{\widetilde G}^{(Y)}(y)\lambda(dy)=v(x).
      \end{align*}
     \end{proof}
     
\begin{proof} [Proof of Theorem \ref{th:clt}]
       The proof is based on \citet[Theorem A.1]{crimaldi2009}  (see  Theorem \ref{th:app:almost sure conditional} in Section B). Fix  $x\in\mathbb R$ and define, for each $n\geq 0$,
       $$
       \mathcal F_{n,0}=\mathcal G_n,\quad \mathcal F_{n,j}=\mathcal G_{n+j-1}\mbox{ for }j\geq 1,
       $$
       $$
       M_{n,0}=0,\quad M_{n,j}= b_n^{1/2}(f_{\widetilde G_n}^{(X)}(x)-f_{\widetilde G_{n+j-1}}^{(X)}(x))\mbox{ for }j\geq 1.
       $$
       For every $n\geq 0$, the sequence $(M_{n,j})_{j\geq 0}$ is a bounded martingale with respect to the filtration $(\mathcal F_{n,j})_{j\geq 0}$, converging $P$-a.s. and in $L^1$ to $$M_{n,\infty}= b_n^{1/2}(f_{\widetilde G_n}^{(X)}(x)-f_{\widetilde G}^{(X)}(x)).$$
       Let
       $$
       X_{n,j}=M_{n,j}-M_{n,j-1}=b_n^{1/2}(f_{\widetilde G_{n+j-1}}^{(X)}(x)-f_{\widetilde G_{n+j}}^{(X)}(x)),
       $$
       $$
       X_n^*=\sup_{j\geq 1}|X_{n,j}|= b_n^{1/2} \sup_{j\geq 1} |f_{\widetilde G_{n+j}}^{(X)}(x)-f_{\widetilde G_{n+j-1}}^{(X)}(x)|,
       $$
       and
       $$
       U_n=\sum_{j=0}^\infty X_{n,j}^2=b_n\sum_{k\geq n}(f_{\widetilde G_{k+1}}^{(X)}(x)-f_{\widetilde G_k}^{(X)}(x))^2.
       $$
       The thesis follows from Theorem \ref{th:app:almost sure conditional}, if we can verify that:
  \begin{itemize}
            \item[(a)] The sequence $(X^*_n)$ is dominated in $L^1$ and converges to zero $P$-a.s. as $n\rightarrow+\infty$.
      \item[(b)] $
      U_n\rightarrow v(x) $,   $P$-a.s. as $n\rightarrow+\infty$.
  \end{itemize}
   Notice that
   \begin{align*}
        f_{\widetilde G_k}^{(X)}(x)-  f_{\widetilde G_{k-1}}^{(X)}(x)&=\int_\Theta k(x\mid\theta) (\widetilde G_{k}(\ddr\theta)-\widetilde G_{k-1}(\ddr\theta))\\
       &=\widetilde\alpha_{k} \int_\Theta k(x\mid\theta)\left(\widetilde G_{k-1}(\ddr\theta\mid Y_k)-\widetilde G_{k-1}(\ddr\theta)\right),
   \end{align*}
   where $\widetilde G_{k-1}(\ddr\theta\mid Y_k):=\widetilde k(Y_{k}\mid\theta)\widetilde G_{k-1}(\ddr\theta)/\int_\Theta \widetilde k(Y_{k}\mid\theta')\widetilde G_{k-1}(\ddr\theta')$. 
   Thus,
   $$
   X^*_n\leq b_n^{1/2} \sup_{k\geq n}
   \widetilde\alpha_{k} \int_\Theta k(x\mid\theta)\left|\widetilde G_{k-1}(\ddr\theta\mid Y_k)-\widetilde G_{k-1}(\ddr\theta)\right|\leq \sup_{\theta\in\Theta}k(x\mid\theta)b_n^{1/2} \sup_{k\geq n}
   \widetilde\alpha_{k}
     $$
is bounded for $\lambda$-almost every $x$, and converges to zero $P$-a.s., since $b_n\sup_{k\geq n}\widetilde \alpha_k^2\rightarrow 0$.
   To prove (b), we apply \citet[Lemma 4.1 (b)]{crimaldi2016} (see Theorem \ref{th:app:almost sure2} (b)), with 
   $$
   Z_n=\widetilde\alpha_n^{-2}(  f_{\widetilde G_n}^{(X)}(x)-  f_{\widetilde G_{n-1}}^{(X)}(x))^2,
   $$
   and $a_k=(b_k\widetilde\alpha_k)^{-2}$.
Defining, for each $n\geq 1$, $f_{\widetilde G_n}^{(X)}(x\mid y)$ as $f^{(X)}_{\widetilde{G}}(\cdot\mid y)$ in \eqref{eq:xgiveny} with $\widetilde{G}_{n}$ in place of $\widetilde{G}$, we can write that
   \begin{align*}
   E_{n-1}(Z_n)
       &= E_{n-1}\left(
\widetilde\alpha_n^{-2}(  f_{\widetilde G_n}^{(X)}(x)-  f_{\widetilde G_{n-1}}^{(X)}(x))^2
   \right)\\
   &=
   \int\left[\int_\Theta k(x\mid\theta)\left(\widetilde G_{n-1}(\ddr\theta\mid y)-\widetilde G_{n-1}(\ddr\theta)\right)\right]^2f_{\widetilde G_{n-1}}^{(Y)}(y)\lambda(dy)\\
    &=
   \int\left[f_{\widetilde G_{n-1}}^{(X)}(x\mid y)-f_{\widetilde G_{n-1}}^{(X)}(x)\right]^2f_{\widetilde G_{n-1}}^{(Y)}(y)\lambda(dy)=v_n(x).
   \end{align*}
Thus, by Lemma \ref{lem:vn},
$$
   E_{n-1}\left(Z_n
   \right)
  \rightarrow v(x),\quad P-a.s.
  $$ Moreover, $b_n\uparrow+\infty$ by assumption, and
   $$
   \sum_{k=1}^\infty\frac{E(Z_k^2)}{a_k^2b_k^2}=\sum_{k=1}^\infty 
    \widetilde\alpha_k^4b_k^2
  E\left(  
    \left[f_{\widetilde G_{k-1}}^{(X)}(x\mid Y_k)-f_{\widetilde G_{k-1}}^{(X)}(x)\right]^4\right)<+\infty.
   $$
  Since, by assumption, $b_n\sum_{k\geq n}a_k^{-1}b_k^{-2}\rightarrow 1$, then by  Theorem \ref{th:app:almost sure2} (b),
  $$
  b_n\sum_{k\geq n}(f_{\widetilde G_k}^{(X)}(x)-f_{\widetilde G_{k-1}}^{(X)}(x))^2=  b_n\sum_{k\geq n}\frac{Z_k}{a_kb_k^2}\rightarrow  v(x), \quad P\mbox{-a.s.}
        $$
        It follows by Theorem \ref{th:app:almost sure conditional} that
        $
       M_{n,\infty}= b_n^{1/2}(f_{\widetilde G_n}^{(X)}(x)-f_{\widetilde G}^{(X)}(x))
        $ converges to $\mathcal N(0,v(x))$ in the sense of almost-sure conditional convergence. \\
 Now suppose that $\widetilde \alpha_n=(\alpha+n)^{-\gamma}$. With  $b_n=(2\gamma-1)n^{2\gamma-1}$, we have $b_n\uparrow+\infty$, 
  $$
  \sum_{n=1}^\infty \widetilde\alpha_n^4b_n^2=
  \sum_{n=1}^\infty (\alpha+n)^{-4\gamma}(2\gamma-1)^2n^{4\gamma-2}
  <+\infty,
  $$
  $$
  b_n\sup_{k\geq n}\widetilde\alpha_k^2=(2\gamma-1)n^{2\gamma-1}(\alpha+n)^{-2\gamma}\rightarrow 0,
  $$
  and
  $$  b_n\sum_{k\geq n}{\widetilde\alpha_k^2}=(2\gamma-1)n^{2\gamma-1}\sum_{k\geq n}(\alpha+k)^{-2\gamma}\rightarrow 1.
  $$
  \end{proof}

\subsubsection{Proof of Theorem \ref{th:uniformclt}}\label{app:proofth3}
Since a stochastic process taking values in  $C(\mathbb R)$, with the topology of uniform convergence on compact sets, converges in distribution if and only if its restrictions to compact sets converge (see e.g. \cite{kallenberg2002}, Proposition 16.6), then, in this section, we consider, for each $n\geq 1$,  the restriction of $b_n^{1/2}(f_{\widetilde G_n}^{(X)}-f_{\widetilde G}^{(X)})$ to a fixed compact interval $I=[a,b]$. With an abuse of notation, we use the same symbol $b_n^{1/2}(f_{\widetilde G_n}^{(X)}-f_{\widetilde G}^{(X)})$ irrespective of its domain.   Let $\nu_n$ denote the (random) conditional distribution 
of $b_n^{1/2}(f_{\widetilde G_n}^{(X)}-f_{\widetilde G}^{(X)})$, as an element in $C(I)$,
given $\mathcal G_n$.  To prove the theorem, we need to show that 
$\nu_n$
is tight 
and that the conditional finite dimensional distributions of $\nu_n$ converge almost surely to those of a centered Gaussian process $\mathbb G_{\widetilde G}$ with covariance function \eqref{eq:covariance} (see Theorem 7.1 in \cite{billingsley1999}).

\begin{lem}
    \label{lem:tight}
    Let $\nu_n$ denote the (random) conditional distribution of $b_n^{1/2}(f_{\widetilde G_n}^{(X)}-f_{\widetilde G}^{(X)})$, with $(b_n)_{n\geq1}$ a sequence of strictly positive numbers such that $b_n\uparrow+\infty$, $\sum_{n\geq1} \widetilde\alpha_n^4b_n^{2}<+\infty$, $b_n\sup_{k\geq n}\widetilde\alpha_k^2\rightarrow 0$ and $ b_n\sum_{k\geq n}\widetilde \alpha_k^2\rightarrow 1$. Then $(\nu_n)$ is $P$-a.s. tight.
\end{lem}
\begin{proof}
By \citet[Corollary 16.9]{kallenberg2002}, it is sufficient to show that there exists a set $N$ with $P(N)=0$ such that for every $\omega\not\in N$
the following conditions hold:
\begin{itemize}
\item[(i)]
For each positive $\eta$, there exist $k=k(\omega)>0$ and an $n_0=n_0(\omega)$ such that
\begin{equation}
    \label{eq:bill1}
\sup_{n\geq n_0}P(b_n^{1/2}(f_{\widetilde G_n}^{(X)}(a)-f_{\widetilde G}^{(X)}(a))|\geq k\mid\mathcal G_n)(\omega)\leq \eta.
\end{equation}
\item[(ii)] The exists a constant $C$ such that, for every $s,t\in I$ and for $n$ large enough
\begin{equation}
    \label{eq:bill2}
E\left(b_n(f_{\widetilde G_n}^{(X)}(s)-f_{\widetilde G}^{(X)}(s)-f_{\widetilde G_n}^{(X)}(t)+f_{\widetilde G}^{(X)}(t))^2\mid\mathcal G_n\right)(\omega)\leq C(s-t)^2.
\end{equation}
\end{itemize}
The condition (i) holds, since $b_n^{1/2}(f_{\widetilde g_n}^{(X)}(a)-f_{\widetilde g}^{(X)}(a))$ converges in the sense of almost-sure conditional convergence. To prove (ii) we can write by martingales properties that
\begin{align*}
    &E\left(b_n(f_{\widetilde G_n}^{(X)}(s)-f_{\widetilde G}^{(X)}(s)-f_{\widetilde G_n}^{(X)}(t)+f_{\widetilde G}^{(X)}(t))^2\mid\mathcal G_n\right)\\
    &=b_nE\left(\left(\sum_{k\geq n}f_{\widetilde G_{k+1}}^{(X)}(s)-f_{\widetilde G_k}^{(X)}(s)-f_{\widetilde G_{k+1}}^{(X)}(t)+f_{\widetilde G_k}^{(X)}(t)\right)^2\mid\mathcal G_n\right)\\
    &=b_nE\left(\sum_{k\geq n}\left(f_{\widetilde G_{k+1}}^{(X)}(s)-f_{\widetilde G_k}^{(X)}(s)-f_{\widetilde G_{k+1}}^{(X)}(t)+f_{\widetilde G_k}^{(X)}(t)\right)^2\mid\mathcal G_n\right)\\
    &=b_n\sum_{k\geq n}\alpha_k^2 E\left(E\left( \left[
    f_{\widetilde G_k}^{(X)}(t)-f_{\widetilde G_k}^{(X)}(s)-f_{\widetilde G_k}^{(X)}(t\mid Y_{k+1})-f_{\widetilde G_k}^{(X)}(s\mid Y_{k+1})\right]^2\mid\mathcal G_k\right)\mid\mathcal G_n\right)\\
    &=b_n\sum_{k\geq n}\alpha_k^2 E\left(\int \left[
    f_{\widetilde G_k}^{(X)}((t)-f_{\widetilde G_k}^{(X)}(s)-f_{\widetilde G_k}^{(X)}(t\mid y)-f_{\widetilde G_k}^{(X)}(s\mid y)\right]^2f_{\widetilde G_k}^{(Y)}(y)\lambda(\ddr y)\mid\mathcal G_n\right)\\
    &\leq 2b_n\sum_{k\geq n}\alpha_k^2 E\left(\int \left[
    f_{\widetilde G_k}^{(X)}((t)-f_{\widetilde G_k}^{(X)}(s)\right]^2f_{\widetilde G_k}^{(Y)}(y)\lambda(\ddr y)\mid\mathcal G_n\right)\\
    &+2b_n\sum_{k\geq n}\alpha_k^2 E\left(\int \left[
    f_{\widetilde G_k}^{(X)}(t\mid y)-f_{\widetilde G_k}^{(X)}(s\mid y)\right]^2f_{\widetilde G_k}^{(Y)}(y)\lambda(\ddr y)\mid\mathcal G_n\right)\\
    &\leq2 b_n\sum_{k\geq n}\alpha_k^2 E\left( \left[
    \int_\Theta (k(t\mid\theta)-k(s\mid\theta))\widetilde G_k(\ddr\theta)\right]^2\mid\mathcal G_n\right)\\
    &+2b_n\sum_{k\geq n}\alpha_k^2 E\left(\int \left[
    \int_\Theta (k(t\mid\theta)-k(s\mid\theta))\widetilde G_k(\ddr\theta\mid y)\right]^2f_{\widetilde G_k}(y)\lambda(\ddr y)\mid\mathcal G_n\right)\\   
    &\leq 2b_n\sum_{k\geq n}\alpha_k^2 E\left( 
    \int_\Theta (k(t\mid\theta)-k(s\mid\theta))^2\widetilde G_k(\ddr\theta)\mid\mathcal G_n\right)\\
    &+2b_n\sum_{k\geq n}\alpha_k^2 E\left(\int 
    \int_\Theta (k(t\mid\theta)-k(s\mid\theta))^2\widetilde G_k(\ddr\theta\mid y)f_{\widetilde G_k}(y)\lambda(\ddr y)\mid\mathcal G_n\right)\\    
    &\leq 4\sup_{x,\theta} |k'(x\mid\theta)|^2 (s-t)^2,
\end{align*}
where $k'(x\mid\theta)$ denotes the derivative of $k(x\mid\theta)$ with respect to $x$. Thus, properties (i) and  (ii) hold and the sequence $(\nu_n)$ is tight, on a set of probability one. 
\end{proof}

\begin{lem}\label{lem:findim}
    Let $(b_n)_{n\geq1}$ be sequence of strictly positive numbers such that $b_n\uparrow+\infty$, $\sum_{n\geq1} \widetilde\alpha_n^4b_n^{2}<+\infty$, $b_n\sup_{k\geq n}\widetilde\alpha_k^2\rightarrow 0$ and $ b_n\sum_{k\geq n}\widetilde \alpha_k^2\rightarrow 1$. Then, for every $J\geq 1$, $x_1,\dots,x_J\in
    I$ and $u_1,\dots,u_J\in \mathbb R$,  $b_n^{1/2}\sum_{i=1}^J  u_i(f_{\widetilde G_n}^{(X)}(x_i)-f_{\widetilde G}^{(X)}(x_i))$, converges
    in the sense of almost sure conditional convergence, to a Gaussian kernel, with zero mean and variance $\sum_{i,j=1}^Ju_iu_j R(x_i,x_j)$, with $R$ as in \eqref{eq:covariance}.
\end{lem}
\begin{proof} The proof closely follows the steps of the proof of Theorem \ref{th:clt}; nevertheless, we provide it here for the reader’s convenience.

       Without loss of generality, we can assume that $\sum_{j=1}^J u_i^2\leq 1$. 
       The proof is based on \citet[Theorem A.1]{crimaldi2009}  (see  Theorem \ref{th:app:almost sure conditional} in Section B). 
       Define, for each $n\geq 0$,
       $$
       \mathcal F_{n,0}=\mathcal G_n,\quad \mathcal F_{n,j}=\mathcal G_{n+j-1}\mbox{ for }j\geq 1,
       $$
       $$
       M_{n,0}=0,\quad M_{n,j}= b_n^{1/2}\sum_{i=1}^Ju_i(f_{\widetilde G_n}^{(X)}(x_i)-f_{\widetilde G_{n+j-1}}^{(X)}(x_i))\mbox{ for }j\geq 1.
       $$
       For every $n\geq 0$, the sequence $(M_{n,j})_{j\geq 0}$ is a bounded martingale with respect to the filtration $(\mathcal F_{n,j})_{j\geq 0}$, converging in $L^1$ to $$M_{n,\infty}= b_n^{1/2}\sum_{i=1}^Ju_i(f_{\widetilde G_n}^{(X)}(x_i)-f_{\widetilde G}^{(X)}(x_i)).$$
       Let
       $$
       X_{n,j}=M_{n,j}-M_{n,j-1}=b_n^{1/2}\sum_{i=1}^Ju_i(f_{\widetilde G_{n+j-1}}^{(X)}(x_i)-f_{\widetilde G_{n+j}}^{(X)}(x_i)),
       $$
       $$
       X_n^*=\sup_{j\geq 1}|X_{n,j}|= b_n^{1/2} \sup_{j\geq 1} |\sum_{i=1}^Ju_j(f_{\widetilde G_{n+j}}^{(X)}(x_i)-f_{\widetilde G_{n+j-1}}^{(X)}(x_i))|,
       $$
       and
       $$
       U_n=\sum_{j=0}^\infty X_{n,j}^2=b_n\sum_{k\geq n}(\sum_{i=1}^Ju_i(f_{\widetilde G_{k+1}}^{(X)}(x)-f_{\widetilde G_k}^{(X)}(x)))^2.
       $$
       The thesis follows from Theorem \ref{th:app:almost sure conditional}, if we can verify that the following conditions hold:
  \begin{itemize}
            \item[(a)] The sequence $(X^*_n)$ defined by $X^*_n=b_n^{1/2} \sup_{k\geq n} | \sum_{i=1}^Ju_i( f_{\widetilde G_{k+1}}^{(X)}(x_i)-  f_{\widetilde G_k}^{(X)}(x_i)|$ is bounded in $L^1$ and converges to zero $P$-a.s. as $n\rightarrow+\infty$.
      \item[(b)] $
      U_n=b_n\sum_{k\geq n}( \sum_{i=1}^Ju_i( f_{\widetilde G_k}^{(X)}(x_i)-  f_{\widetilde G_{k-1}}^{(X)}(x_i)))^2$ converges $P$-a.s., as $n\rightarrow+\infty$, to
      \begin{align*}
      &\sum_{i,j=1}^Ju_iu_j\int_{\mathbb R}\left( f_{\widetilde G}^{(X)}(x_i|y)-f_{\widetilde G}^{(X)}(x_i)\right)\left(f_{\widetilde G}^{(X)}(x_j\mid y)-f_{\widetilde G}^{(X)}(x_j)\right)f_{\widetilde G}^{(Y)}(y)\lambda(dy).
      \end{align*}
  \end{itemize}
   Notice that, for every $i=1,\dots,J$,
   \begin{align*}
        f_{\widetilde G_k}^{(X)}(x_i)-  f_{\widetilde G_{k-1}}^{(X)}(x_i)&
        =\int k(x_i\mid\theta) (\widetilde G_{k}(\ddr\theta)-\widetilde G_{k-1}(\ddr\theta))\\
       &=\widetilde\alpha_{k} \int k(x_i\mid\theta)\left(\widetilde G_{k-1}(\ddr\theta\mid Y_k)-\widetilde G_{k-1}(\ddr\theta)\right),
   \end{align*}
   where $\widetilde G_{k-1}(\ddr\theta\mid Y_k)=\widetilde k(Y_{k}\mid\theta)\widetilde G_{k-1}(\ddr\theta)/\int_\Theta \widetilde k(Y_{k}\mid\theta')\widetilde G_{k-1}(\ddr\theta')$ is a probability measure on $\Theta$. 
   Thus,
   $$
   X^*_n\leq b_n^{1/2} \sup_{k\geq n}
   \widetilde\alpha_{k}\sum_{i=1}^J|u_i| \int k(x_i\mid\theta)\left|\widetilde G_{k-1}(\ddr\theta\mid Y_k)-\widetilde G_{k-1}(\ddr\theta)\right|\leq \sum_{i=1}^J|u_i| \sup_{\theta\in\Theta} k(x_i\mid\theta) b_n^{1/2} \sup_{k\geq n}
   \widetilde\alpha_{k}
     $$
is bounded and converges to zero $P$-a.s., since $b_n\sup_{k\geq n}\widetilde \alpha_k^2\rightarrow 0$ by assumption.
   To prove (b), we apply \citet[Lemma 4.1 (b)]{crimaldi2016} (see Theorem \ref{th:app:almost sure2} (b)), with 
   $$
   Z_n=\widetilde\alpha_n^{-2}\left(\sum_{i=1}^Ju_i(  f_{\widetilde G_n}^{(X)}(x_i)-  f_{\widetilde G_{n-1}}^{(X)}(x_i))\right)^2.
   $$
We can write that
   \begin{align*}
   E_{n-1}(Z_n)
       &= E_{n-1}\left(
\widetilde\alpha_n^{-2}\left(\sum_{i=1}^Ju_i(  f_{\widetilde G_n}^{(X)}(x_i)-  f_{\widetilde G_{n-1}}^{(X)}(x_i))\right)^2
   \right)\\
   &=
   \int\left[\sum_{i=1}^Ju_i\int_\Theta k(x_i\mid\theta)\left(\widetilde G_{n-1}(\ddr\theta\mid y)-\widetilde G_{n-1}(\ddr\theta)\right)\right]^2f_{\widetilde G_{n-1}}^{(Y)}(y)\lambda(dy)\\
      &\stackrel{a.s.}{\rightarrow}\sum_{i,j=1}^Ju_iu_{j}\int (f_{\widetilde G}^{(X)}(x_i\mid y)-f_{\widetilde G}^{(X)}(x_i))(f_{\widetilde G}^{(X)}(x_{j}\mid y)-f_{\widetilde G}^{(X)}(x_{j}))f_{\widetilde G}^{(Y)}(y)\ddr y.
   \end{align*}
By assumption $b_n\uparrow+\infty$, and, defining $a_n=\widetilde\alpha_n^{-2}b_n^{-2}$, we obtain
\begin{align*}
   \sum_{k=1}^\infty\frac{E(Z_k^2)}{a_k^2b_k^2}&=\sum_{k=1}^\infty \widetilde\alpha_k^4b_k^2
  E\left(  
    \left[\sum_{i=1}^ju_i(f_{\widetilde G_{k-1}}^{(X)}(x_i\mid Y_k)-f_{\widetilde G_{k-1}}^{(X)}(x_i))\right]^4\right)\\
    &\leq
  \sum_{k=1}^\infty \widetilde\alpha_k^4b_k^2
     \left[\sum_{i=1}^ju_i\sup_{\theta\in\Theta} k(x_i\mid\theta)\right]^4  <+\infty.
\end{align*}
  Since, by assumption, $b_n\sum_{k\geq n}a_k^{-1}b_k^{-2}\rightarrow 1$, then by  Theorem \ref{th:app:almost sure2} (b),
  \begin{align*}
       b_n\sum_{k\geq n}\frac{Z_k}{a_kb_k^2} &=b_n\sum_{k\geq n}(\sum_{i=1}^Ju_i(f_{\widetilde G_k}^{(X)}(x_i)-f_{\widetilde G_{k-1}}^{(X)}(x_i)))^2\\
              &\stackrel{a.s.}{\rightarrow}
        \sum_{i,j=1}^Ju_iu_{j}\int (f_{\widetilde G}^{(X)}(x_i\mid y)-f_{\widetilde G}^{(X)}(x_i))(f_{\widetilde G}^{(X)}(x_{j}\mid y)-f_{\widetilde G}^{(X)}(x_{j}))f_{\widetilde G}^{(Y)}(y)\ddr y.
         \end{align*}
        It follows by Theorem \ref{th:app:almost sure conditional} that
        $
       M_{n,\infty}= b_n^{1/2}\sum_{i=1}^Ju_i(f_{\widetilde G_n}^{(X)}(x)-f_{\widetilde G}^{(X)}(x))
        $ converges in the sense of almost-sure conditional convergence to a  Gaussian kernel with variance $\sum_{i,j=1}^Ju_iu_{j}\int (f_{\widetilde G}^{(X)}(x_i\mid y)-f_{\widetilde G}^{(X)}(x_i))(f_{\widetilde G}^{(X)}(x_{j}\mid y)-f_{\widetilde G}^{(X)}(x_{j}))f_{\widetilde G}^{(Y)}(y)\ddr y$. 
  \end{proof}

      \begin{proof}[Proof of Theorem \ref{th:uniformclt}]
The proof of  Theorem \ref{th:uniformclt} is a direct consequence of Lemma \ref{lem:tight}, Lemma \ref{lem:findim} and Theorem 7.1 in  \cite
{billingsley1999}. 
\end{proof}

\subsubsection{Proof of Theorem \ref{th:sup}}\label{app:proofth4}
To prove the theorem, we will first provide a bound for $E(\sup_{x\in I} \mathbb G_{\widetilde G}(x))$ and then use concentration inequalities to bound $\mathbb G_{\widetilde G}(x)$ uniformly with respect to $x\in I$

\begin{lem}
\label{lem:expsup}
Let $\mathbb G_{\widetilde G}$ be a centered Gaussian process with random covariance function \eqref{eq:covariance}.
Then
    \begin{align*}
    E(\sup_{x\in I}\mathbb G_{\widetilde G}(x)\mid\widetilde G)&\leq 12\int_0^{\sigma(I)}\left(\log\left(1+\frac{\lambda(I)}{2\psi^{-1}(z/2)}\right)\right)^{1/2}dz,
\end{align*}
where $\sigma(I)$ and $\psi$ 
are defined as in  \eqref{eq:sigmaI} and \eqref{eq:psi}, respectively.
\end{lem}
\begin{proof}
    The proof is based on \citet[Theorem 3.18]{massart2007} (see Theorem \ref{th:supexp}). 
    Consider the pseudometric on $I$ induced by $R$
\begin{align*}
    d(x_1,x_2)&=(R(x_1,x_1)+R(x_2,x_2)-2R(x_1,x_2))^{1/2}\\
        &=\left(\int (f_{\widetilde G}^{(X)}(x_1\mid y)-f_{\widetilde G}^{(X)}(x_1)-f_{\widetilde G}^{(X)}(x_2\mid y)+f_{\widetilde G}^{(X)}(x_2))^2f_{\widetilde G}^{(Y)}(y)\lambda(dy)\right)^{1/2}\\
&=\left(\int (f_{\widetilde G}^{(X)}(x_1\mid y)-f_{\widetilde G}^{(X)}(x_2\mid y))^2f_{\widetilde G}^{(Y)}(y)\lambda(dy)-(f_{\widetilde G}^{(X)}(x_1)-f_{\widetilde G}^{(X)}(x_2))^2\right)^{1/2}.
\end{align*}
Let $\psi$ be defined as in \eqref{eq:psi} and let $\psi^{-1}(t)=\inf\{z:\psi(z)>t\}$. If $|x_1-x_2|\leq \psi^{-1}(\epsilon/2)$, then $d(x_1,x_2)\leq \epsilon/2$. Hence a uniform grid on I of step $2\psi^{-1}(\epsilon/2)$ allows to build a $\epsilon/2$-covering of $I$. It follows that
$$
N(\epsilon,I)\leq N'(\epsilon/2,I)\leq  \Bigl\lceil\frac{\lambda(I)}{2\psi^{-1}(\epsilon/2)}\Bigr\rceil \leq \frac{\lambda(I)}{2\psi^{-1}(\epsilon/2)}+1.
$$
Notice that 
\begin{align*}
    d(x_1,x_2)&
   \leq \left(\int (f_{\widetilde G}^{(X)}(x_1\mid y)-f_{\widetilde G}^{(X)}(x_2\mid y))^2f_{\widetilde G}^{(Y)}(y)\lambda(dy)\right)^{1/2}\\
    &\leq \left(\int \int_\Theta (k(x_1\mid \theta)-k(x_2\mid \theta))^2\widetilde G(\ddr\theta\mid y)f_{\widetilde G}^{(Y)}(y)\lambda(dy)\right)^{1/2}\\
    &\leq k'(I)|x_1-x_2|,
\end{align*}
with $k'(I)$ defined as 
\begin{equation*}
k'(I)=\sup_{x\in I,\theta\in\Theta}|k'(x\mid\theta)|,
\end{equation*}
where $k'(x\mid\theta)$ is the derivative of $k(x\mid\theta)$ with respect to $x$. 
Hence
$
\psi(z)\leq k'(I)z
$, and
$$\psi^{-1}(\epsilon/2)\geq 
\frac{\epsilon}{2k'(I)}.
$$
It follows that
$$
\sqrt{H(\epsilon,I)}=\sqrt{\log N(\epsilon,I)}\leq \sqrt{\log\left(1+\frac{\lambda(I)}{2\psi^{-1}(\epsilon/2)}\right)}
\leq 
\sqrt{
\frac{\lambda(I)k'(I)}{\epsilon}
},
$$
which implies that $\sqrt{H(\cdot,I)}$ is integrable in $0$. Moreover,
\begin{align*}
    E(\sup_{x\in I}\mathbb G_{\widetilde G}(x)\mid \widetilde G)&\leq 12\int_0^{\sigma(I)}\sqrt{H(z,I)} dz\\
    &\leq 12\int_0^{\sigma(I)}\sqrt{\log(N(z,I)} dz\\
    &\leq 12\int_0^{\sigma(I)}\left(\log\left(1+\frac{\lambda(I)}{2\psi^{-1}(z/2)}\right)\right)^{1/2}dz.
    \end{align*}
\end{proof}

\begin{proof}[Proof of Theorem \ref{th:sup}] The proof is based on Lemma \ref{lem:expsup}   and on the following concentration inequality (see Theorem \ref{th:supX}): for every $z>0$,
    \begin{align*}
        &P(\sup_{x\in I}\mathbb G_{\widetilde G}(x)\geq E(\sup_{x\in I} \mathbb G_{\widetilde G}(x)\mid \widetilde g)+\sigma(I)\sqrt{2z}\mid \widetilde G)
        \leq e^{-z},
    \end{align*}
where $\sigma(I)$ is defined as in \eqref{eq:sigmaI}.
Taking into account the symmetry of the distribution of $\mathbb G_{\widetilde G}$, we can write that
\begin{align*}
    &P((\sup_{x\in I}\mathbb G_{\widetilde G}(x)\geq E(\sup_{x\in I}\mathbb G_{\widetilde G}(x)\mid \widetilde G)+\sigma(I)\sqrt{2z})\\
    &\qquad\qquad\cup
(\inf_{x\in I}\mathbb G_{\widetilde G}(x)\leq -E(\sup_{x\in I}\mathbb G_{\widetilde G}(x)|\widetilde G)-\sigma(I)\sqrt{2z})|\widetilde G)\\
&\quad \quad \leq  2e^{-z}.
\end{align*}
Setting $2e^{-z}=\beta$, we obtain that
\begin{align*}
   & P(\sup_{x\in I}|\mathbb G_{\widetilde G}(x)|<12\int_0^{\sigma(I)}\left(\log\left(1+\frac{\lambda(I)}{2\psi^{-1}(z/2)}\right)\right)^{1/2}dz+\sigma(I)\sqrt{2|\log(\beta/2)|}\;|\widetilde G)\\
&\geq P(\sup_{x\in I}|\mathbb G_{\widetilde G}(x)|< E(\sup_{x\in I}\mathbb G_{\widetilde G}(x)|\widetilde G)+\sigma(I)\sqrt{2|\log(\beta/2)|}\;|\widetilde G)\\
&\geq 1-\beta.
\end{align*}
\end{proof}

\subsubsection{Proof of Theorem \ref{th:convv}}\label{app:proofth5}
The proof is split into several lemmas.

\begin{lem}
\label{lem:convsigma}
As $n\rightarrow+\infty$,  $\sigma_n(I)$ converges to $\sigma(I)$, $P$-a.s. 
\end{lem}
\begin{proof} The functions
$$
\int (f_{\widetilde G_n}^{(X)}(x|y)-f_{\widetilde G_n}^{(X)}(x))^2f_{\widetilde G_n}^{(Y)}(y)\lambda(dy)
=
\int f_{\widetilde G_n}^{(X)}(x|y)^2f_{\widetilde G_n}^{(Y)}(y)\lambda(\ddr y)-f_{\widetilde G_n}^{(X)}(x)^2
$$
are continuously differentiable in $I$, with
$$
\left|\frac{\ddr}{\ddr x}f_{\widetilde G_n}^{(X)}(x)^2\right|=2f_{\widetilde G_n}^{(X)}(x)\left|\int_\Theta k'(x\mid\theta)\widetilde G_n(\ddr\theta)\right|\leq 2\sup_{x,\theta}k(x\mid\theta)\sup_{x,\theta}|k'(x\mid\theta)|,
$$
and
\begin{align*}
    \left|\frac{\ddr }{\ddr x}\int f_{\widetilde G_n}^{(X)}(x|y)^2f_{\widetilde G_n}^{(Y)}(y)\lambda(\ddr y)\right|
&\leq 2 \left|\int \int_\Theta k'(x\mid\theta)\widetilde G_n(\ddr\theta\mid y) f_{\widetilde G}^{(X)}(x|y)f_{\widetilde G}^{(Y)}(y)\lambda(\ddr y)\right|\\
&\leq 2\sup_{x,\theta}k(x\mid\theta)\sup_{x,\theta}|k'(x\mid\theta)|.
\end{align*}
Thus, they have uniformly bounded derivative, and therefore are equicontinuous. On the other hand they converge pointwise to 
$$
\int (f_{\widetilde G}^{(X)}(x|y)-f_{\widetilde G}^{(X)}(x))^2f_{\widetilde G}^{(Y)}(y)\lambda(dy)
$$
$P$-a.s. 
By Ascoli-Arzel\`a theorem the convergence is uniform in $I$, and therefore
\begin{align*}
    \sup_{x\in I}&\int (f_{\widetilde G_n}^{(X)}(x|y)-f_{\widetilde G_n}^{(X)}(x))^2f_{\widetilde G_n}^{(Y)}(y)\lambda(dy)\\
    &\rightarrow\sup_{x\in I}\int (f_{\widetilde G}^{(X)}(x|y)-f_{\widetilde G}^{(X)}(x))^2f_{\widetilde G}^{(Y)}(y)\lambda(dy)
\end{align*}
    $P$-a.s. as $n\rightarrow+\infty$.
    \end{proof}

\begin{lem}\label{lem:convpsi}
    For every $z\in [0,\lambda(I)]$,  $\psi_n(z)$ converges to $\psi(z)$, $P$-a.s. as $n\rightarrow+\infty$,
  \end{lem}
  \begin{proof}
  The functions
  $$
  (x_1,x_2)\mapsto \int (f_{\widetilde{G_n}}^{(X)}(x_1\mid y)-f_{\widetilde{G_n}}^{(X)}(x_2\mid y))^2f_{\widetilde{G_n}}^{(Y)}(y)\lambda(\ddr y)-(f_{\widetilde{G_n}}^{(X)}(x_1)-f_{\widetilde{G_n}}^{(X)}(x_2))^2
  $$
  are continuously differentiable in $I\times I$, with
  \begin{align*}
      \left|\frac{\partial }{\partial x_1}(f_{\widetilde{G_n}}^{(X)}(x_1)-f_{\widetilde{G_n}}^{(X)}(x_2))^2
      \right|&\leq 2(  f_{\widetilde{G_n}}^{(X)}(x_1)+f_{\widetilde{G_n}}^{(X)}(x_2) )\int_\Theta| k'(x_1\mid\theta)|\widetilde G_n(\ddr\theta)\\&\leq 4\sup_{x,\theta}k(x\mid\theta)\sup_{x,\theta}|k'(x\mid\theta)|,
  \end{align*}
  and
\begin{align*}
      &\left|\frac{\partial }{\partial x_1}\int (f_{\widetilde{G_n}}^{(X)}(x_1\mid y)-f_{\widetilde{G_n}}^{(X)}(x_2\mid y))^2f_{\widetilde{G_n}}^{(Y)}(y)\lambda(\ddr y)
      \right|\\
      &\qquad\qquad\leq 2 \int (f_{\widetilde{G_n}}^{(X)}(x_1\mid y)+f_{\widetilde{G_n}}^{(X)}(x_2\mid y))\int |k'(x_1\mid\theta)|\widetilde G_n(\ddr\theta\mid y)f_{\widetilde{G_n}}^{(Y)}(y)\lambda(\ddr y)
      \\
      &\qquad\qquad\leq 4\sup_{x,\theta}k(x\mid\theta)\sup_{x,\theta}|k'(x\mid\theta)|,
  \end{align*}
  and analogously for the derivatives with respect to $x_2$.
  Therefore, they are
  equicontinuous in $I\times I$, and converge pointwise to 
  $$
 \int (f_{\widetilde{G}}^{(X)}(x_1\mid y)-f_{\widetilde{G}}^{(X)}(x_2\mid y))^2f_{\widetilde{G}}^{(Y)}(y)\lambda(\ddr y)-(f_{\widetilde{G}}^{(X)}(x_1)-f_{\widetilde{G}}^{(X)}(x_2))^2,
  $$
  $P$-a.s. By Ascoli-Arzel\`a Theorem, the convergence is uniform, and therefore 
  \begin{align*}
      &\sup_{\small{\begin{array}{c}x_1,x_2\in I,\\|x_1-x_2|<z\end{array}}}\!\!\!(\int (f_{\widetilde{G_n}}^{(X)}(x_1\mid y)-f_{\widetilde{G_n}}^{(X)}(x_2\mid y))^2f_{\widetilde{G_n}}^{(Y)}(y)\lambda(\ddr y)-(f_{\widetilde{G_n}}^{(X)}(x_1)-f_{\widetilde{G_n}}^{(X)}(x_2))^2)^{1/2}\\
       &=(\sup_{\small{\begin{array}{c}x_1,x_2\in I,\\|x_1-x_2|<z\end{array}}}\!\!\!\int (f_{\widetilde{G_n}}^{(X)}(x_1\mid y)-f_{\widetilde{G_n}}^{(X)}(x_2\mid y))^2f_{\widetilde{G_n}}^{(Y)}(y)\lambda(\ddr y)-(f_{\widetilde{G_n}}^{(X)}(x_1)-f_{\widetilde{G_n}}^{(X)}(x_2))^2)^{1/2}\\
      &\rightarrow(\sup_{\small{\begin{array}{c}x_1,x_2\in I,\\|x_1-x_2|<z\end{array}}}\!\!\!\int (f_{\widetilde{G}}^{(X)}(x_1\mid y)-f_{\widetilde{G}}^{(X)}(x_2\mid y))^2f_{\widetilde{G}}^{(Y)}(y)\lambda(\ddr y)-(f_{\widetilde{G}}^{(X)}(x_1)-f_{\widetilde{G}}^{(X)}(x_2))^2)^{1/2}\\
      &=\sup_{\small{\begin{array}{c}x_1,x_2\in I,\\|x_1-x_2|<z\end{array}}}\!\!\!(\int (f_{\widetilde{G}}^{(X)}(x_1\mid y)-f_{\widetilde{G}}^{(X)}(x_2\mid y))^2f_{\widetilde{G}}^{(Y)}(y)\lambda(\ddr y)-(f_{\widetilde{G}}^{(X)}(x_1)-f_{\widetilde{G}}^{(X)}(x_2))^2)^{1/2}. 
  \end{align*}
 \end{proof}

\begin{lem}\label{lem:convintegral}
    As $n\rightarrow+\infty$, 
    $$
    \int_0^{\sigma_n(I)}\left(\log\left(1+\frac{\lambda(I)}{2\psi_n^{-1}(z/2)}\right)\right)^{1/2}dz\rightarrow  \int_0^{\sigma(I)}\left(\log\left(1+\frac{\lambda(I)}{2\psi^{-1}(z/2)}\right)\right)^{1/2}dz
        $$
        $P$-a.s.
\end{lem}
\begin{proof}
    For $P$-a.s. every $\omega\in\Omega$, the functions $\psi_n(\cdot)(\omega)$ and $\psi(\cdot)(\omega)$ are  continuous and monotone non decreasing on the closed interval $[0,\lambda(I)]$. Hence, by Lemma \ref{lem:convpsi}, we can find a set $N$ with $P(N)=0$ such that for every $\omega\in N^c$, $\psi_n(\cdot)(\omega)$ converges to $\psi(\cdot)(\omega)$ in $C([0,\lambda(I)])$. In the rest of the proof we restrict to $N^c$. It can be proved as in Lemma \ref{lem:expsup} that
    $$
    \left(\log\left(1+\frac{\lambda(I)}{2\psi_n^{-1}(z/2)}\right)\right)^{1/2}
    \mbox{ and }
    \left(\log\left(1+\frac{\lambda(I)}{2\psi^{-1}(z/2)}\right)\right)^{1/2}
    $$
    are integrable in a right neighborhood of the origin. Moreover
    $$
    \int_0^{\sigma_n(I)}
    \left(\log\left(1+\frac{\lambda(I)}{2\psi_n^{-1}(z/2)}\right)\right)^{1/2}dz=2
    \int_0^{\psi_n^{-1}(\sigma_n(I)/2)}\left(\log\left(1+\frac{\lambda(I)}{2t}\right)\right)^{1/2}d\psi_n(t)
    $$
    and 
    $$
    \int_0^{\sigma(I)}
    \left(\log\left(1+\frac{\lambda(I)}{2\psi^{-1}(z/2)}\right)\right)^{1/2}dz=2
    \int_0^{\psi^{-1}(\sigma(I)/2)}\left(\log\left(1+\frac{\lambda(I)}{2t}\right)\right)^{1/2}d\psi(t),
    $$
    where $d\psi_n$ and $d\psi$ denote the measure associated to the monotone non-decrasing and continuous function $\psi_n$ and $\psi$, respectively.
    
    Let $u_1,u_2\in [0,\lambda(I)]$ be such $u_1\leq u_2$ and $\psi(u)=\sigma(I)/2$ for every $u\in [u_1,u_2]$. If $\underline u<u_1$ and $\overline u>u_2$, then $\psi(\underline u)<\sigma(I)/2$ and $\psi(\overline u)>\sigma(I)/2$, which implies that $\psi_n(\underline u)<\sigma(I)/2$ and $\psi_n(\overline u)>\sigma(I)/2$ for $n$ large enough. Since $\psi_n$ converges weakly to  $\psi$, then, for $n$ large enough,
    $$
    \int_0^{\underline u}\left(\log\left(1+\frac{\lambda(I)}{2t}\right)\right)^{1/2}d\psi(t)\leq  \int_0^{\psi_n^{-1}(\sigma_n(I)/2)}\left(\log\left(1+\frac{\lambda(I)}{2t}\right)\right)^{1/2}d\psi_n(t)
    $$
    and
     $$
    \int_0^{\overline u}\left(\log\left(1+\frac{\lambda(I)}{2t}\right)\right)^{1/2}d\psi(t)\geq  \int_0^{\psi_n^{-1}(\sigma_n(I)/2)}\left(\log\left(1+\frac{\lambda(I)}{2t}\right)\right)^{1/2}d\psi_n(t).
    $$
    Hence, for every $\underline u<u_1\leq u_2<\overline u$,
    \begin{align*}
        &\int_0^{\underline u}\left(\log\left(1+\frac{\lambda(I)}{2t}\right)\right)^{1/2}d\psi(t)\leq  \liminf_n\int_0^{\psi_n^{-1}(\sigma_n(I)/2)}\left(\log\left(1+\frac{\lambda(I)}{2t}\right)\right)^{1/2}d\psi_n(t)\\
        &\leq \limsup_n \int_0^{\psi_n^{-1}(\sigma_n(I)/2)}\left(\log\left(1+\frac{\lambda(I)}{2t}\right)\right)^{1/2}d\psi(t)
\leq \int_0^{\overline u}\left(\log\left(1+\frac{\lambda(I)}{2t}\right)\right)^{1/2}d\psi(t).
    \end{align*}
    Since the integrals are continuous with respect to the endpoints,
\begin{align*}
        &\int_0^{u_1}
        \left(\log\left(1+\frac{\lambda(I)}{2t}\right)\right)^{1/2}d\psi(t)\leq  \liminf_n\int_0^{\psi_n^{-1}(\sigma_n(I)/2)}\left(\log\left(1+\frac{\lambda(I)}{2t}\right)\right)^{1/2}d\psi_n(t)\\
        &\leq \limsup_n \int_0^{\psi_n^{-1}(\sigma_n(I)/2)}\left(\log\left(1+\frac{\lambda(I)}{2t}\right)\right)^{1/2}d\psi(t)
\leq \int_0^{u_2}\left(\log\left(1+\frac{\lambda(I)}{2t}\right)\right)^{1/2}d\psi(t).
    \end{align*}
Since $\psi$ is constant over the interval $[u_1,u_2]$, then the first and last terms of the above equation are equal. Moreover, by $\psi^{-1}$ and continuity of $\psi$, $u_2=\psi^{-1}(\sigma(I)/2)$.
    \end{proof}

\begin{proof}[Proof of Theorem \ref{th:convv}]
    The proof is an immediate consequence of Lemma \ref{lem:convsigma} and Lemma \ref{lem:convintegral}.
\end{proof}

\subsubsection{Proof of Equation \ref{th:confband}}\label{app:proofeq}
Since 
$s_n(I,\beta)$ is measurable with respect to $Y_{1:n}$, $s_{\widetilde G}(I,\beta)$ is measurable with respect to $\widetilde G$, and $s_n(I,\beta)$ converges to $s_{\widetilde G}(I,\beta)$ $P$-a.s., then, for every $t_1,t_2\in\mathbb R$,
\begin{align*}
&E(\exp(i(t_1s_n(I,\beta)+t_2 \sup_{x\in I}b_n^{1/2}|f_{\widetilde G_n}^{(X)}(x)-f_{\widetilde G}^{(X)}(x)|))\mid\mathcal G_n)\\
&\quad\stackrel{P-a.s.}{\longrightarrow} E(\exp(i(t_1 s_{\widetilde G}(I,\beta)+t_2 \sup_{x\in I}|\mathbb G_{\widetilde G}(x))|))|\widetilde G).
\end{align*}
In other words, $(s_n(I,\beta),\sup_{x\in I}b_n^{1/2}|f_{\widetilde G_n}^{(X)}(x)-f_{\widetilde G}^{(X)}(x)|)$ converges to $s_{\widetilde G}(I,\beta),\sup_{x\in I}|\mathbb G_{\widetilde G}(x)|)$ in the sense of almost-sure conditional convergence. By the 
Portmanteau theorem,
\begin{align*}
&\liminf_{n\rightarrow+\infty}P(\sup_{x\in I}|b_n^{1/2}(f_{\widetilde G_n}^{(X)}(x)-f_{\widetilde G}^{(X)}(x))|<\max(s_n(I,\beta),\epsilon)\mid\mathcal G_n)
\\
        &\quad\geq P(\sup_{x\in I}|\mathbb G_{\widetilde G(\omega)}|< \max(s_{\widetilde G}(I,\beta),\epsilon)\mid \widetilde G)\\
    &\quad\geq 1-\beta.
\end{align*}

\subsection{Proofs for Section \ref{sec5}}
\subsubsection{Proof of  Proposition \ref{teo_cons}}\label{app:cons}
\begin{proof}
        Under $P^*$, the random variables $Y_n$ are i.i.d. with density function $f_{G^*}^{(Y)}(y)=\int_\Theta \widetilde k(y\mid\theta)G^*(\ddr\theta)$, where $\widetilde k(y\mid\theta)$ identifies the mixing distribution by the assumption A4).
   The sequence $\widetilde G_n$ in \eqref{eq:newton} represents the recursive estimates of the true mixing distribution $G^*$ obtained via Newton's algorithm for the mixture model $\int_\Theta \widetilde k(y \mid \theta)\, G(\ddr\theta)$.  Let $\mathbb{F}$ denote the class of mixing distributions that are absolutely continuous with respect to $\mu$, and let $\overline{\mathbb{F}}$ be its weak closure. We first prove that $\widetilde G_n$ converges weakly to $G^*$, $P$-a.s.
   
      The proof is based on
    Corollary 4.7 of \cite{MarTok(09)}. The assumptions in \cite{MarTok(09)} differ slightly from A1)--A5), F1)--F3). In particular, \cite{MarTok(09)} requires that $\mathbb F$ is precompact, which holds in our setting since $\Theta$ is compact, and also that $\sum_n \left(\sum_{i=1}^n \widetilde \alpha_i\right) \widetilde \alpha_n^2 < +\infty$, but a careful inspection of the proofs reveals that the latter condition is  actually necessary only for the rate of convergence. 
    Thus, we can apply Corollary 4.7 in \cite{MarTok(09)}, and deduce that the sequence $\widetilde G_n$ converges weakly to the unique minimizer $G^\dag$ in $\overline{\mathbb{F}}$ of the Kullback–Leibler divergence between $f_{G^*}^{(X)}$ and $f_G^{(X)}$. Since, by assumption F1), $G^* \in \overline{\mathbb{F}}$, then  $G^\dag = G^*$, and thus $G_n$ converges weakly to $G^*$.

    We now prove that $f_{\widetilde G_n}^{(X)}$ converges to $f_{G^*}^{(X)}$ pointwise and in $L^1$.
    Since $k(x\mid\theta)$ is bounded and continuous in $\theta$  by the assumption A3), then, for every $x$, $f_{\widetilde G_n}^{(X)}(x)=\int_\Theta k(x\mid\theta)\widetilde G_n(\ddr\theta)$ converges to $f_{G^*}^{(X)}(x)=\int_\Theta k(x\mid\theta)G^*(\ddr\theta)$, $P^*$-a.s. Since $f_{\widetilde{G}_n}^{(X)}$ is dominated by $\sup_{\theta}k(x\mid\theta)$ which is integrable by the assumption A3), then $f_{\widetilde G_n}^{(X)}$
 converges to $f_{G^*}^{(X)}$ in $L^1$, $P$-a.s.
\end{proof}
\subsubsection{Proof of Corollary \ref{teo_merg}}\label{app:direct}
\begin{proof}
By Proposition~\ref{teo_cons},  we have
\[
\int \left| f^{(X)}_{\widetilde G_n}(x)-f^{(X)}_{G^*}(x)\right|\,\lambda(dx)
\rightarrow 0,
\qquad P^*\text{-a.s.}
\]
Moreover, by Corollary~4.6 in \cite{MarTok(09)}, applied to the direct mixture model with kernel $k$, the recursive estimate based on the latent variables $(X_n)$ satisfies
\[
\int \left| f^{(X)}_{G_n}(x)-f^{(X)}_{G^*}(x)\right|\,\lambda(dx)
\rightarrow 0,
\qquad P^*\text{-a.s.}
\]
The result follows immediately by the triangle inequality.
\end{proof}

\subsubsection{Proof of Proposition \ref{th:rateNewton}}\label{sec:proofrateNewtron}
For a function $f$, we denote by $f^{(j)}$ its $j$-th derivative, with $f^{(0)}:=f$.
For $\kappa>0$, let
\begin{equation}\label{eq:ell}
\ell=\max\{j\in\mathbb N_0:j<\kappa\}.
\end{equation}
A kernel $K$ is said to be of order $\ell+1$ if
\[
\int_{\mathbb R}K(z)\,dz=1,
\qquad
\int_{\mathbb R}z^jK(z)\,dz=0,
\quad
j=1,\ldots,\ell+1.
\]
For $h>0$, let
\[
K_h(x):=h^{-1}K(x/h)
\]
denote the corresponding rescaled kernel.

The proof of Proposition \ref{th:rateNewton} combines the inversion inequality of Theorem~3.1 in \cite{Rou(23)} with the convergence-rate result of Corollary~4.10 in \cite{MarTok(09)}. The following two lemmas verify the assumptions required to apply the inversion inequality.

\begin{lem}\label{lemma:rate1}
Let Assumption (i) of Proposition \ref{th:rateNewton} hold. Let $K$ be a kernel of order $\ell+1$ satisfying
$\int_{\mathbb{R}} |z|^{\kappa+1}|K(z)|\,dz<+\infty$.
Then there exists a constant $C$ such that, for every probability measure $G$ on $\Theta$,
\begin{equation}
||{F_G^{(X)}}-{F_G^{(X)}}*K_h||_1
\leq C h^{\kappa+1},
\qquad 0<h\leq 1,
\label{eq:condition-3.6}
\end{equation}
where $||\cdot ||_1$ denotes the $L^1(\mathbb R)$ norm, and $C$ is a constant that does not depend on $G$. Moreover, the derivative ${f_{G^*}^{(X)}}^{(\ell)}$ satisfies
\begin{equation}
\left|
{f_{G^*}^{(X)}}^{(\ell)}(x+\delta)
-
{f_{G^*}^{(X)}}^{(\ell)}(x)
\right|
\leq
L_0(x)|\delta|^{\kappa-\ell},
\qquad x,\delta\in\mathbb{R},
\label{eq:assumption-3.3}
\end{equation}
with $\ell$ as in \eqref{eq:ell} and $L_0=M_\kappa\in L^1(\mathbb{R})$.
\end{lem}

\begin{proof}
Differentiation under the integral sign gives
\[
{f_G^{(X)}}^{(\ell)}(x)
=
\int_{\Theta}\partial_x^\ell k(x\mid\theta)\,G(d\theta).
\]
Hence, by condition (i) of Proposition \ref{th:rateNewton},
\begin{equation}
    \label{eq:ineqdifff}
\left|
{f_G^{(X)}}^{(\ell)}(x+\delta)-{f_G^{(X)}}^{(\ell)}(x)
\right|
\leq
M_\kappa(x)|\delta|^{\kappa-\ell}.
\end{equation}
Integrating with respect to $x$ yields
\begin{equation}
\int_{\mathbb R}
\left|{f_G^{(X)}}^{(\ell)}(x+\delta)-{f_G^{(X)}}^{(\ell)}(x)\right|d x
\leq
\|M_\kappa\|_1|\delta|^{\kappa-\ell},
\label{eq:L1-holder}
\end{equation}
uniformly in $G$. 
Now, we consider
\begin{align*}
||F_G^{(X)}-F_G^{(X)}*K_h||_1
&=
\int_{\mathbb R}\left|F_G^{(X)}(x)-\int_{\mathbb R}F_G^{(X)}(x-u)K_h(u)\,du\right|dx
\\
&=\int_{\mathbb R}\left|\int_{\mathbb R}\left(F_G^{(X)}(x)-F_G^{(X)}(x-hz)\right)K(z)dz\right|dx,
\end{align*}
where the last equality comes from $\int_{\mathbb R} K(z)dz=1$. Moreover,
\begin{align*}
\int_{\mathbb R}F_G^{(X)}(x)K(z)dz
&=\int_{\mathbb R}\left(\sum_{j=0}^{\ell+1}\frac{(-hz)^j}{j!}{F_G^{(X)}}^{(j)}(x)\right)K(z)\,dz,
\end{align*}
where we have used $\int_{\mathbb R} z^jK(z)dz=0$ for $j=1,\dots,\ell+1$.

Now we consider $F_G^{(X)}(x-hz)$. Since $F_G^{(X)(\ell+1)}=f_G^{(X)(\ell)}$, Taylor's formula with integral remainder yields for every $u\in\mathbb R$
\begin{align}
F_G^{(X)}(x-u)-
\sum_{j=0}^{\ell+1}
\frac{(-u)^j}{j!}
F_G^{(X)(j)}(x)
&
=\frac{(-u)^{\ell+2}}{(\ell+1)!}
\int_0^1
(1-t)^{\ell+1}
F_G^{(X)(\ell+2)}(x-tu)\,dt
\nonumber\\
&=
\frac{(-u)^{\ell+2}}{(\ell+1)!}
\int_0^1
(1-t)^{\ell+1}
f_G^{(X)(\ell+1)}(x-tu)\,dt.
\label{eq:taylor}
\end{align}
Since $\frac{d}{dt}f_G^{(X)(\ell)}(x-tu)=-u\,f_G^{(X)(\ell+1)}(x-tu)$, then integration by parts gives
\begin{align*}
&
\frac{(-u)^{\ell+2}}{(\ell+1)!}
\int_0^1
(1-t)^{\ell+1}
f_G^{(X)(\ell+1)}(x-tu)\,dt
\\
&=
\frac{(-u)^{\ell+1}}{(\ell+1)!}
f_G^{(X)(\ell)}(x)
-
\frac{(-u)^{\ell+1}}{\ell!}
\int_0^1
(1-t)^\ell
f_G^{(X)(\ell)}(x-tu)\,dt
\\
&=
\frac{(-u)^{\ell+1}}{\ell!}
\int_0^1
(1-t)^\ell
\left(
f_G^{(X)(\ell)}(x-tu)
-
f_G^{(X)(\ell)}(x)
\right)dt.
\end{align*}
Substituting this expression into \eqref{eq:taylor} gives
\[
F_G^{(X)}(x-u)
-
\sum_{j=0}^{\ell+1}
\frac{(-u)^j}{j!}
F_G^{(X)(j)}(x)
=
\frac{(-u)^{\ell+1}}{\ell!}
\int_0^1
(1-t)^\ell
\left[
f_G^{(X)(\ell)}(x-tu)
-
f_G^{(X)(\ell)}(x)
\right]dt.
\]
Thus, by Assumption (i),
\[
\begin{aligned}
\|{F_G^{(X)}}-{F_G^{(X)}}*K_h\|_1
&\leq
\int_\mathbb R\left| \int_{\mathbb R}
\frac{(-hz)^{\ell+1}}{\ell!}
\int_0^1
(1-t)^\ell
\left[
f_G^{(X)(\ell)}(x-thz)
-
f_G^{(X)(\ell)}(x)
\right]dt
K(z)\,dz  \right|dx
\\
&\leq 
C_\kappa\|M_\kappa\|_1h^{\kappa+1}
\int_{\mathbb{R}}|z|^{\kappa+1}|K(z)|\,dz,
\end{aligned}
\]
which proves \eqref{eq:condition-3.6}.
To prove \eqref{eq:assumption-3.3}, take $G=G^*$ in \eqref{eq:ineqdifff}. This yields 
\[
\left|
{f_{G^*}^{(X)}}^{(\ell)}(x+\delta)-{f_{G^*}^{(X)}}^{(\ell)}(x)
\right|
\leq
M_\kappa(x)|\delta|^{\kappa-\ell}.
\]
Thus \eqref{eq:assumption-3.3} holds with
$L_0=M_\kappa\in L^1(\mathbb{R})$.
\end{proof}

\medskip

\begin{lem}\label{lemma:rate2}
Assume that $Z$ has Laplace distribution and that, for some
$a\in(0,1)$,
\begin{equation}
\sup_{\theta\in\Theta}
\int_{\mathbb R} e^{a|x|}k(x\mid\theta)\,dx
<+\infty.
\label{eq:uniform-exp-moment-kernel}
\end{equation}
If, for some sequence
$\widetilde\epsilon_n\to0$,
\[
\left\|f_{ \widetilde G_n}^{(Y)}-f_{G^*}^{(Y)}\right\|_1
\leq \widetilde\epsilon_n,
\]
then
\begin{equation}
W_1\!\left(
F_{\widetilde G_n}^{(Y)},
F_{G^*}^{(Y)}
\right)\leq C\, \widetilde \epsilon_n
\log\!\left(\frac{1}{\widetilde\epsilon_n}\right),
\label{eq:wasserstein-from-L1}
\end{equation}
for a constant $C$ not depending on ${\widetilde G_n}$.
\end{lem}

\begin{proof}
By equation \eqref{eq:uniform-exp-moment-kernel},
\[
\sup_G
\int_{\mathbb R} e^{a|x|}f_G^{(X)}(x)\,dx
\leq
\sup_{\theta\in\Theta}
\int_{\mathbb R} e^{a|x|}k(x\mid\theta)\,dx
<+\infty.
\]
Since $a<1$ and $Z$ has Laplace distribution, then $E([e^{a|Z|})<+\infty$. Therefore,
\[
\sup_G
\int_{\mathbb R}e^{a|y|}f_G^{(Y)}(y)\,dy
\leq
\sup_GE_G(e^{a|X|}) E(e^{a|Z|})<+\infty.
\]
Let
\[
C_0
:=
\sup_G
\int_{\mathbb R}
e^{a|y|}f_G^{(Y)}(y)\,dy
<+\infty.
\]
For every \(R>0\),
\begin{align*}
\sup_G
\int_{|y|>R}
|y|f_G^{(Y)}(y)\,dy
&=
\sup_G
\int_{|y|>R}(|y|e^{-a|y|})
e^{a|y|}f_G^{(Y)}(y)\,dy
\\
&\leq\sup_{y>R}(y e^{-ay})\sup_G\int_{|y|>R}e^{a|y|}f_G^{(Y)}(y)\,dy
\\
&\leq C_0\sup_{y>R}(y e^{-ay}).
\end{align*}
Defining $y=R+s$ with $s\geq 0$, we can write 
\begin{align*}
\sup_{y>R}(y e^{-ay})&=e^{-aR}\sup_{s>0}((R+s)e^{-as})
\\
&\leq
e^{-aR}
\left(
R+\sup_{s>0} s e^{-as}
\right)
\\
&=
e^{-aR}
\left(
R+\frac{1}{ae}
\right)
\\
&\leq
C_a(1+R)e^{-aR},
\end{align*}
where $C_a=\max(1,1/(ae))$, and  we have used
$\sup_{s>0}s e^{-as}=\frac{1}{ae}$. Thus,
\begin{align}
\sup_G
\int_{|y|>R}
|y|f_G^{(Y)}(y)\,dy
&\leq
C_0 C_a(1+R)e^{-aR}\nonumber
\\
&=
C_1(1+R)e^{-aR},\label{eq:inequalityR}
\end{align}
Since $W_1\!\left(F_G^{(Y)},
F_{G^*}^{(Y)}(y)
\right)=\int_{\mathbb R}\left|F_G^{(Y)}(y)-F_{G^*}^{(Y)}(y)\right|\,dy$, then there exists a constant $C_2$ such that 
\[
W_1\!\left(
F_G^{(Y)}(y),
F_{G^*}^{(Y)}
\right)
\leq C_2 
R\left\|f_G^{(Y)}-f_{G^*}^{(Y)}\right\|_1
+C_1
(1+R)e^{-aR},
\]
where we have used \eqref{eq:inequalityR} and
\begin{align*}
\int_{-R}^{R}
\left|
F_G^{(Y)}(y)-F_{G^*}^{(Y)}(y)
\right|\,dy
&\leq
2R\,
\sup_{y\in\mathbb R}
\left|
F_G^{(Y)}(y)-F_{G^*}^{(Y)}(y)
\right|
\\
&=
2R\,
\sup_{y\in\mathbb R}
\left|
\int_{-\infty}^{y}
\left(
f_G^{(Y)}(u)-f_{G^*}^{(Y)}(u)
\right)\,du
\right|
\\
&\leq
2R
\int_{\mathbb R}
\left|
f_G^{(Y)}(u)-f_{G^*}^{(Y)}(u)
\right|\,du
\\
&=
2R
\left\|
f_G^{(Y)}-f_{G^*}^{(Y)}
\right\|_1.
\end{align*}
Choose
\[
R=R_n
:=
\frac{1}{a}
\log\!\left(\frac{1}{\widetilde\epsilon_n}\right).
\]
Since \(\widetilde\epsilon_n\to0\), we have \(R_n>0\) for all sufficiently large $n$. Moreover,
\[
e^{-aR_n}
=
\exp\!\left(-\log\!\left(\frac{1}{\widetilde\epsilon_n}\right)\right)
=\widetilde\epsilon_n.
\]
Therefore,
\begin{align*}
W_1\!\left(F_G^{(Y)},F_{G^*}^{(Y)}\right)
&\leq C_2 R_n\left\|f_G^{(Y)}-f_{G^*}^{(Y)}\right\|_1+C_1(1+R_n)e^{-aR_n}
\\
&\leq C_2R_n\widetilde\epsilon_n + C_1(1+R_n)e^{-aR_n}
\\
&=\frac{C_2}{a}\widetilde\epsilon_n\log\!\left(\frac{1}{\widetilde\epsilon_n}\right)+C_1\left[1+\frac{1}{a}\log\!\left(\frac{1}{\widetilde\epsilon_n}\right)\right]\widetilde\epsilon_n
\\
&=C_1\widetilde\epsilon_n+\frac{C_1+C_2}{a}\widetilde\epsilon_n\log\!\left(\frac{1}{\widetilde\epsilon_n}\right).
\end{align*}
Since $\widetilde \epsilon_n\rightarrow 0$, then for $n$ sufficiently large, $\widetilde \epsilon_n\leq e^{-1}$. Thus,
\begin{align*}
W_1\!\left(
F_G^{(Y)},
F_{G^*}^{(Y)}
\right)
&\leq
\left(
C_1+\frac{C_1+C_2}{a}
\right)
\widetilde\epsilon_n
\log\!\left(\frac{1}{\widetilde\epsilon_n}\right)
\\
&\leq
C\,
\widetilde\epsilon_n
\log\!\left(\frac{1}{\widetilde\epsilon_n}\right),
\end{align*}
where $C$ is a constant that does not depend on $G$.
\end{proof}

\begin{proof}[Proof of Proposition \ref{th:rateNewton}]
By Corollary 4.10 in \cite{MarTok(09)},
\[
\rho(f_{\widetilde G_n}^{(Y)},f_{G^*}^{(Y)})=o_{P^*}\left(n^{-\frac{1-\gamma}{2}}\right),
\]
where $\rho$ denotes the Hellinger distance. Since the $L^1$ distance is dominated by the Hellinger distance, we can write 
\begin{equation}\label{eq:tildeepsilonn}
\left\|
f_{\widetilde G_n}^{(Y)}
-
f_{G^*}^{(Y)}
\right\|_1
=
O_{{P^*}}\!\left(\widetilde\epsilon_n\right).
\end{equation}
with $\widetilde \epsilon_n=o(n^{-\frac{1-\gamma}{2}})$.

By Lemma \ref{lemma:rate1}, the condition \eqref{eq:M} implies that $f_{G^*}^{(X)}$ satisfies Assumption~3.3 and Equation (3.6) of \cite{Rou(23)}. The Laplace distribution is ordinary smooth of order $\beta=2$.
Therefore, the inversion inequality in Theorem 3.1 of \cite{Rou(23)} yields, for every $h\in(0,1]$,
\begin{align}
W_1\!\left(F_{\widetilde G_n}^{(X)},F_{G^*}^{(X)}\right)\leq C_1 &\Bigl(h^{\kappa+1}
+W_1\!\left(F_{\widetilde G_n}^{(Y)},F_{G^*}^{(Y)}\right)
+
h^{-1}||f_{\widetilde G_n}^{(Y)}-f_{G^*}^{(Y)}||_1 \Bigr),
\label{eq:proof-inversion}
\end{align}
for a constant $C_1$.
By Lemma \ref{lemma:rate2}, conditions \eqref{eq:theorem-exp-moment} and \eqref{eq:tildeepsilonn} imply that
\begin{equation}
W_1\!\left(F_{\widetilde G_n}^{(Y)},F_{G^*}^{(Y)}\right)
\leq C\left(\widetilde\epsilon_n\log\frac{1}{\widetilde\epsilon_n}\right),
\label{eq:proof-direct-W1}
\end{equation}
for a constant $C$ not depending on $G_n$. Combining \eqref{eq:tildeepsilonn},
\eqref{eq:proof-direct-W1}, and \eqref{eq:proof-inversion}, we obtain for every $h\in(0,1]$
\[
W_1\!\left(F_{\widetilde G_n}^{(X)},F_{G^*}^{(X)}\right)
\leq C_2\left(h^{\kappa+1}+\widetilde\epsilon_n\log\frac{1}{\widetilde\epsilon_n}+h^{-1}\widetilde\epsilon_n\right),
\]
for a constant $C_2$.
Setting $h=h_n=\widetilde\epsilon_n^{1/(\kappa+2)}$, we get 
\[
h_n^{\kappa+1}
=
h_n^{-1}\widetilde\epsilon_n
=
\widetilde\epsilon_n^{\frac{\kappa+1}{\kappa+2}},
\]
so that
\[
W_1\!\left(F_{\widetilde G_n}^{(X)},F_{G^*}^{(X)}\right)
=
O_{{P^*}}\!\left(
\widetilde\epsilon_n^{\frac{\kappa+1}{\kappa+2}}
+
\widetilde\epsilon_n
\log\frac{1}{\widetilde\epsilon_n}
\right).
\]
Since
\[
\frac{
\widetilde\epsilon_n\log(1/\widetilde\epsilon_n)
}{
{\widetilde\epsilon_n}^{(\kappa+1)/(\kappa+2)}
}
=
\widetilde\epsilon_n^{1/(\kappa+2)}
\log\frac{1}{\widetilde\epsilon_n}
\longrightarrow0,
\]
as $n\rightarrow+\infty$, then  $\widetilde\epsilon_n\log(1/\widetilde\epsilon_n)$ is asymptotically negligible with respect to $\widetilde\epsilon_n^{(\kappa+1)/(\kappa+2)}$. It follows that
\[
W_1\!\left(F_{\widetilde G_n}^{(X)},F_{G^*}^{(X)}\right)
=
O_{{P^*}}\!\left(
\widetilde\epsilon_n^{\frac{\kappa+1}{\kappa+2}}
\right)=o_{{P^*}}\left(n^{-\frac{(1-\gamma)(\kappa+1)}{2(\kappa+2)}}\right).
\]
\end{proof}

\subsubsection{Proof of Corollary \ref{th:Newton-Bayes-merging}}\label{app:Newton-Bayes}
\begin{proof}
    Set
\[
r_n
:=
n^{-\frac{(1-\gamma)(\kappa+1)}
{2(\kappa+2)}}.
\]
By Proposition \ref{th:rateNewton},
\begin{equation}
W_1\!\left(
F_{\widetilde G_n}^{(X)},
F_{G^*}^{(X)}
\right)
=o_{P^*}(r_n).
\label{eq:Newton-part-merging}
\end{equation}
Under conditions \eqref{eq:true-density-tail-Bayes}--\eqref{eq:DP-base-tail}, Theorem~4.3 and Corollary~4.1 of
\cite{Rou(23)} imply that there exist constants $M_B>0$ and $q_B>0$ such that
\begin{equation}
W_1\!\left(
\widehat F_n^{\,B},
F_{G^*}^{(X)}
\right)
=
O_{P^*}\!\left(
n^{-1/5}(\log n)^{q_B}
\right).
\label{eq:Bayes-mean-rate}
\end{equation}
Since $\gamma\in(2/3,1)$ and $\kappa>0$,
\[
\frac{(1-\gamma)(\kappa+1)}
{2(\kappa+2)}
<
\frac{1}{6}
<
\frac{1}{5}.
\]
Consequently, for every fixed $q_B>0$,
\begin{align}
\frac{n^{-1/5}(\log n)^{q_B}}{r_n}
&=
n^{-\frac15+\frac{(1-\gamma)(\kappa+1)}{2(\kappa+2)}}(\log n)^{q_B}\rightarrow 0.
\label{eq:Bayes-negligible}
\end{align}
It follows from \eqref{eq:Bayes-mean-rate} and \eqref{eq:Bayes-negligible} that
\begin{equation}
W_1\!\left(\widehat F_n^{\,B},F_{G^*}^{(X)}\right)=o_{P^*}(r_n).
\label{eq:Bayes-part-merging}
\end{equation}
Finally, by the triangle inequality,
\begin{align*}
W_1\!\left(F_{\widetilde G_n}^{(X)},\widehat F_n^{\,B}\right)
&\leq W_1\!\left(F_{\widetilde G_n}^{(X)},F_{G^*}^{(X)}\right)+W_1\left(\widehat F_n^{\,B},F_{G^*}^{(X)}\right).
\end{align*}
Combining \eqref{eq:Newton-part-merging} and \eqref{eq:Bayes-part-merging} proves \eqref{eq:Newton-Bayes-merging}.
\end{proof}


\section{Other numerical experiments and applications}
\subsection{Synthetic-data analysis}\label{sec:synt}

\subsubsection{Unimodal example}\label{sec:supp-unimodal}

\paragraph{Synthetic-data generation}

For each $n\in\{250;\, 500;\, 750;\,1,000\}$ we generate random variables $X_1,\ldots,X_n$ i.i.d. according to the density $f^{(X)}(\cdot)=\phi(\cdot\mid 2,2)$, where $\phi(\cdot\mid\mu,\sigma^2)$ denotes the Gaussian density with mean $\mu$ and variance $\sigma^2$. Using the mixture representation \eqref{eq:mixture}, we write $f^{(X)}$ as $f_{G^*}^{(X)}$, with $G^*=\delta_{(2,2)}$ the true mixing distribution. We consider ordinary-smooth and super-smooth distributions for the noise random variables $Z_{i}$'s. In the ordinary-smooth case, $Z_i\stackrel{\mathrm{iid}}{\sim}\mathrm{Laplace}(0,b_l)$ with $b_l=\frac{\sigma_l}{\sqrt{2}}$ and $\sigma_l\in\{0.25,0.50\}$, namely $\sigma_l$ is the standard deviation of the Laplace noise; in the super-smooth case, $Z_i\stackrel{\mathrm{iid}}{\sim}N(0,\sigma_g^2)$ with standard deviation $\sigma_g\in\{0.25,0.50\}$. The  $Z_{i}$'s are independent of the $X_{i}$'s, and the observations are modeled as follows:
\begin{displaymath}
 Y_i=X_i+Z_i,\qquad i=1,\ldots,n.
\end{displaymath}
For each $n$, the same $Y_{1:n}=(Y_{1},\ldots,Y_{n})$ is used for all the methods under comparison in this section.

\paragraph{Quasi-Bayesian estimation and uncertainty quantification}\label{sec:supp-unimodal-newton}

To implement the quasi-Bayesian approach for estimating $f_{G}^{(X)}$, we consider Newton's algorithm~\eqref{eq:newton} with the Gaussian kernel $k(\cdot\mid\theta)=\phi(\cdot\mid\mu,\sigma^2)$, where $\theta=(\mu,\sigma^2)\in\mathbb{R}\times\mathbb{R}^{+}$. The parameter space $\mathbb{R}\times\mathbb{R}^{+}$ is restricted to $\Theta=[-20,20]\times[0.1,5]$, and it is discretized using increments of $0.1$ in both coordinates, giving $401\times50=20,050$ grid points. We use the learning rate $\widetilde\alpha_i=(1+i)^{-1}$, for $i\geq1$, and set $\widetilde G_0$ as the Uniform distribution on $\Theta$. All integrals with respect to the mixing distribution are evaluated by the two-dimensional trapezoidal quadrature, and the mixing distribution $\widetilde G_{n}$ is numerically renormalized after each update. The updates use the closed-form Gaussian--Gaussian or Gaussian--Laplace convolution kernel $\widetilde{k}(\cdot\mid\theta)=(f_{Z}\ast k(\cdot\mid\theta))(\cdot)$, depending on the choice of the noise distribution with density $f_{Z}$. The estimate of $f_{G}^{(X)}$ is computed on the grid $\mathcal X=\{-12,-11.9,\ldots,12\}$ as
\begin{displaymath}
 f_{\widetilde G_n}^{(X)}(x) =\int_\Theta \phi(x\mid\mu,\sigma^2)\widetilde G_n(\ddr\mu,\ddr\sigma^2),\qquad x\in\mathcal{X}.
\end{displaymath}
The discretization of the parameter space $\Theta$ is used only for numerical evaluation, through the implementation of the two-dimensional trapezoidal quadrature, and it imposes no modeling restrictions. Newton's algorithm depends on the order in which the observations $Y_{i}$'s are processed.  We therefore repeat the estimation procedure over $R_n=n/10$ random permutations of $Y_{1:n}$, which are generated using random seed $123$. If  $\widetilde G_n^{(r)}$ denotes the Newton's mixing-distribution estimate obtained from the $r$-th permutation, $r=1,\ldots,R_{n}$, then we form the order-averaged estimate
\begin{equation}\label{eq:unim-Gbar}
 \overline G_n = \frac{1}{R_n}\sum_{r=1}^{R_n}\widetilde G_n^{(r)},
\end{equation}
which is renormalized (numerically) after averaging. Therefore, based on \eqref{eq:predX}, the quasi-Bayes estimate of $f_{G}^{(X)}$ is
\begin{equation}\label{eq:est_unim}
f_{\overline  G_n}^{(X)}(x)= \int_{\Theta} \phi(x\mid\mu,\sigma^2) \overline  G_n(\ddr\mu,\ddr\sigma^2),\qquad x\in\mathcal X.
\end{equation}
The quasi-Bayes asymptotic credible intervals and bands for $f_G^{(X)}$ are also computed from the order-averaged estimate $\overline G_n$, based on \eqref{asym_intervals} and \eqref{asym_bands}, respectively. In particular, let $\overline  v_n(x)=v_n(x;\overline  G_n)$ denote the plug-in version of the variance in \eqref{eq:vx}, which is evaluated by replacing the unknown mixing distribution with the estimate $\overline  G_n$.  Since $\gamma=1$, we have $b_n=n$, and the lower and upper limits of the $95\%$ asymptotic credible interval are, at a given $x\in\mathcal X$,
\begin{displaymath}
L_{n,\mathrm{as}}(x)=\max\left\{0,\, f_{\overline  G_n}^{(X)}(x) -\frac{z_{0.975}}{\sqrt n}\left[\max\left\{\overline  v_n(x),10^{-8}\right\}\right]^{1/2}\right\}
\end{displaymath}
and
\begin{displaymath}
U_{n,\mathrm{as}}(x)=f_{\overline  G_n}^{(X)}(x)+\frac{z_{0.975}}{\sqrt n}\left[\max\left\{\overline  v_n(x),10^{-8}\right\}\right]^{1/2},
\end{displaymath}
respectively. From these lower and upper limits, the $95\%$ asymptotic credible interval \eqref{asym_intervals}, at a given $x\in\mathcal X$, is
\begin{equation}\label{eq:ci_unim}
 \left[
L_{n,\mathrm{as}}(x),
 U_{n,\mathrm{as}}(x)
 \right].
\end{equation}
Similarly, we construct the $95\%$ asymptotic credible band on $I=[-4,6]$. Let $\overline  s_n(I,0.05) =s_n(I,0.05;\overline  G_n)$ be the plug-in version of \eqref{eq_vn}, evaluated using $\overline  G_n$, and define $\overline \rho_n(I)=\frac{1}{\sqrt n}\max\{\overline  s_n(I,0.05),10^{-8}\}$. The lower and upper functions of the $95\%$ asymptotic credible band are, for every $x\in I$,
\begin{displaymath}
 L_{n,\mathrm{as}}^{I}(x)=\max\left\{0,\,f_{\overline  G_n}^{(X)}(x)-\overline \rho_n(I)\right\}
\end{displaymath}
and
\begin{displaymath}
U_{n,\mathrm{as}}^{I}(x)= f_{\overline  G_n}^{(X)}(x)+\overline \rho_n(I),
\end{displaymath}
respectively.  From these lower and upper functions, the $95\%$ asymptotic credible band  \eqref{asym_bands}, for every $x\in I$, is
\begin{equation}\label{eq:cb_unim}
\left\{h\in C(I): L_{n,\mathrm{as}}^{I}(x)\leq h(x)\leq U_{n,\mathrm{as}}^{I}(x)\text{ for every }x\in I\right\}.
\end{equation}
All integrals entering $\overline  v_n(x)$ and $\overline  s_n(I,0.05)$ are approximated by means of the trapezoidal quadrature. In particular, all integrals are evaluated on the same grid that is used for the estimate $f_{\overline  G_n}^{(X)}(x)$, whereas all quantities entering the credible bands are restricted to $\mathcal X_I=\mathcal X\cap I$.

Figures~\ref{fig:supp-unimodal-newton-laplace-small}--\ref{fig:supp-unimodal-newton-gaussian-small} report the quasi-Bayes estimates of $f_{G}^{(X)}$ in~\eqref{eq:est_unim}, and the asymptotic credible intervals and bands in~\eqref{eq:ci_unim}--\eqref{eq:cb_unim}. The figures consist of two columns, which correspond to the noise standard deviations $0.25$ and $0.50$, respectively.  The first row of each figure compares the true density $f_{G^*}^{(X)}$ with the quasi-Bayes density estimates, while the second row displays the credible intervals, and the third row displays the credible bands on the interval $I=[-4,6]$.

\begin{figure}[t]
 \centering
\includegraphics[
  width=1\textwidth,
  trim=20 15 20 0,
  clip
]{unimodal_laplace_small_n.png}
 \caption{\footnotesize{Unimodal example, with Laplace noise: estimates, credible intervals and bands.}}
 \label{fig:supp-unimodal-newton-laplace-small}
\end{figure}

\begin{figure}[t]
 \centering
 \includegraphics[
  width=1\textwidth,
  trim=20 15 20 0,
  clip
]{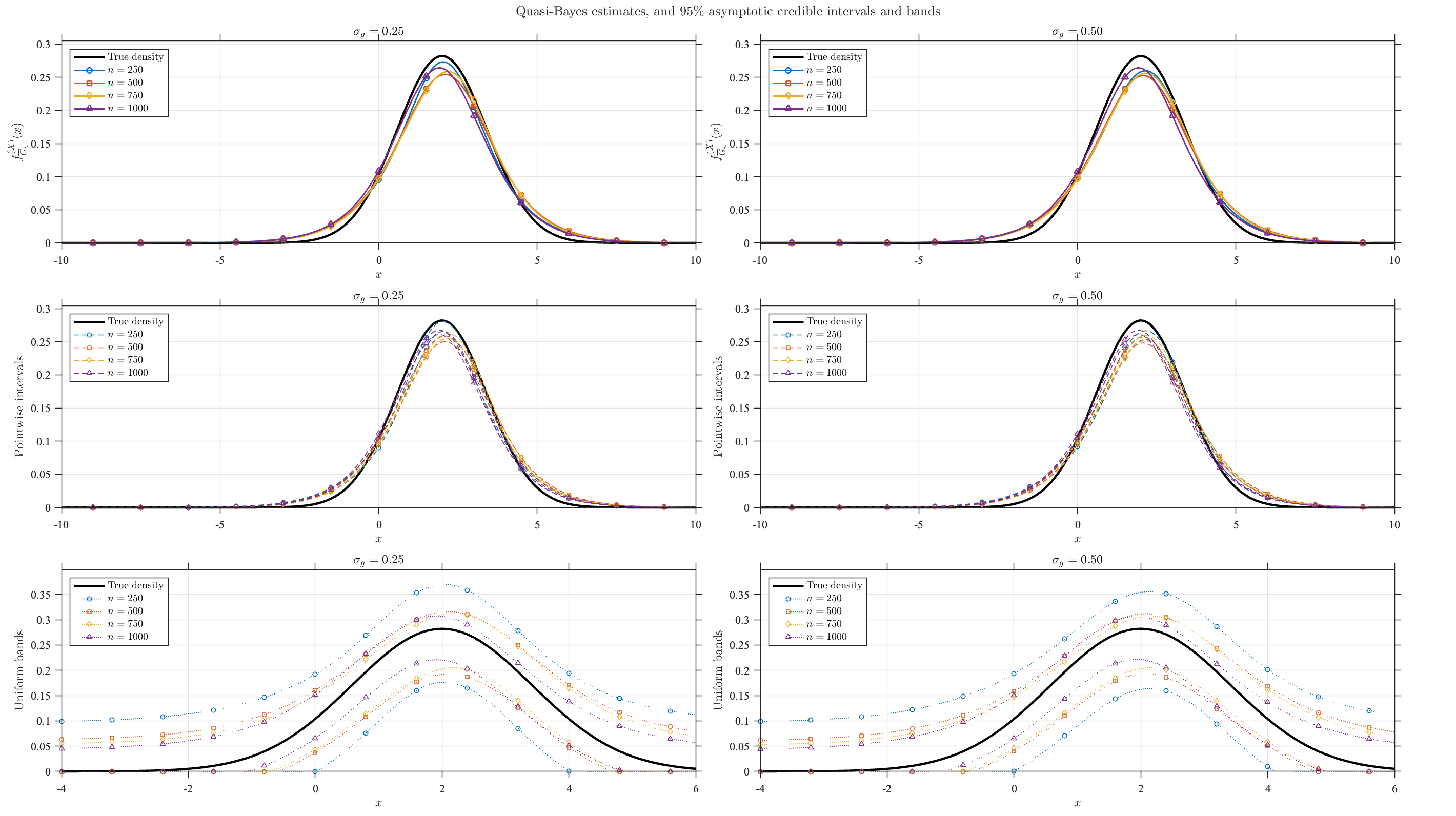}
 \caption{\footnotesize{Unimodal example, with Gaussian noise: estimates, credible intervals and bands.}}
 \label{fig:supp-unimodal-newton-gaussian-small}
\end{figure}

\paragraph{Monte Carlo uncertainty quantification}\label{sec:supp-unimodal-mc}

We approximate the conditional distribution of $f_{\widetilde G}^{(X)}$, given $Y_{1:n}$, by continuing the quasi-Bayesian learning process \eqref{eq:Y}--\eqref{eq:newton}, with the Gaussian kernel $k(\cdot\mid\theta)=\phi(\cdot\mid\mu,\sigma^2)$, where $\theta=(\mu,\sigma^2)\in\mathbb{R}\times\mathbb{R}^{+}$. The continuations start from the same order-averaged estimate $\overline  G_n$ in \eqref{eq:unim-Gbar}. For each $b=1,\ldots,B$, we simulate a continuation of length $N_{\mathrm{Future}}=10,000$. More precisely, at the continuation step $s$,  we draw
\begin{displaymath}
 \theta_{n+s}^{(b)}\sim \widetilde G_{n+s-1}^{(b)},\qquad X_{n+s}^{(b)}\mid\theta_{n+s}^{(b)} \sim k(\cdot\mid\theta_{n+s}^{(b)}),
\end{displaymath}
generate $Z_{n+s}^{(b)}$ from the noise distribution under consideration, set $Y_{n+s}^{(b)}=X_{n+s}^{(b)}+Z_{n+s}^{(b)}$, and update $\widetilde G_{n+s-1}^{(b)}$ using~\eqref{eq:newton} with learning rate $(1+n+s)^{-1}$.  The resulring $\widetilde G_{n+N_{\mathrm{Future}}}^{(b)}$ is then used to compute
\begin{displaymath}
 f_b(x):=\int_\Theta \phi(x\mid \theta)\widetilde G_{n+N_{\mathrm{Future}}}^{(b)}(\ddr\theta), \qquad x\in\mathcal X,
\end{displaymath}
with the integral being evaluated numerically using trapezoidal quadrature. Therefore, the $b$-th Monte Carlo replication produces the density $f_b(x)$, for $x\in\mathcal X$. We make use of $B=1000$ Monte Carlo replications for both the noise distributions, generated using random seed $123$.

According to the above Monte Carlo sampling scheme, a Monte Carlo estimate of $f_{G}^{(X)}$ is
\begin{equation}\label{eq:unimodal-mc-est}
f^{(X)}_{\mathrm{MC},n}(x)= \frac{1}{B}\sum_{b=1}^{B}f_b(x), \qquad x\in\mathcal X.
\end{equation}
For $\tau\in(0,1)$, and for a given $x\in\mathcal{X}$, we define $q^{\mathrm{MC}}_{\tau,n}(x)=\operatorname{Quantile}_{\tau}\left\{f_b(x):b=1,\ldots,B\right\}$, namely the empirical $\tau$-quantile of $f_1(x),\ldots,f_B(x)$. The $95\%$ Monte Carlo credible interval at a given $x\in\mathcal{X}$ is
\begin{equation}\label{eq:supp-unimodal-mc-interval}
 \left[q^{\mathrm{MC}}_{0.025,n}(x),q^{\mathrm{MC}}_{0.975,n}(x)\right].
\end{equation}
Similarly, to construct the Monte Carlo credible band on $I=[-4,6]$, for every $x\in \mathcal X_I$ let $\overline  s_{\mathrm{MC},n}(x)$ denote the empirical standard deviation of
$f_1(x),\ldots,f_B(x)$ and define $T_b= \max_{x\in\mathcal X_I}\frac{ \left|f_b(x)-f^{(X)}_{\mathrm{MC},n}(x)\right| }{\max\left\{\overline  s_{\mathrm{MC},n}(x),10^{-10}\right\}
 }$, for $b=1,\ldots,B$. Let $c_{0.95}$ be the empirical $0.95$-quantile of $T_1,\ldots,T_B$.  The lower and upper functions of the $95\%$ Monte Carlo
credible band are, for every $x\in I$, 
\begin{displaymath}
 L_{\mathrm{MC},n}^{I}(x) =\max\left\{0,\,f^{(X)}_{\mathrm{MC},n}(x)-c_{0.95}\overline  s_{\mathrm{MC},n}(x)\right\}
\end{displaymath}
and
\begin{displaymath}
 U_{\mathrm{MC},n}^{I}(x)=f^{(X)}_{\mathrm{MC},n}(x) +c_{0.95}\overline  s_{\mathrm{MC},n}(x),
\end{displaymath}
respectively.  From these lower and upper functions, the $95\%$ Monte Carlo credible band, for every $x\in I$ is
\begin{equation}\label{eq:supp-unimodal-mc-band}
 \left[ L_{\mathrm{MC},n}^{I}(x), U_{\mathrm{MC},n}^{I}(x)\right].
\end{equation}
Figures~\ref{fig:supp-unimodal-mc-laplace-025}--\ref{fig:supp-unimodal-mc-gaussian-050} display Monte Carlo credible intervals and bands in~\eqref{eq:supp-unimodal-mc-interval}--\eqref{eq:supp-unimodal-mc-band}, which are compared with the corresponding quasi-Bayes asymptotic credible intervals and bands in~\eqref{eq:ci_unim}--\eqref{eq:cb_unim}. The quasi-Bayesian credible intervals closely agree with their Monte Carlo counterparts, whereas the quasi-Bayesian credible bands are wider and thus more conservative, due to the metric-entropy and Gaussian concentration bounds used in their construction.

\begin{figure}[t]
 \centering
 \includegraphics[
  width=0.9\textwidth,
  trim=20 15 20 0,
  clip
]{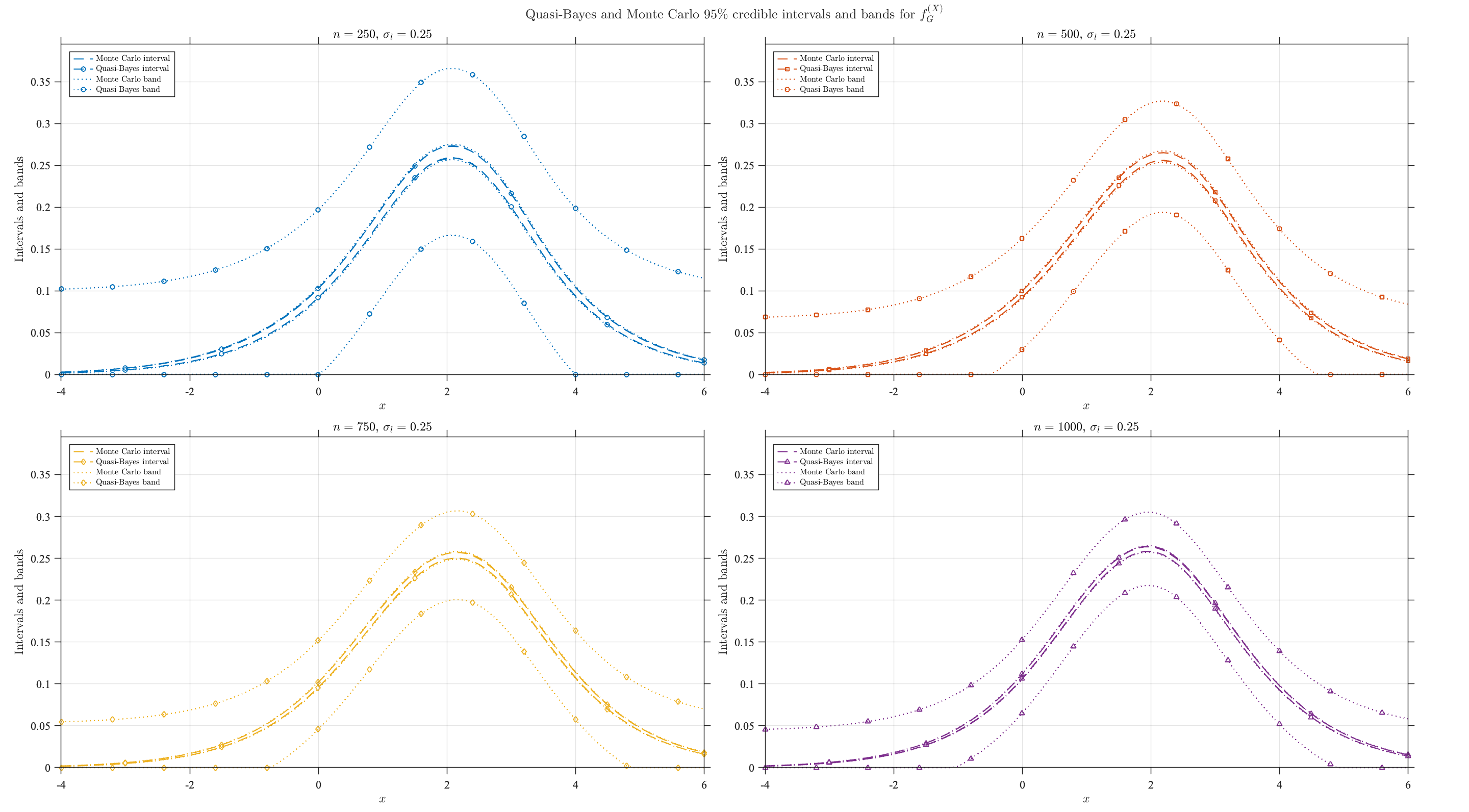}
 \caption{\footnotesize{Unimodal example, with Laplace noise ($\sigma_{l}=0.25$): quasi-Bayes versus Monte Carlo credible intervals and bands.}}
 \label{fig:supp-unimodal-mc-laplace-025}
\end{figure}

\begin{figure}[t]
 \centering
  \includegraphics[
  width=0.9\textwidth,
  trim=20 15 20 0,
  clip
]{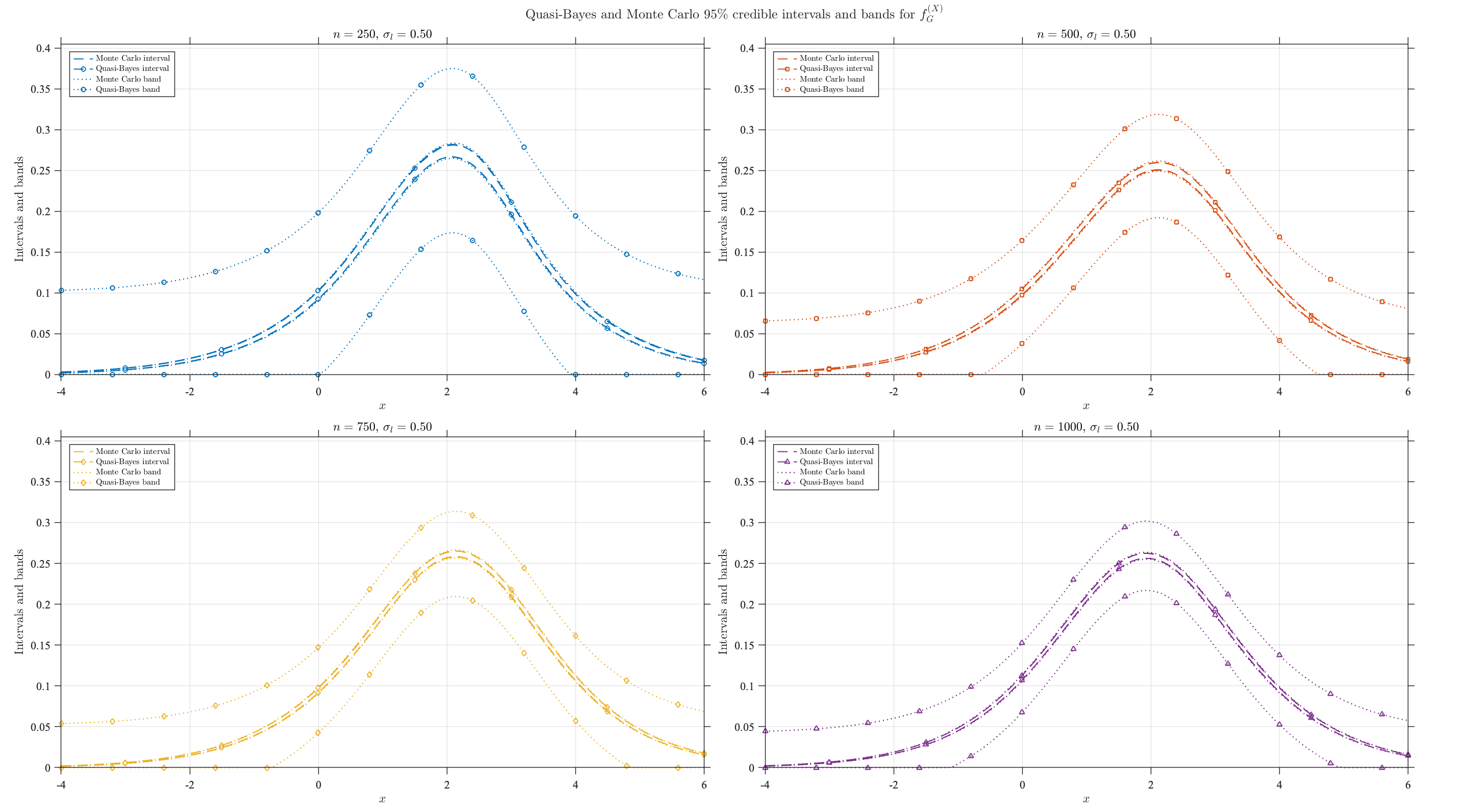}
 \caption{\footnotesize{Unimodal example, with Laplace noise ($\sigma_{l}=0.50$): quasi-Bayes versus Monte Carlo credible intervals and bands.}}
 \label{fig:supp-unimodal-mc-laplace-050}
\end{figure}

\begin{figure}[t]
 \centering
 \includegraphics[
  width=0.9\textwidth,
  trim=20 15 20 0,
  clip
]{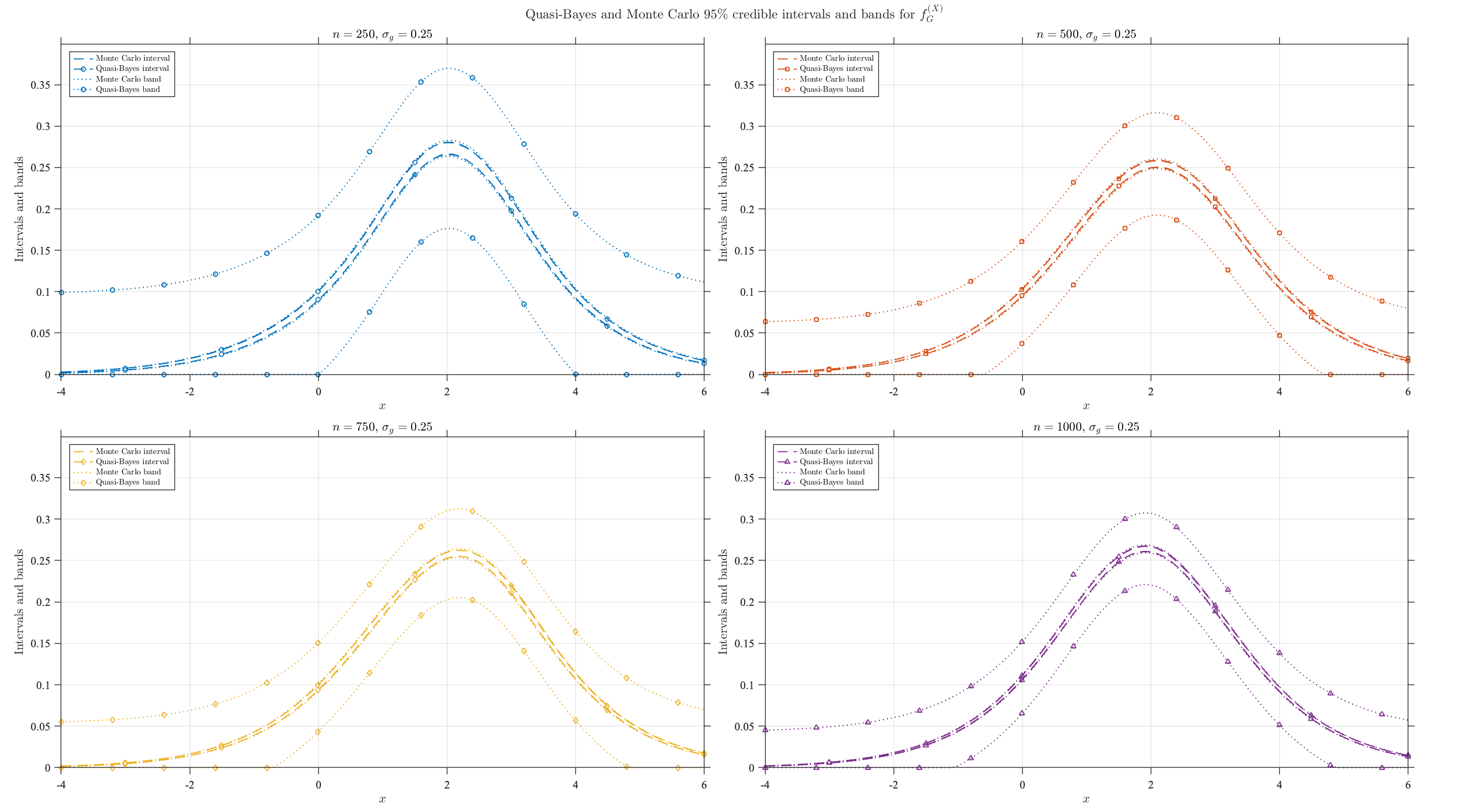}
 \caption{\footnotesize{Unimodal example, with Gaussian noise ($\sigma_{g}=0.25$): quasi-Bayes versus Monte Carlo credible intervals and bands.}}
 \label{fig:supp-unimodal-mc-gaussian-025}
\end{figure}

\begin{figure}[t]
 \centering
 \includegraphics[
  width=0.9\textwidth,
  trim=20 15 20 0,
  clip
]{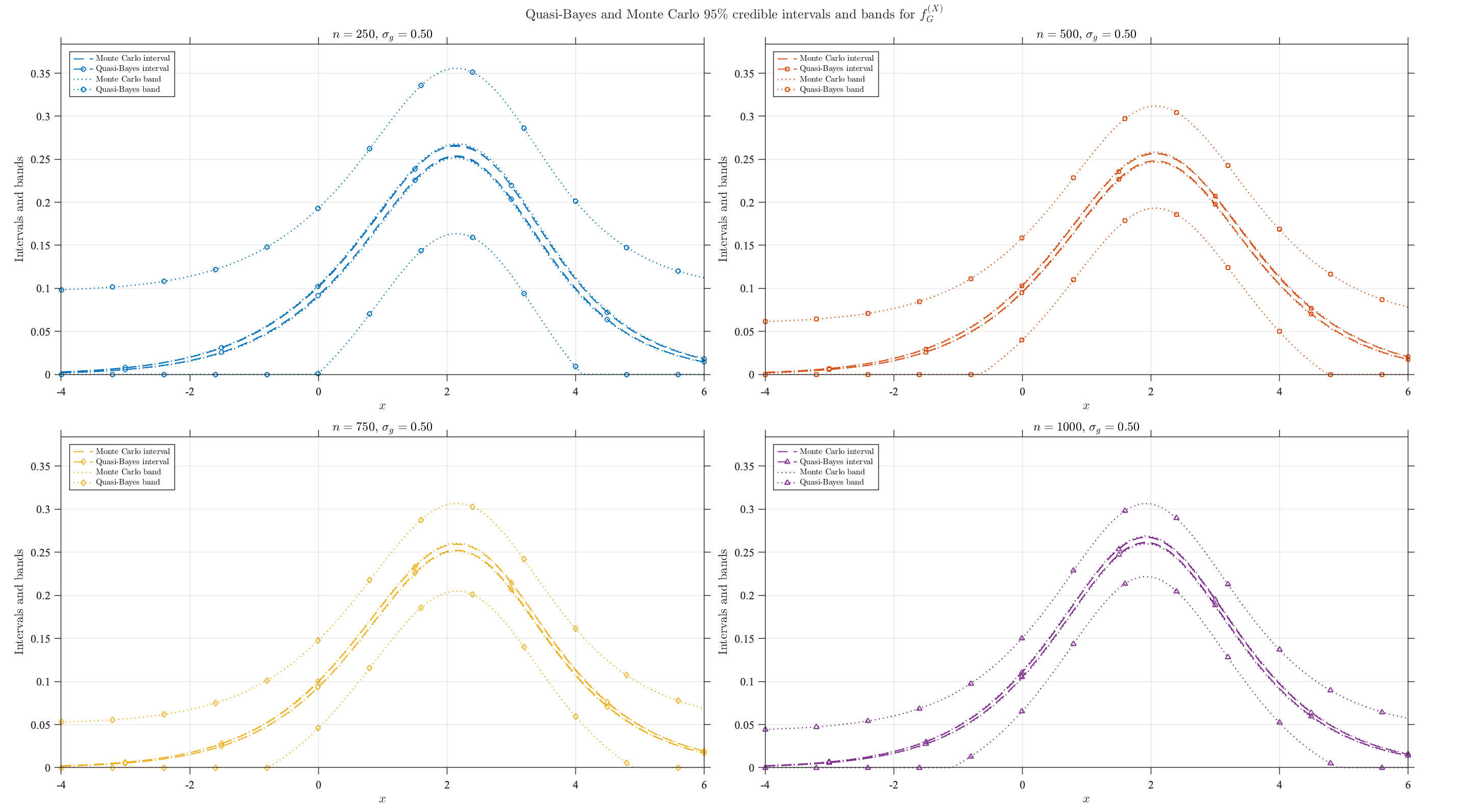}
 \caption{\footnotesize{Unimodal example, with Gaussian noise ($\sigma_{g}=0.50$): quasi-Bayes versus Monte Carlo credible intervals and bands.}}
 \label{fig:supp-unimodal-mc-gaussian-050}
\end{figure}

\paragraph{Dirichlet process mixture model benchmarks}\label{sec:supp-unimodal-bayes}

We compare the quasi-Bayesian approach with a Bayesian approach in which the mixing distribution $G$ is endowed with a Dirichlet process prior \citep{Fer(73),Lo(84)}. Consider the Gaussian kernel $k(\cdot\mid\theta)=\phi(\cdot\mid\mu,\sigma^2)$, where $\theta=(\mu,\sigma^2)\in\mathbb{R}\times\mathbb{R}^{+}$, and assume the hierarchical model
\begin{equation}\label{eq:supp-unimodal-dp-model}
 \begin{split}
 Y_i&=X_i+Z_i,\qquad i=1,\ldots,n\\[-0.2cm]
 X_i\mid\theta_i&\sim k(\cdot\mid\theta_i),\\[-0.2cm]
  \theta_i\mid G&\stackrel{\mathrm{iid}}{\sim}G,\\[-0.2cm]
  G\mid M&\sim\mathrm{DP}(M,H),\\[-0.2cm]
  M&\sim\mathrm{Gamma}(1,1),
  \end{split}
\end{equation}
where $\text{Gamma}(\cdot,\,\cdot)$ is the Gamma distribution, and $DP(\cdot,\cdot)$ is the law of a Dirichlet process with strength parameter $M>0$ and non-atomic (diffuse) base probability measure $H$ \citep{Fer(73)}. We set $H$ to be the Uniform distribution on $\Theta=[-20,20]\times[0.1,5]$.  Inference is based on the closed-form Gaussian--Gaussian or Gaussian--Laplace convolution  kernel $\widetilde{k}(\cdot\mid\theta_{i})=(f_{Z}\ast k(\cdot\mid\theta_{i}))(\cdot)$, depending on the choice of the noise distribution with density $f_{Z}$.

For the Dirichlet process Gaussian-mixture model \eqref{eq:supp-unimodal-dp-model}, we perform sequential posterior inference using the sequential Monte Carlo (SMC) algorithm of \citet{Fea(04)}; see also \citet[Algorithm 1]{Gri(17)} and references therein. The algorithm uses $P=500$ particles. At time $i$, the $p$-th particle is characterized by $K_{i-1}^{(p)}$ occupied atoms, their corresponding counts, the strength parameter $M_{i-1}^{(p)}$, and the normalized weight $W_{i-1}^{(p)}$. The particle first draws a candidate atom $\theta_{\mathrm{new},i}^{(p)}\sim H$, such that the allocation weights for the occupied components are $a_{i,j}^{(p)} =n_{j,i-1}^{(p)} \widetilde k(Y_i\mid\theta_{j,i-1}^{(p)})$, for $j=1,\ldots,K_{i-1}^{(p)}$, and the allocation weight for the candidate is $a_{i,0}^{(p)}=M_{i-1}^{(p)}\widetilde k(Y_i\mid\theta_{\mathrm{new},i}^{(p)})$. The observation is assigned to one of these components with probabilities proportional to the corresponding weights. Selecting the proposed atom creates a new occupied component with count one. Letting $A_i^{(p)}=a_{i,0}^{(p)} +\sum_{j=1}^{K_{i-1}^{(p)}}a_{i,j}^{(p)}$ the normalizing constant, the incremental importance-weight update is $\widetilde W_i^{(p)}=W_{i-1}^{(p)}\frac{A_i^{(p)}}{M_{i-1}^{(p)}+i-1}$ after which the particle weights are normalized. Systematic resampling is performed whenever $\operatorname{ESS}_i = \{\sum_{p=1}^{P}(W_i^{(p)})^2\}^{-1}<250$. Following resampling, the weights are reset to $1/P$ and $M_i^{(p)}$ is updated once in every particle using the Escobar--West step \citep{Esc(95)}. One additional Escobar--West step is performed at the end of the sequential pass. For the $p$-th terminal particle, the posterior estimate of the mixing distribution $G$ is
\begin{displaymath}
 \widehat G_n^{(p)}(\cdot)=\frac{M_n^{(p)}H(\cdot) + \sum_{j=1}^{K_n^{(p)}}n_{j,n}^{(p)}\delta_{\theta_{j,n}^{(p)}}(\cdot)}{M_n^{(p)}+n},
\end{displaymath}
such that, by letting $m_H(x)=\int_\Theta k(x\mid\theta)H(\ddr\theta)$, the induced posterior estimate of $f_{G}^{(X)}$ is given by
\begin{displaymath}
f_{p,n}^{(X)}(x)=\int_{\Theta} k(x\mid\theta) \widehat G_n^{(p)}(\ddr\theta)\\=\frac{M_n^{(p)}m_H(x)+ \sum_{j=1}^{K_n^{(p)}}n_{j,n}^{(p)}
k(x\mid\theta_{k,n}^{(p)})}{M_n^{(p)}+n}.
\end{displaymath}
Therefore, the reported estimate of $f_G^{(X)}$ is the average $f_{\mathrm{SMC},n}^{(X)}(x)=\sum_{p=1}^{P}W_n^{(p)} f_{p,n}^{(X)}(x)$, for $x\in\mathcal X$.

For the hierarchical model \eqref{eq:supp-unimodal-dp-model}, we also perform Markov chain Monte Carlo (MCMC) posterior inference using the batch implementation of Algorithm~8 in \citet{Nea(00)}. At each MCMC sweep, the observations $Y_{1:n}$ are grouped into $K$ currently occupied mixture components, with distinct atoms $\theta_1,\ldots,\theta_K$. Use $c_i=j$ to indicate that $Y_i$ is assigned to component $j$, whose size is $n_j=\sum_{i=1}^{n}I(c_i=j)$. To update $c_i$, the $i$-th observation is temporarily removed and reassigned either to an occupied component, with weight $n_{-i,j}\widetilde k(Y_i\mid\theta_j)$, or to one of $m=2$ auxiliary components, each with weight $\frac{M}{2}\widetilde k(Y_i\mid\theta)$. If removing $i$ eliminates a singleton component, its former atom is used as one of the two auxiliary atoms; otherwise, both auxiliary atoms are drawn from $H$. After updating all allocations, each occupied atom is updated by one random-walk Metropolis--Hastings step, using proposal standard deviations $0.20$ for the location and $0.10$ for the variance; proposals outside $\Theta$ are rejected. Finally, the strength parameter $M$ is updated by one Escobar--West step. We run $5{,}000$ MCMC sweeps, discard the first $2{,}500$, and retain every tenth subsequent sweep, obtaining $S=250$ states. For the $s$-th retained state, the posterior estimate of the mixing distribution $G$ is
\begin{displaymath}
 \widehat G^{(s)}_{n}(\cdot)= \frac{M^{(s)}H(\cdot)+ \sum_{j=1}^{K^{(s)}}n_j^{(s)} \delta_{\theta_j^{(s)}}(\cdot)}{M^{(s)}+n},
 \end{displaymath}
such that, by letting $m_H(x)=\int_\Theta k(x\mid\theta)H(\ddr\theta)$, the induced posterior estimate of $f_{G}^{(X)}$ is given by
 \begin{displaymath}
f_{s,n}^{(X)}(x)=\int_{\Theta}k(x\mid\theta)\widehat G^{(s)}_{n}(\ddr\theta)=\frac{M^{(s)}m_H(x)+ \sum_{j=1}^{K^{(s)}}n_j^{(s)} k(x\mid\theta_j^{(s)})}{M^{(s)}+n}.
 \end{displaymath}
Therefore, the reported estimate of $f_{G}^{(X)}$ is the average $f_{\mathrm{MCMC},n}^{(X)}(x)=\frac{1}{S}\sum_{s=1}^{S}f_{s,n}^{(X)}(x)$, for $x\in\mathcal{X}$.

The comparison between the quasi-Bayesian and Bayesian approaches is performed for sample size $n\in\{250,500,750,1000\}$, separately for Laplace and Gaussian noise distributions with $\sigma_{l}=\sigma_{g}=0.25$. To account for order dependence in the observations $Y_{i}$'s, results obtained with Newton's and SMC algorithms are averaged over the same $R_n=n/10$ permutations.  Instead Algorithm~8 is run once on each full dataset, as it is designed to work in batch.

In addition to quasi-Bayes and Bayes estimates, we record total computational work and CPU time.  A unit of work is one evaluation of the convolution kernel.  Newton's algorithm evaluates the kernel at all $20,050$ grid points for every incoming observation, and therefore has exactly $20,050n$ work units per pass and constant work per update. For the SMC algorithm, the work required by the $i$-th observation is $\sum_{p=1}^{500}\bigl(K_{i-1}^{(p)}+1\bigr)$, where $K_{i-1}^{(p)}$ is the number of occupied components in particle $p$ before processing the $i$-th observation.  For Algorithm~8, the work count includes the convolved-likelihood evaluations used in both the allocation updates and the Metropolis--Hastings updates, accumulated over all $5,000$ sweeps.  For Newton's and SMC algorithms the plotted total work and CPU time correspond to one sequential pass, which is averaged over all the permutations considered, rather than to the aggregate cost of computing the order-averaged estimate; for Algorithm~8 they correspond to the complete batch run of the algorithm. CPU time is measured around the updating operations, excluding file input/output and figure creation. Figures~\ref{fig:supp-unimodal-bayes-laplace}--\ref{fig:supp-unimodal-bayes-gaussian} summarize the comparison between the quasi-Bayesian and Bayesian approaches. The first row of these figures displays the estimates of $f_{G}^{(X)}$ obtained with Newton's algorithm, the SMC algorithm, and Algorithm~8; the second and third rows report total work and total CPU time, respectively.

\begin{figure}[t]
 \centering
 \includegraphics[
  width=1\textwidth,
  trim=20 15 20 15,
  clip
]{unimodal_bayes_comparison_laplace.png}
 \caption{\footnotesize{Unimodal example, with Laplace noise ($\sigma_l=0.25$): comparison between  quasi-Bayesian and Bayesian approaches.}}
 \label{fig:supp-unimodal-bayes-laplace}
\end{figure}

\begin{figure}[t]
 \centering
  \includegraphics[
  width=1\textwidth,
  trim=20 15 20 15,
  clip
]{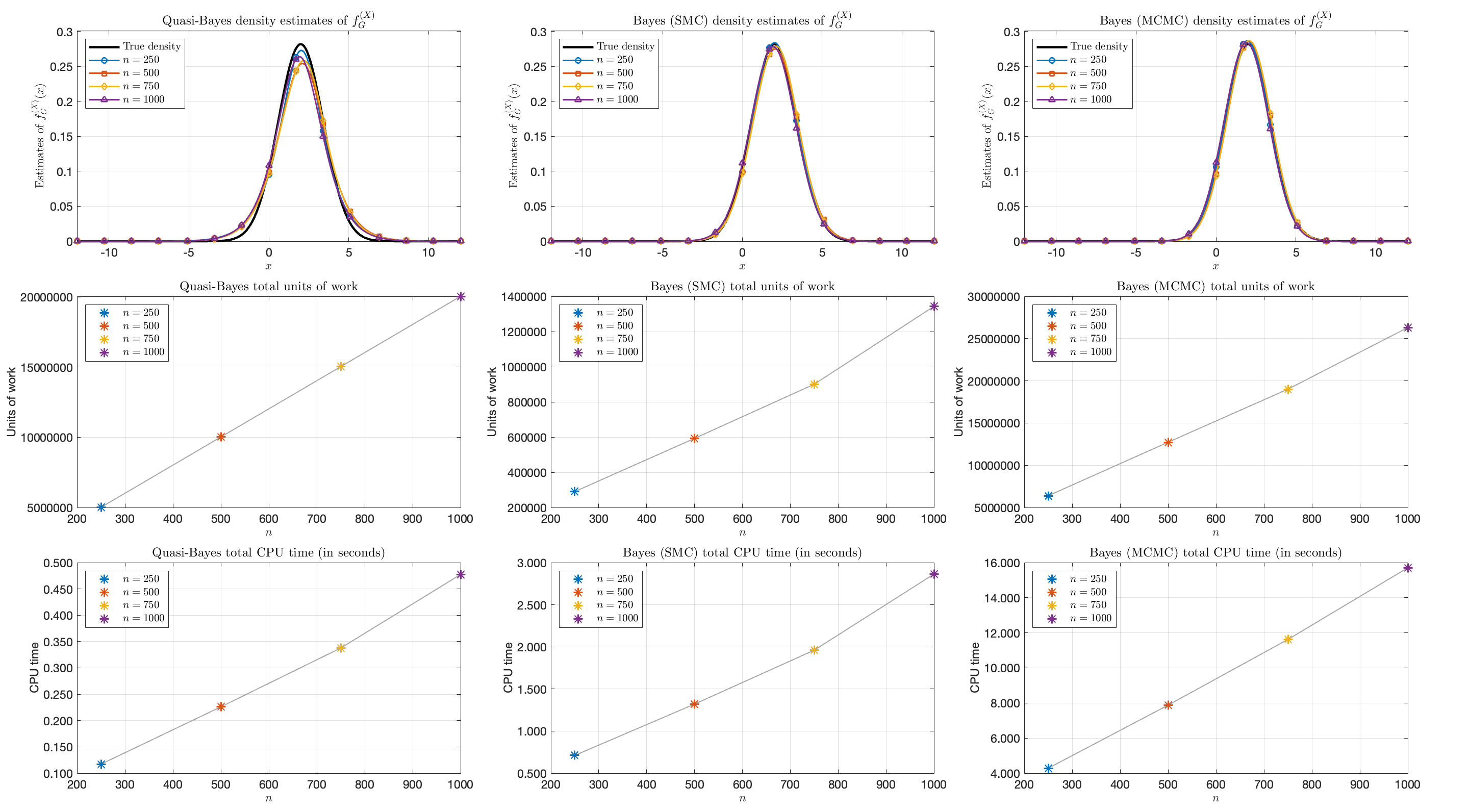}
 \caption{\footnotesize{Unimodal example, with Gaussian noise ($\sigma_g=0.25$): comparison between quasi-Bayesian and Bayesian approaches.}}
 \label{fig:supp-unimodal-bayes-gaussian}
\end{figure}

\subsubsection{Bimodal example I}\label{sec:supp-bimodal1}

\paragraph{Synthetic-data generation.}
For each $n\in\{250,500,750,1000\}$, we generate random variables $X_1,\ldots,X_n$ i.i.d. according to the density $f^{(X)}(\cdot)
 =0.3\,\phi(\cdot\mid-1,2)+0.7\,\phi(\cdot\mid3,1.5)$, where $\phi(\cdot\mid\mu,\sigma^2)$ denotes the Gaussian density with mean $\mu$ and variance $\sigma^2$. Using the mixture representation \eqref{eq:mixture}, we write $f^{(X)}$ as $f_{G^*}^{(X)}$, with $G^*=0.3\,\delta_{(-1,2)}+0.7\,\delta_{(3,1.5)}$ the true mixing distribution. We consider ordinary-smooth and super-smooth distributions for the noise random variables $Z_{i}$'s. In the ordinary-smooth case, $Z_i\stackrel{\mathrm{iid}}{\sim}\mathrm{Laplace}(0,b_l)$ with $b_l=\frac{\sigma_l}{\sqrt{2}}$ and standard deviation $\sigma_l\in\{0.25,0.50\}$; in the super-smooth case, $Z_i\stackrel{\mathrm{iid}}{\sim}N(0,\sigma_g^2)$ with standard deviation $\sigma_g\in\{0.25,0.50\}$. The  $Z_{i}$'s are independent of the $X_{i}$'s, and the observations are modeled as follows:
\begin{displaymath}
 Y_i=X_i+Z_i,\qquad i=1,\ldots,n.
\end{displaymath}
For each $n$, the same realization of $Y_{1:n}=(Y_1,\ldots,Y_n)$ is used for all methods under comparison.

\paragraph{Quasi-Bayesian estimation and uncertainty quantification}
We use the same implementation of Section~\ref{sec:supp-unimodal-newton}: the parameter space, numerical grids, initial distribution, learning rate, trapezoidal quadrature, renormalization, number of permutations, and random seed are unchanged. For the $r$-th permutation, let $\widetilde G_{n}^{(r)}$ denote the resulting Newton's mixing-distribution estimate, $r=1,\ldots,R_{n}$. Then, the order-averaged estimate is defined as
\begin{equation}\label{mode_estim_bim1}
\overline  G_{n}=\frac{1}{R_n}\sum_{r=1}^{R_n}\widetilde G_{n}^{(r)},
\end{equation}
which, following \eqref{eq:est_unim}, induces the quasi-Bayes estimate $f_{\overline  G_n}^{(X)}$ of $f_{G}^{(X)}$. The quasi-Bayes asymptotic credible intervals and bands for $f_{G}^{(X)}$ are also computed directly from $\overline  G_{n}$, following~\eqref{eq:ci_unim}--\eqref{eq:cb_unim}. Figures~\ref{fig:supp-bimodal1-newton-laplace-small}--\ref{fig:supp-bimodal1-newton-gaussian-small} report the quasi-Bayes estimates of $f_{G}^{(X)}$, and the asymptotic credible intervals and bands. The figures consist of two columns, which correspond to the noise standard deviations $0.25$ and $0.50$, respectively.  The first row of each figure compares the true density $f_{G^*}^{(X)}$ with the quasi-Bayes density estimates, while the second row displays the credible intervals, and the third row displays the credible bands on the interval $I=[-4,6]$.

\begin{figure}[t]
 \centering
\includegraphics[
  width=1\textwidth,
  trim=20 15 20 0,
  clip
]{bimodal1_laplace_small_n.png}
 \caption{\footnotesize{Bimodal example I, with Laplace noise: estimates, credible intervals and bands.}}
 \label{fig:supp-bimodal1-newton-laplace-small}
\end{figure}

\begin{figure}[t]
 \centering
 \includegraphics[
  width=1\textwidth,
  trim=20 15 20 0,
  clip
]{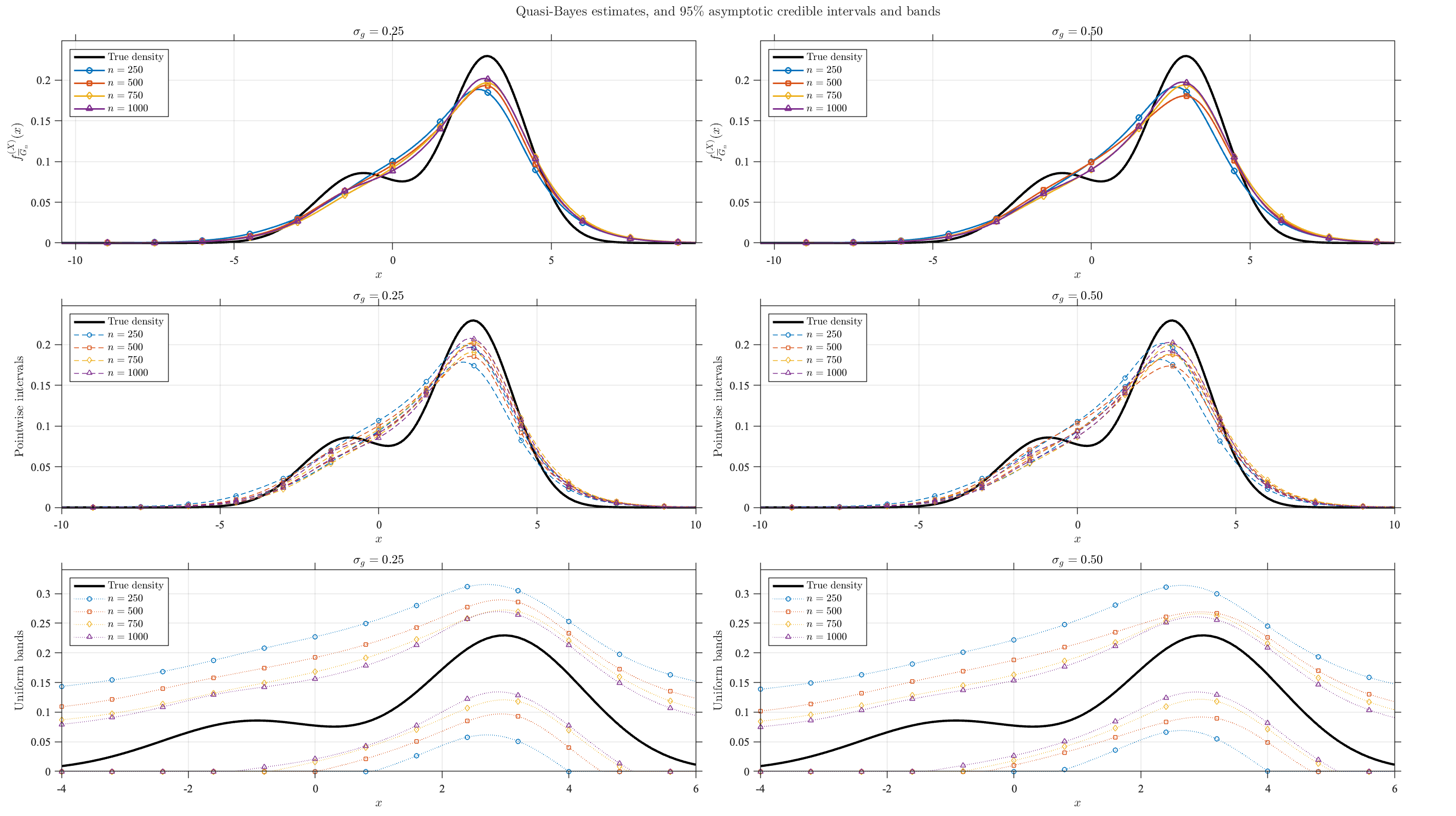}
 \caption{\footnotesize{Bimodal example I, with Gaussian noise: estimates, credible intervals and bands.}}
 \label{fig:supp-bimodal1-newton-gaussian-small}
\end{figure}

\paragraph{Monte Carlo uncertainty quantification}
The Monte Carlo procedure is the same as the procedure presented in Section~\ref{sec:supp-unimodal-mc}; it relies on the continuation of the quasi-Bayesian learning process \eqref{eq:Y}--\eqref{eq:newton}, initialized at the order-averaged mixing distribution $\overline  G_{n}$ in \eqref{mode_estim_bim1}.  The continuation length, number of Monte Carlo replications, and random seed are unchanged. Further, the Monte Carlo estimates, credible intervals and bands are constructed as in \eqref{eq:unimodal-mc-est}--\eqref{eq:supp-unimodal-mc-band}. Figures~\ref{fig:supp-bimodal1-mc-laplace-025}--\ref{fig:supp-bimodal1-mc-gaussian-050} display Monte Carlo credible intervals and bands, which are compared with the corresponding quasi-Bayes asymptotic credible intervals and bands.

\begin{figure}[t]
 \centering
 \includegraphics[
  width=0.9\textwidth,
  trim=20 15 20 0,
  clip
]{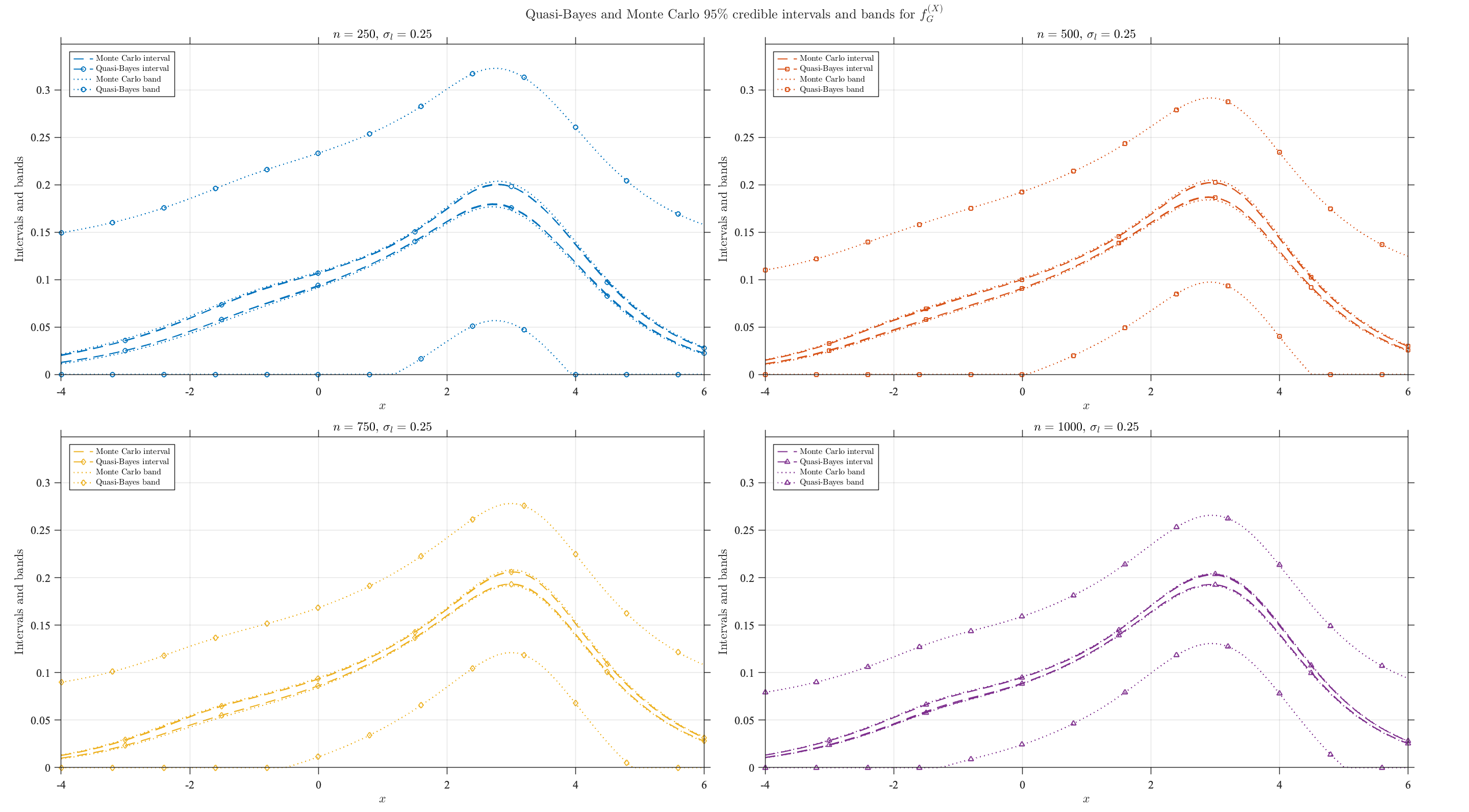}
 \caption{\footnotesize{Bimodal example I, with Laplace noise ($\sigma_{l}=0.25$): quasi-Bayes versus Monte Carlo credible intervals and bands.}}
 \label{fig:supp-bimodal1-mc-laplace-025}
\end{figure}

\begin{figure}[t]
 \centering
  \includegraphics[
  width=0.9\textwidth,
  trim=20 15 20 0,
  clip
]{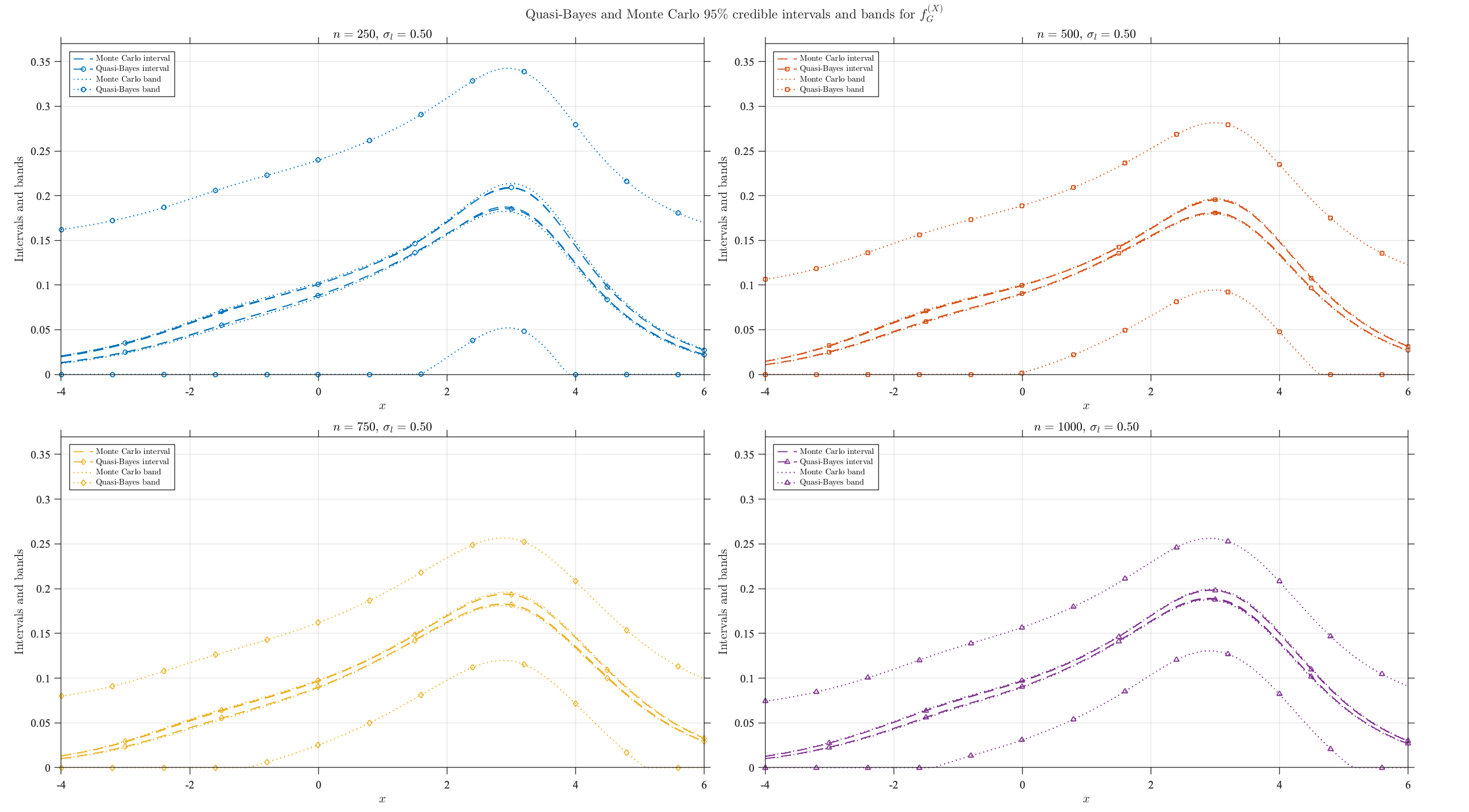}
 \caption{\footnotesize{Bimodal example I, with Laplace noise ($\sigma_{l}=0.50$): quasi-Bayes versus Monte Carlo credible intervals and bands.}}
 \label{fig:supp-bimodal1-mc-laplace-050}
\end{figure}

\begin{figure}[t]
 \centering
 \includegraphics[
  width=0.9\textwidth,
  trim=20 15 20 0,
  clip
]{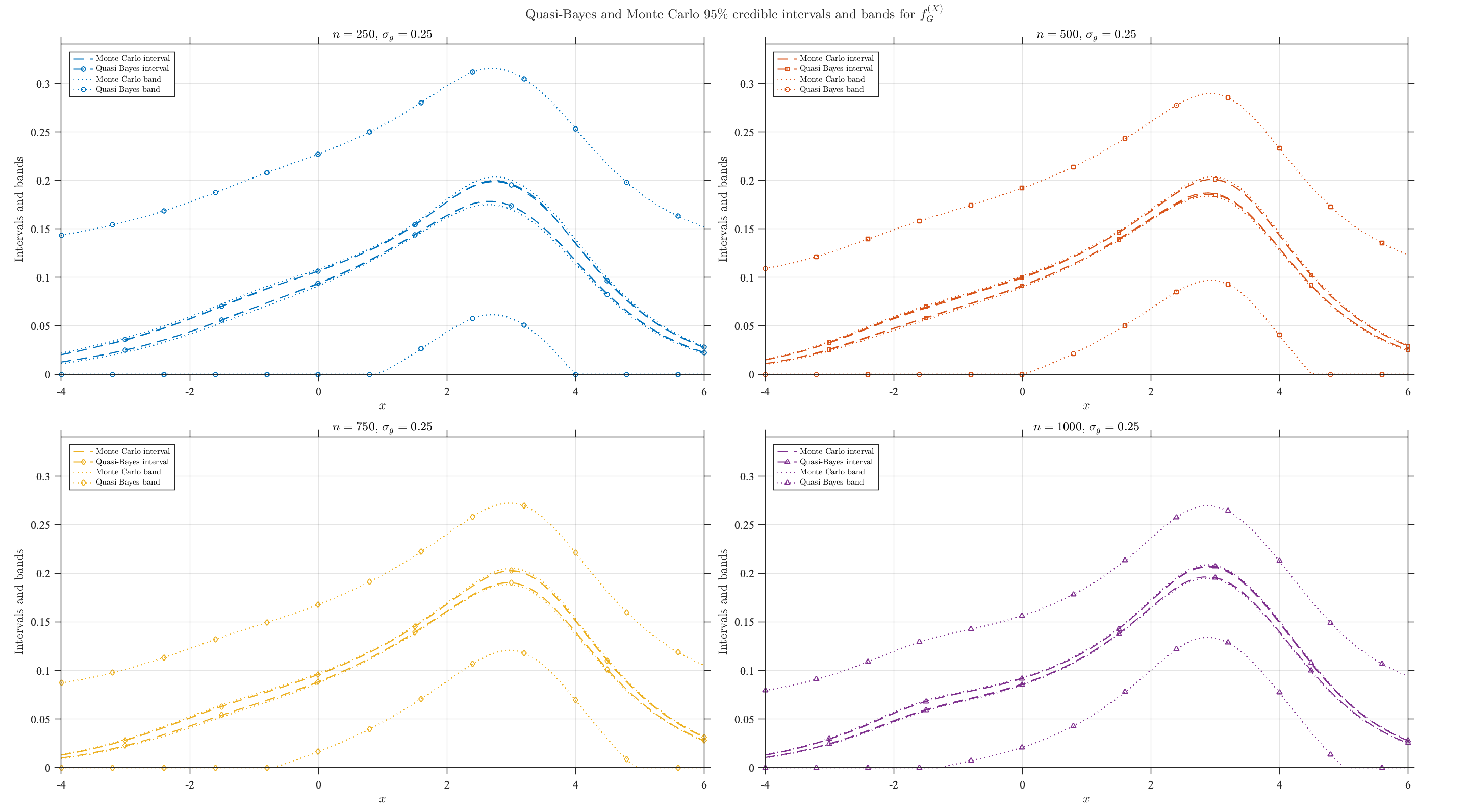}
 \caption{\footnotesize{Bimodal example I, with Gaussian noise ($\sigma_{g}=0.25$): quasi-Bayes versus Monte Carlo credible intervals and bands.}}
 \label{fig:supp-bimodal1-mc-gaussian-025}
\end{figure}

\begin{figure}[t]
 \centering
 \includegraphics[
  width=0.9\textwidth,
  trim=20 15 20 0,
  clip
]{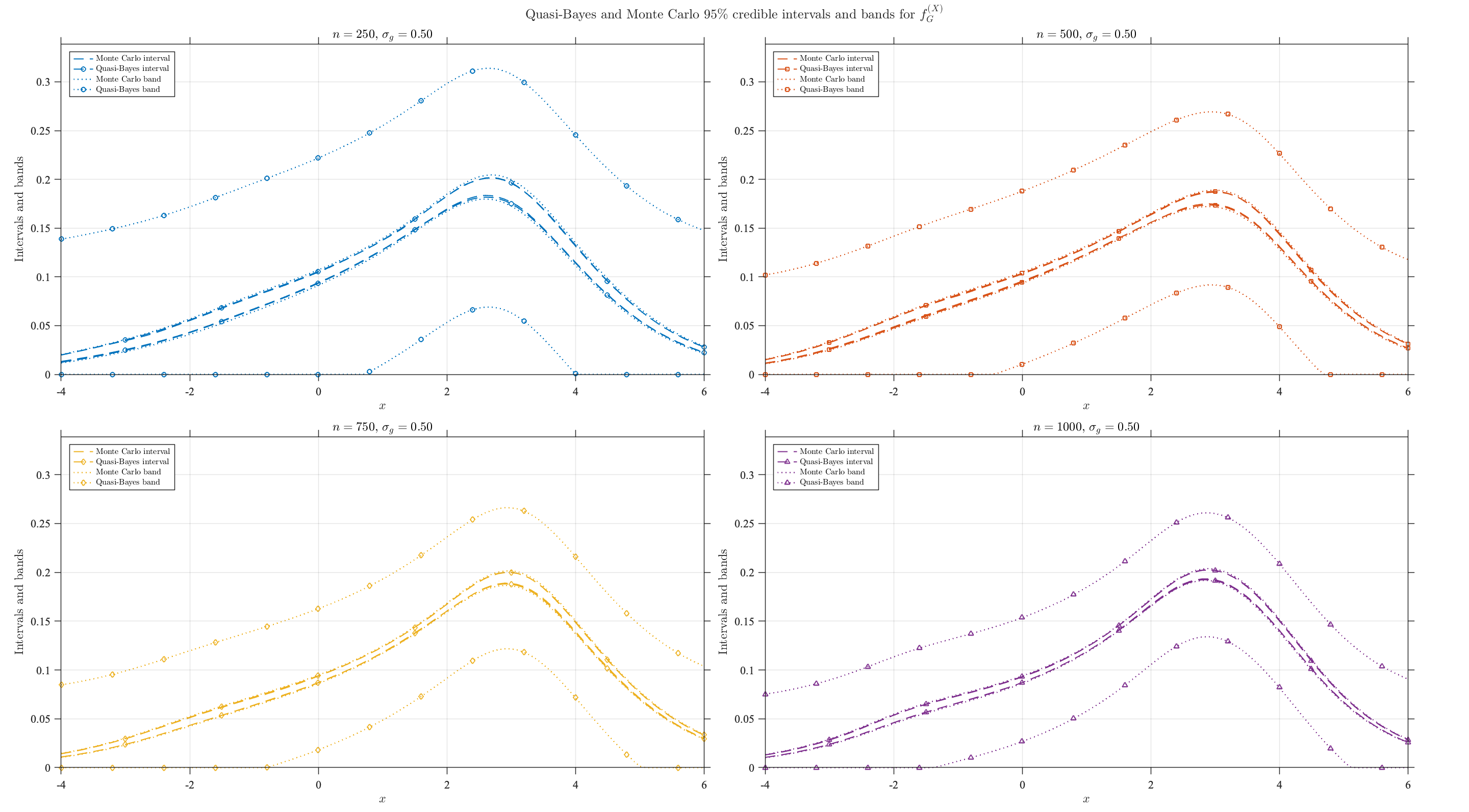}
 \caption{\footnotesize{Bimodal example I, with Gaussian noise ($\sigma_{g}=0.50$): quasi-Bayes versus Monte Carlo credible intervals and bands.}}
 \label{fig:supp-bimodal1-mc-gaussian-050}
\end{figure}

\paragraph{Dirichlet process mixture model benchmarks}
The Dirichlet process Gaussian-mixture model \eqref{eq:supp-unimodal-dp-model}, the SMC algorithm, Algorithm~8, tuning parameters, and computational-work definitions are identical to those described in Section~\ref{sec:supp-unimodal-bayes}.  Furthermore, as in the unimodal example, Newton's and SMC algorithms are averaged over the same $R_{n}$ permutations, whereas Neal's Algorithm~8 is applied once to each complete dataset.  Figures~\ref{fig:supp-bimodal1-bayes-laplace}--\ref{fig:supp-bimodal1-bayes-gaussian} summarize the comparison between the quasi-Bayesian and Bayesian approaches. Their first row displays the estimates of $f_{G}^{(X)}$ obtained with Newton's algorithm, the SMC algorithm, and Algorithm~8; the second and third rows report total work and total CPU time, respectively.

\begin{figure}[t]
 \centering
 \includegraphics[
  width=1\textwidth,
  trim=20 15 20 15,
  clip
]{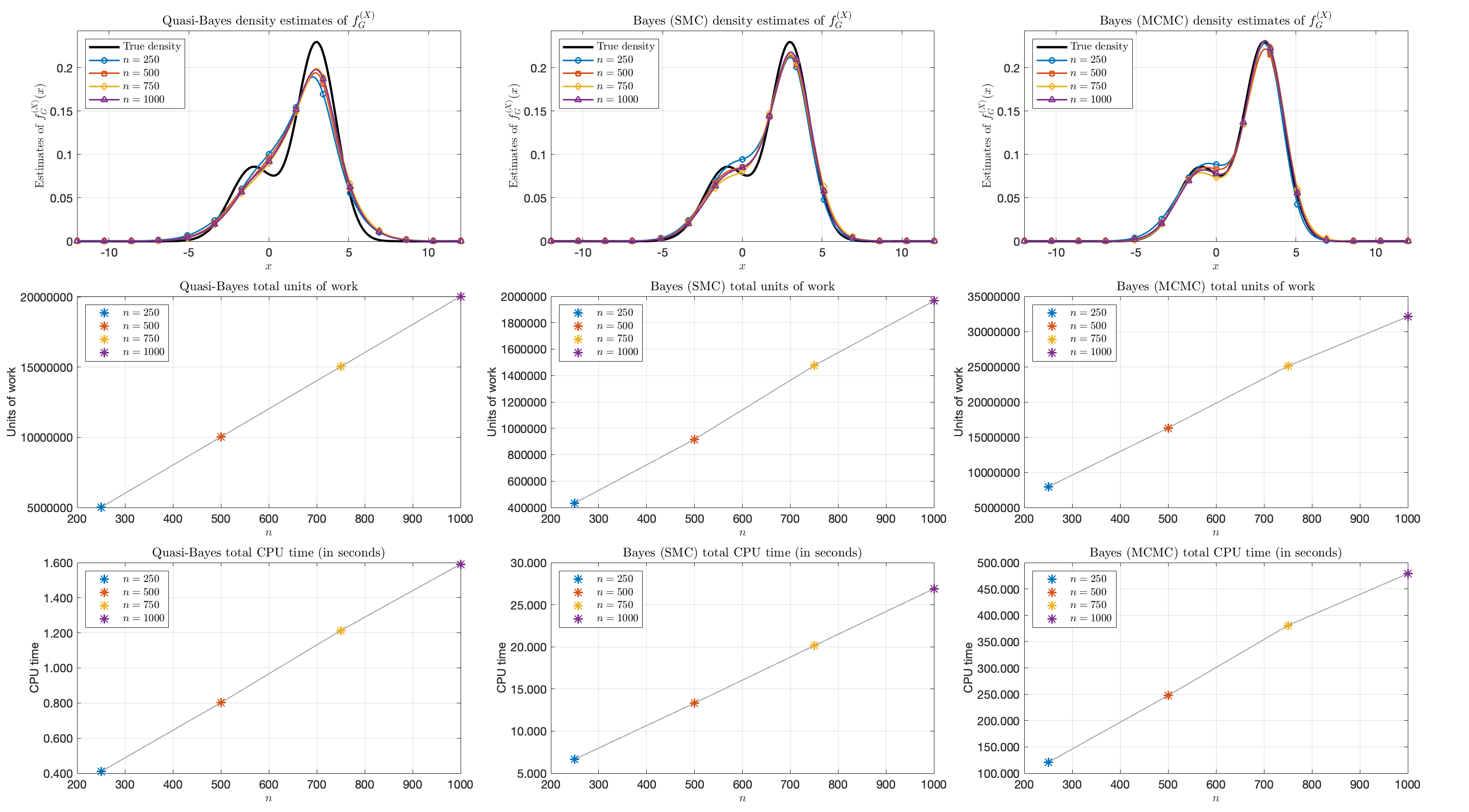}
 \caption{\footnotesize{Bimodal example I, with Laplace noise ($\sigma_l=0.25$): comparison between quasi-Bayesian and Bayesian approaches.}}
 \label{fig:supp-bimodal1-bayes-laplace}
\end{figure}

\begin{figure}[t]
 \centering
  \includegraphics[
  width=1\textwidth,
  trim=20 15 20 15,
  clip
]{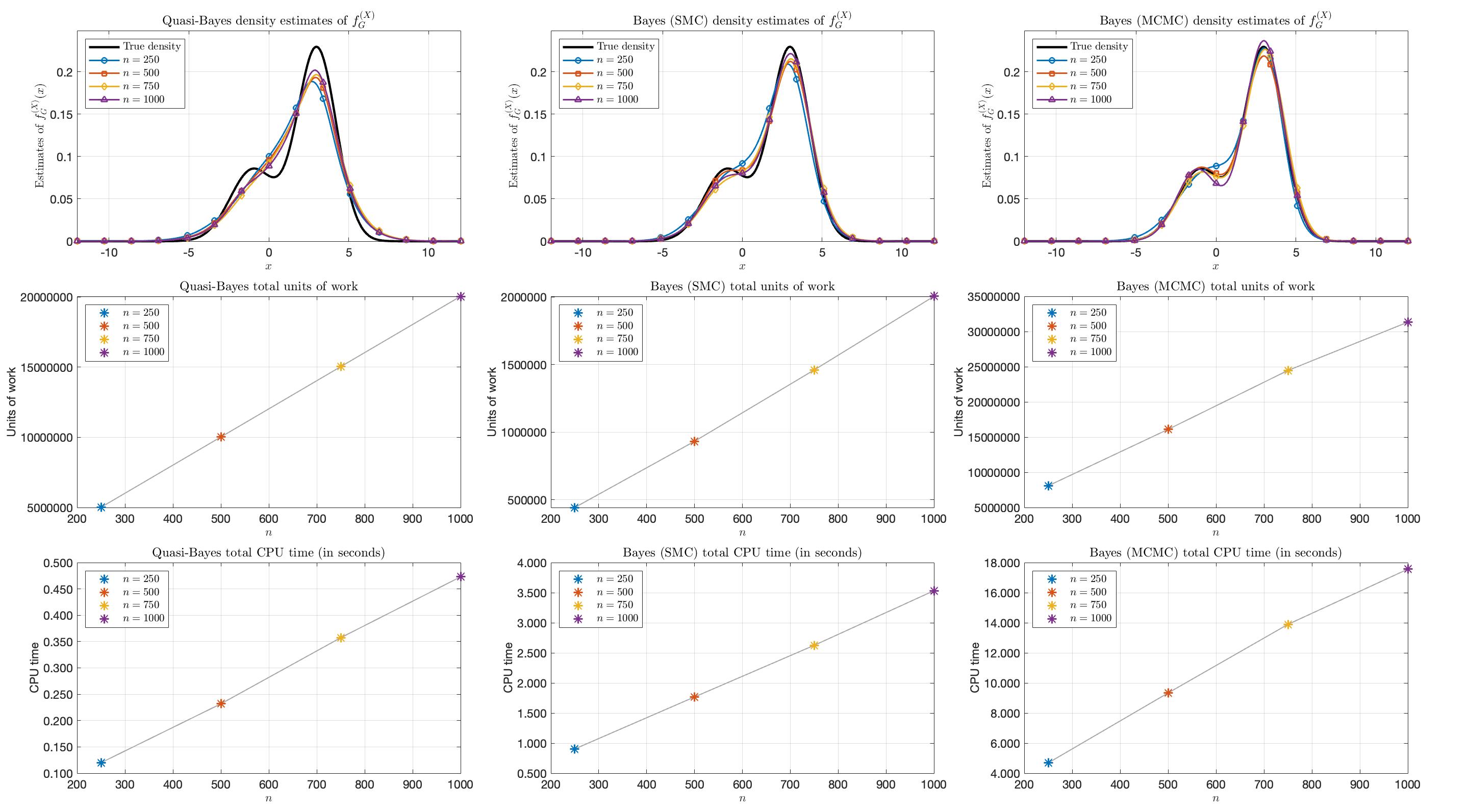}
 \caption{\footnotesize{Bimodal example I, with Gaussian noise ($\sigma_g=0.25$): comparison between quasi-Bayesian and Bayesian approaches.}}
 \label{fig:supp-bimodal1-bayes-gaussian}
\end{figure}

\subsubsection{Bimodal example II}\label{sec:supp-bimodal2}

\paragraph{Synthetic-data generation}
For each $n\in\{250,500,750,1000\}$, we generate random variables $X_1,\ldots,X_n$ i.i.d. according to the density $f^{(X)}(\cdot)=0.5\,\phi(\cdot\mid-2,1)+0.5\,\phi(\cdot\mid4,1)$, where $\phi(\cdot\mid\mu,\sigma^2)$ denotes the Gaussian density with mean $\mu$ and variance $\sigma^2$. Using the mixture representation \eqref{eq:mixture}, we write $f^{(X)}$ as $f_{G^*}^{(X)}$, with $G^*=0.5\,\delta_{(-2,1)}+0.5\,\delta_{(4,1)}$ the true mixing distribution. We consider ordinary-smooth and super-smooth distributions for the noise random variables $Z_{i}$'s. In the ordinary-smooth case, $Z_i\stackrel{\mathrm{iid}}{\sim}\mathrm{Laplace}(0,b_l)$ with $b_l=\frac{\sigma_l}{\sqrt{2}}$ and standard deviation $\sigma_l\in\{0.25,0.50\}$; in the super-smooth case, $Z_i\stackrel{\mathrm{iid}}{\sim}N(0,\sigma_g^2)$ with standard deviation $\sigma_g\in\{0.25,0.50\}$. The  $Z_{i}$'s are independent of the $X_{i}$'s, and the observations are modeled as follows:
\begin{displaymath}
 Y_i=X_i+Z_i,\qquad i=1,\ldots,n.
\end{displaymath}
For each $n$, the same realization of $Y_{1:n}=(Y_1,\ldots,Y_n)$ is used for all methods under comparison.

\paragraph{Quasi-Bayesian estimation and uncertainty quantification}
We use the same implementation of Section~\ref{sec:supp-unimodal-newton}: the parameter space, numerical grids, initial distribution, learning rate, trapezoidal quadrature, renormalization, number of permutations, and random seed are unchanged. For the $r$-th permutation, let $\widetilde G_{n}^{(r)}$ denote the resulting Newton's mixing-distribution estimate, $r=1,\ldots,R_{n}$.  Then, the order-averaged estimate is defined as
\begin{equation}\label{mode_estim_bim2}
 \overline  G_{n}=\frac{1}{R_n}\sum_{r=1}^{R_n}\widetilde G_{n}^{(r)},
\end{equation}
which, following \eqref{eq:est_unim}, induces the quasi-Bayes estimate $f_{\overline  G_n}^{(X)}$ of $f_{G}^{(X)}$. The quasi-Bayes asymptotic credible intervals and bands for $f_{G}^{(X)}$ are also computed directly from $\overline  G_{n}$, following~\eqref{eq:ci_unim}--\eqref{eq:cb_unim}. Figures~\ref{fig:supp-bimodal2-newton-laplace-small}--\ref{fig:supp-bimodal2-newton-gaussian-small} report the quasi-Bayes estimates of $f_{G}^{(X)}$, and the asymptotic credible intervals and bands. The figures consist of two columns, which correspond to the noise standard deviations $0.25$ and $0.50$, respectively.  The first row of each figure compares the true density $f_{G^*}^{(X)}$ with the quasi-Bayes density estimates, while the second row displays the credible intervals, and the third row displays the credible bands on the interval $I=[-4,6]$.

\begin{figure}[t]
 \centering
\includegraphics[
  width=1\textwidth,
  trim=20 15 20 0,
  clip
]{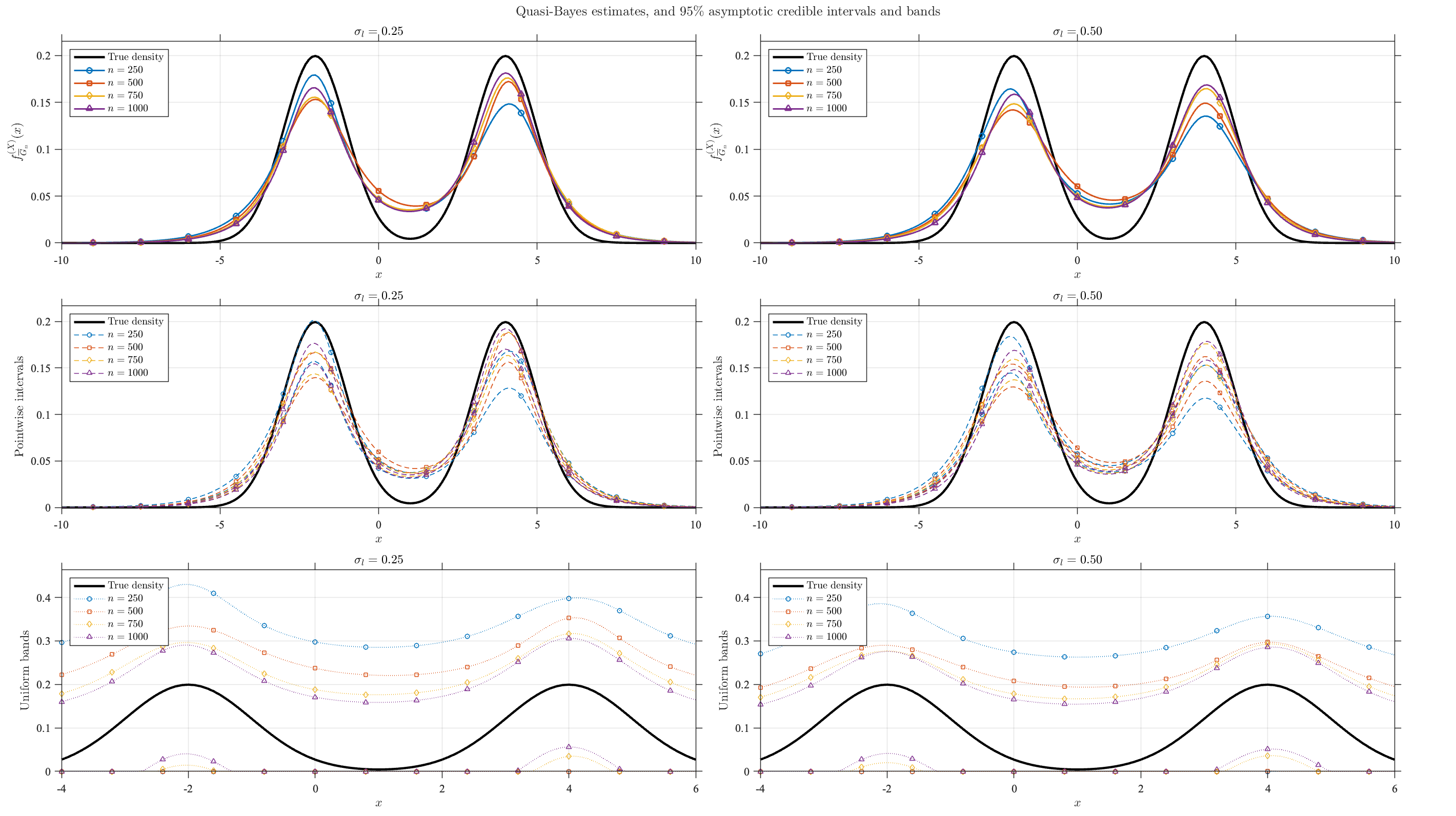}
 \caption{\footnotesize{Bimodal example II, with Laplace noise: estimates, credible intervals and bands.}}
 \label{fig:supp-bimodal2-newton-laplace-small}
\end{figure}

\begin{figure}[t]
 \centering
 \includegraphics[
  width=1\textwidth,
  trim=20 15 20 0,
  clip
]{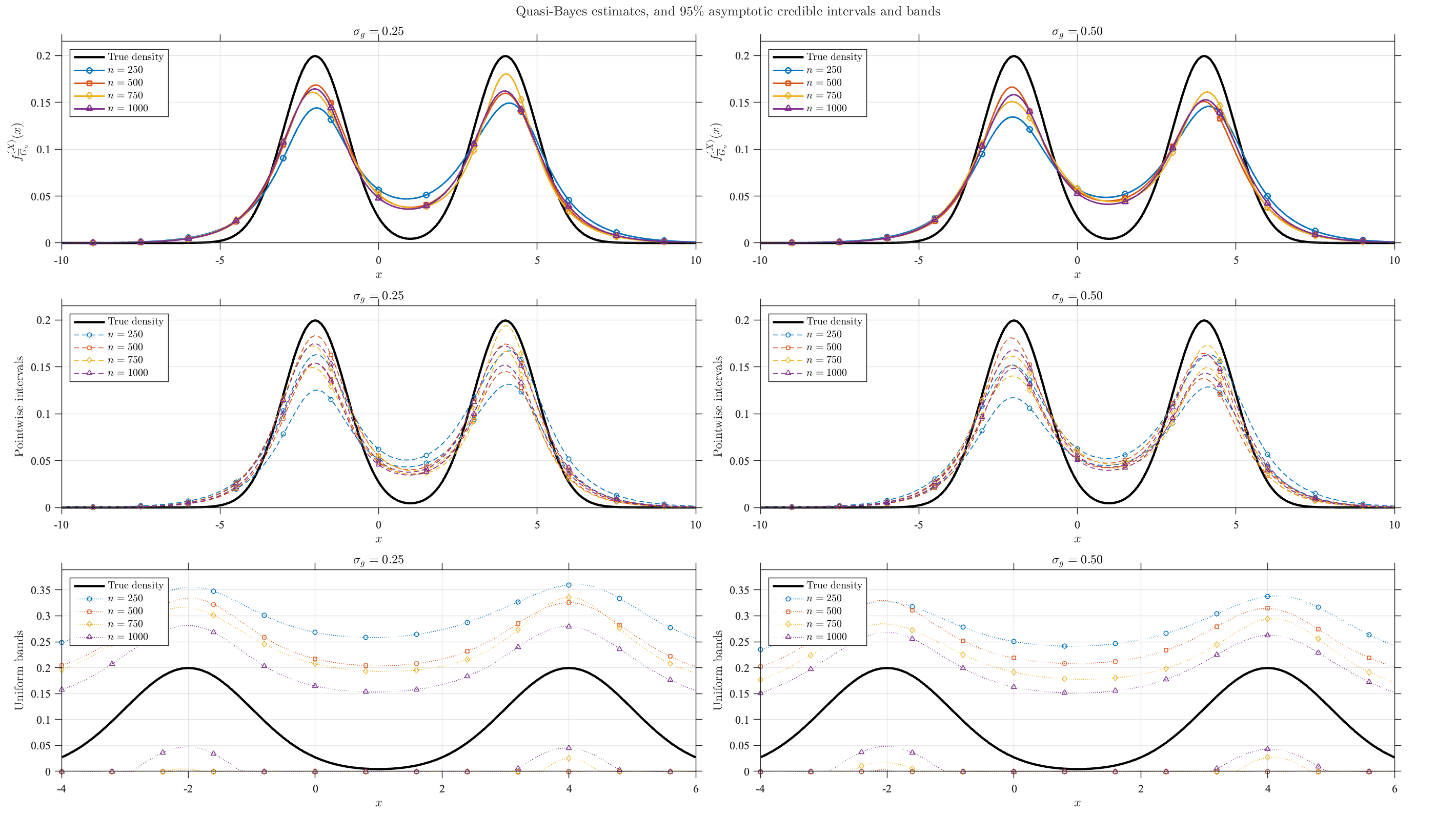}
 \caption{\footnotesize{Bimodal example II, with Gaussian noise: estimates, credible intervals and bands.}}
 \label{fig:supp-bimodal2-newton-gaussian-small}
\end{figure}

\paragraph{Monte Carlo uncertainty quantification}
The Monte Carlo procedure is the same as the procedure presented in Section~\ref{sec:supp-unimodal-mc}; it relies on the continuation of the quasi-Bayesian learning process \eqref{eq:Y}--\eqref{eq:newton}, initialized at the order-averaged mixing distribution $\overline  G_{n}$ in \eqref{mode_estim_bim2}. The continuation length, number of Monte Carlo replications, and random seed are unchanged. Further, the Monte Carlo estimates, credible intervals and bands are constructed as in \eqref{eq:unimodal-mc-est}--\eqref{eq:supp-unimodal-mc-band}. Figures~\ref{fig:supp-bimodal2-mc-laplace-025}--\ref{fig:supp-bimodal2-mc-gaussian-050} display Monte Carlo credible intervals and bands, which are compared with the corresponding quasi-Bayes asymptotic credible intervals and bands.

\begin{figure}[t]
 \centering
 \includegraphics[
  width=0.9\textwidth,
  trim=20 15 20 0,
  clip
]{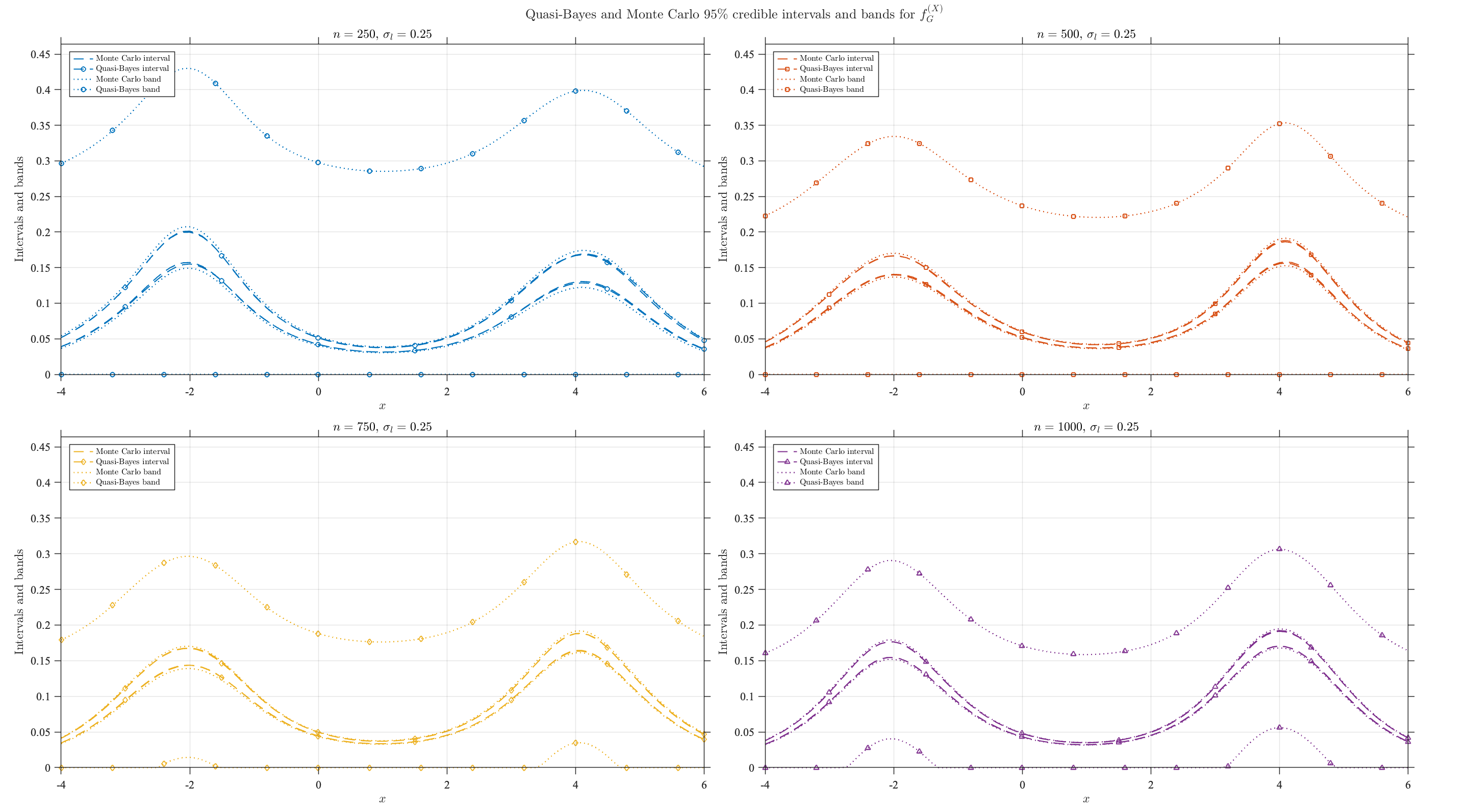}
 \caption{\footnotesize{Bimodal example II, with Laplace noise ($\sigma_l=0.25$): quasi-Bayes versus Monte Carlo credible intervals and bands.}}
 \label{fig:supp-bimodal2-mc-laplace-025}
\end{figure}

\begin{figure}[t]
 \centering
  \includegraphics[
  width=0.9\textwidth,
  trim=20 15 20 0,
  clip
]{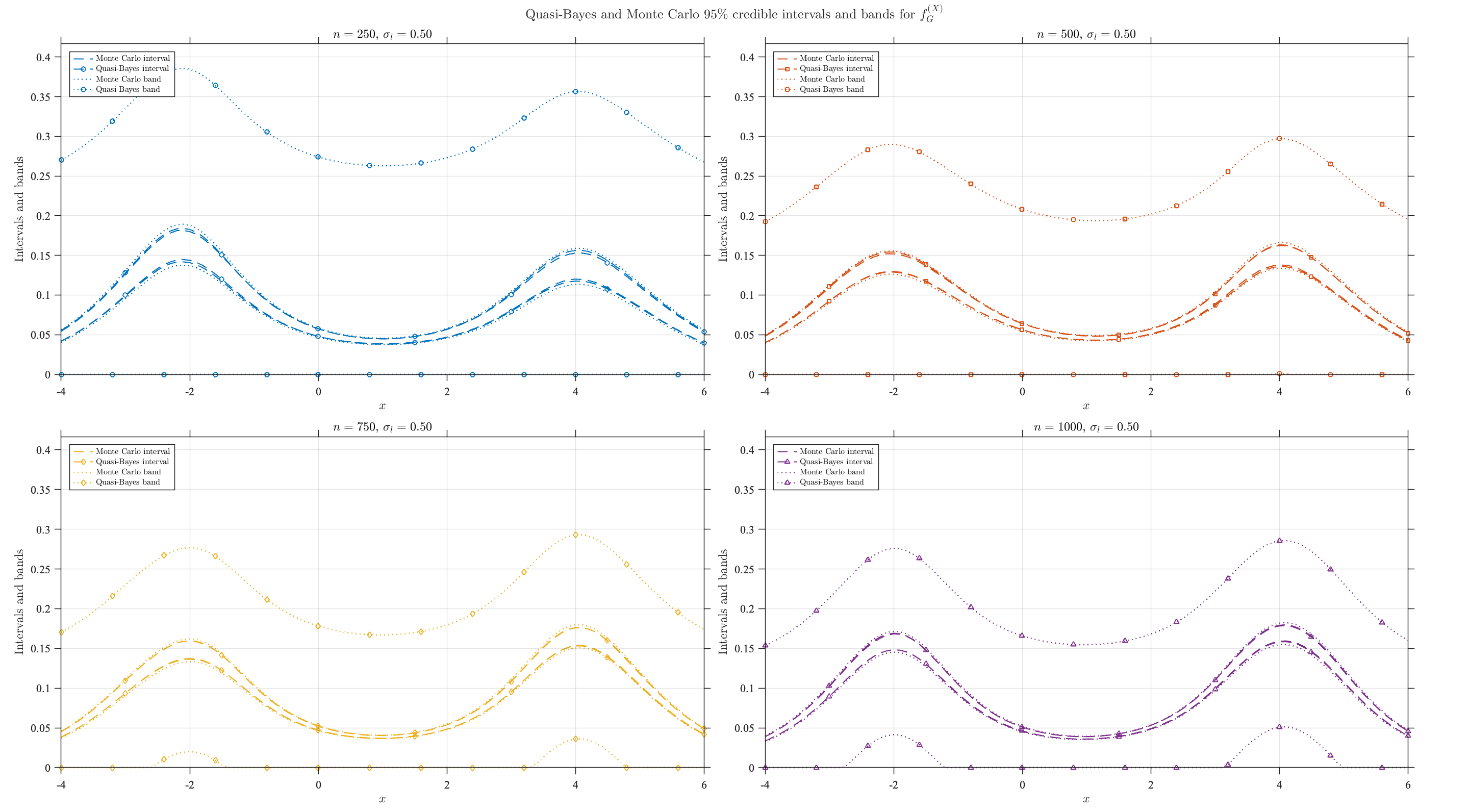}
 \caption{\footnotesize{Bimodal example II, with Laplace noise ($\sigma_l=0.50$): quasi-Bayes versus Monte Carlo credible intervals and bands.}}
 \label{fig:supp-bimodal2-mc-laplace-050}
\end{figure}

\begin{figure}[t]
 \centering
 \includegraphics[
  width=0.9\textwidth,
  trim=20 15 20 0,
  clip
]{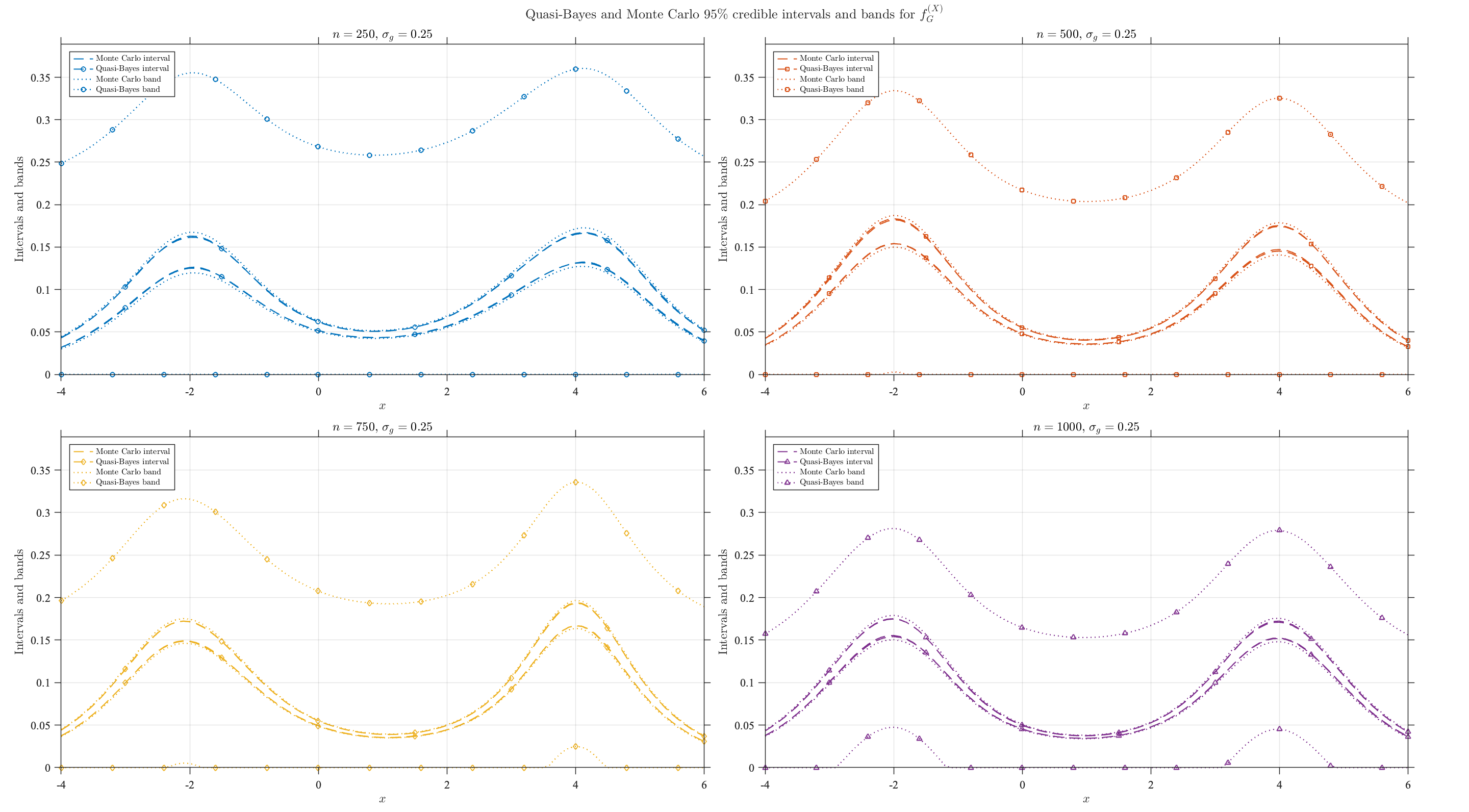}
 \caption{\footnotesize{Bimodal example II, with Gaussian noise ($\sigma_g=0.25$): quasi-Bayes versus Monte Carlo credible intervals and bands.}}
 \label{fig:supp-bimodal2-mc-gaussian-025}
\end{figure}

\begin{figure}[t]
 \centering
 \includegraphics[
  width=0.9\textwidth,
  trim=20 15 20 0,
  clip
]{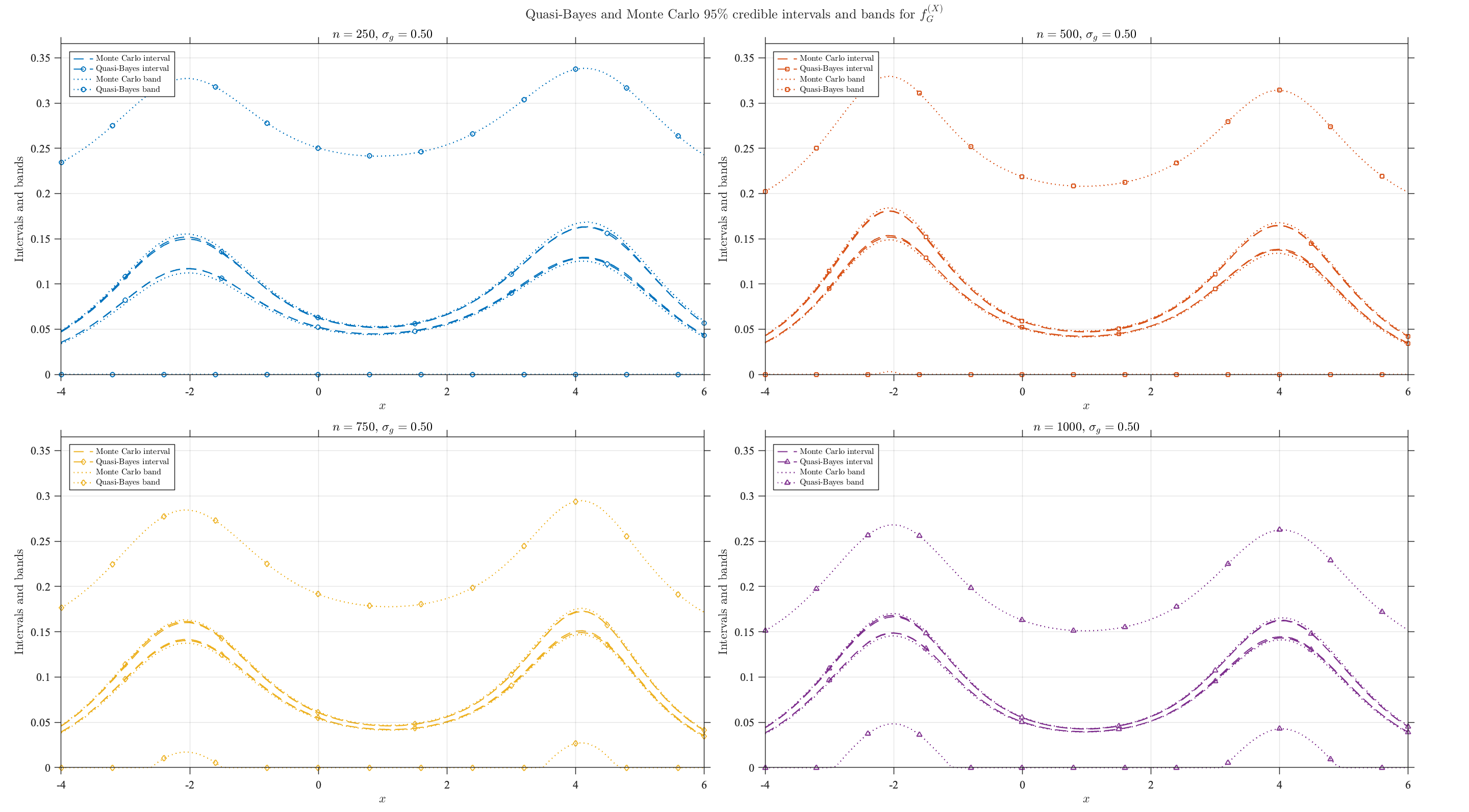}
 \caption{\footnotesize{Bimodal example II, with Gaussian noise ($\sigma_g=0.50$): quasi-Bayes versus Monte Carlo credible intervals and bands.}}
 \label{fig:supp-bimodal2-mc-gaussian-050}
\end{figure}

\paragraph{Dirichlet process mixture model benchmarks}
The Dirichlet process Gaussian-mixture model \eqref{eq:supp-unimodal-dp-model}, the SMC algorithm, Algorithm~8, tuning parameters, and computational-work definitions are identical to those described in Section~\ref{sec:supp-unimodal-bayes}.  Furthermore, as in the unimodal example, Newton's and SMC algorithms are averaged over the same $R_{n}$ permutations, whereas Neal's Algorithm~8 is applied once to each complete dataset.  Figures~\ref{fig:supp-bimodal2-bayes-laplace}--\ref{fig:supp-bimodal2-bayes-gaussian} summarize the comparison between the quasi-Bayesian and Bayesian approaches. Their first row displays the estimates of $f_{G}^{(X)}$ obtained with Newton's algorithm, the SMC algorithm, and Algorithm~8; the second and third rows report total work and total CPU time, respectively.

\begin{figure}[t]
 \centering
 \includegraphics[
  width=1\textwidth,
  trim=20 15 20 15,
  clip
]{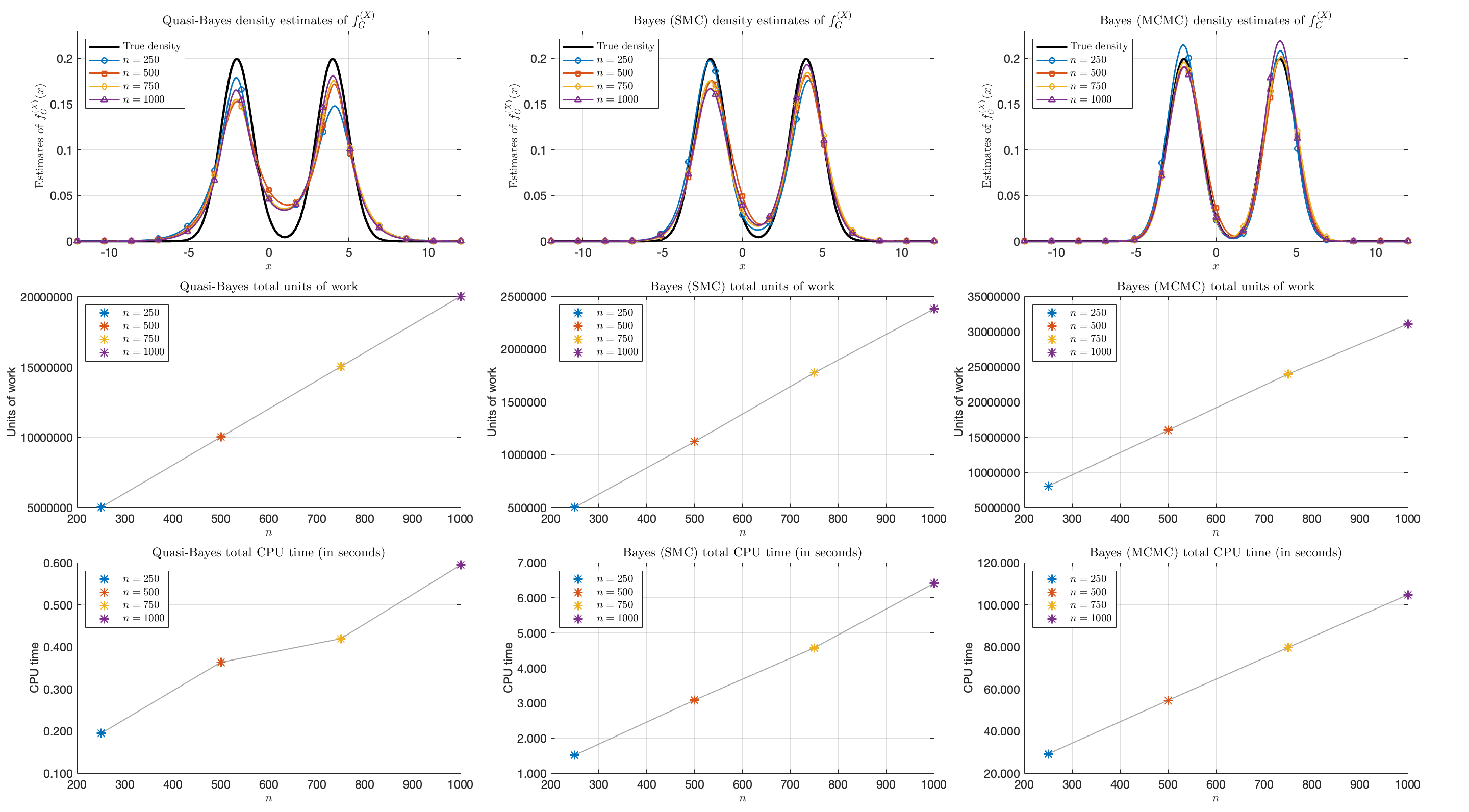}
 \caption{\footnotesize{Bimodal example II, with Laplace noise ($\sigma_l=0.25$): comparison between quasi-Bayesian and Bayesian approaches.}}
 \label{fig:supp-bimodal2-bayes-laplace}
\end{figure}

\begin{figure}[t]
 \centering
  \includegraphics[
  width=1\textwidth,
  trim=20 15 20 15,
  clip
]{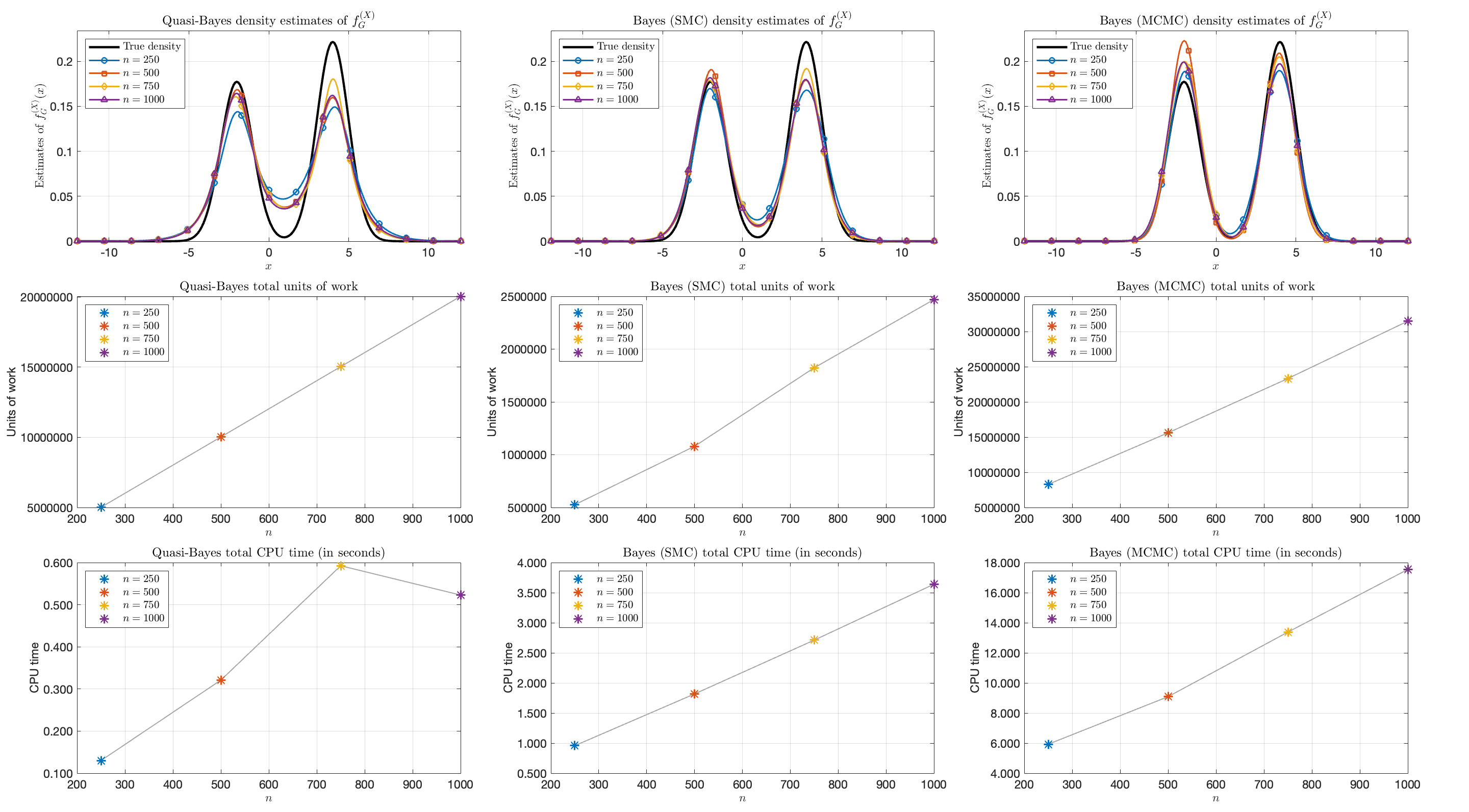}
 \caption{\footnotesize{Bimodal example II, with Gaussian noise ($\sigma_g=0.25$): comparison between quasi-Bayesian and Bayesian approaches.}}
 \label{fig:supp-bimodal2-bayes-gaussian}
\end{figure}


\subsection{Real-data analysis}\label{sec:real-data}

\subsubsection{Flow-cytometry data}\label{sec:supp-flow-cytometry}

In Section~\ref{sec61} we analyzed flow-cytometry data from the study of \citet{Tor(23)} on the expression of Brachyury in differentiating mouse embryonic stem cells. Brachyury is a mesodermal transcription factor whose expression varies across cells during differentiation. Flow cytometry provides a large number of single-cell fluorescence measurements by recording the optical signal generated as cells pass individually through the instrument. Each recorded event therefore corresponds to a distinct cell and is associated with both a fluorescence intensity and an acquisition time. The fluorescence measured on a labeled cell contains the reported signal associated with Brachyury expression together with autofluorescence and other background contributions. To characterize the latter, the experiment also includes a separate population of unlabeled cells. The labeled and unlabeled measurements are consequently obtained from different cells and cannot be matched event-by-event. In particular, an observation from the unlabeled population does not represent the specific background realization affecting any individual labeled-cell measurement; rather, the unlabeled population provides an independent sample from the common background distribution.

The data are available at \url{https://github.com/dsb-lab/scBayesDeconv.jl} and include fluorescence intensities and acquisition-time measurements. Following the organization of the dataset in \citet{Tor(23)}, the first column is used as the sample from the labeled population, whereas the sixth column is used as the autofluorescence control sample. The acquisition-time measurements are used to restore the order in which the labeled cells were recorded. We then process the observations sequentially and examine the estimates over the nested prefixes corresponding to the first $n\in\{250,\,1{,}000,\,5{,}000,\,10{,}000\}$ events. Let $Y_i$ denote the fluorescence intensity of the $i$-th labeled cell in acquisition order and assume the additive model $Y_i=X_i+Z_i$ for $i\geq1$, where $X_i$ is the reporter fluorescence and $Z_i$ is the autofluorescence and background contribution. The control measurements consist of $m\geq1$ independent samples, say $\widetilde Z_1,\ldots,\widetilde Z_m$, measured on unlabeled cells and not paired with the $Y_i$'s. 

The distributional assumptions on the $X_{i}$'s and $Z_{i}$'s, as well as on  $\widetilde Z_1,\ldots,\widetilde Z_m$,  are those specified in Section~\ref{sec61}. In particular, in the analysis of Section~\ref{sec61}, the random variables $Z_{i}$'s are i.i.d. according to a Gaussian distribution, i.e. $f_Z(z)=\phi(z\mid\mu_Z,\sigma_Z^2)$. Here, we investigate the sensitivity of the results to this specification of $f_{Z}$, by considering two additional working models for the noise distributions. More precisely, we consider a Laplace noise distribution, namely
\begin{displaymath}
f_Z(z)=\frac{1}{2b_Z}\exp\left\{-\frac{|z-\mu_Z|}{b_Z}\right\},
\end{displaymath}
and the four-component Gaussian-mixture noise distribution \citep{Tor(23)}, namely
\begin{displaymath}
f_Z(z)=\sum_{\ell=1}^{4}q_\ell\,\phi(z\mid\eta_\ell,\tau_\ell^2).
\end{displaymath}
With regards to the use of the control samples $\widetilde Z_1,\ldots,\widetilde Z_m$, within each noise distribution, the parameters of $f_Z$ are estimated once from the same control samples, and the resulting fitted density is kept fixed over the nested sample sizes $n\in\{250,\,1{,}000,\,5{,}000,\,10{,}000\}$. Moreover, the quasi-Bayesian, Bayesian, and ridge-regularized Fourier deconvolution approaches use exactly the same fitted density. Figures~\ref{fig:supp-fluo-laplace}--\ref{fig:supp-fluo-4gauss} report the estimates of $f_G^{(X)}$ obtained by the quasi-Bayesian, Bayesian, and ridge-regularized Fourier deconvolution approaches.

\begin{figure}[t]
 \centering
 \includegraphics[
   width=\textwidth,
   trim=20 15 20 0,
   clip
 ]{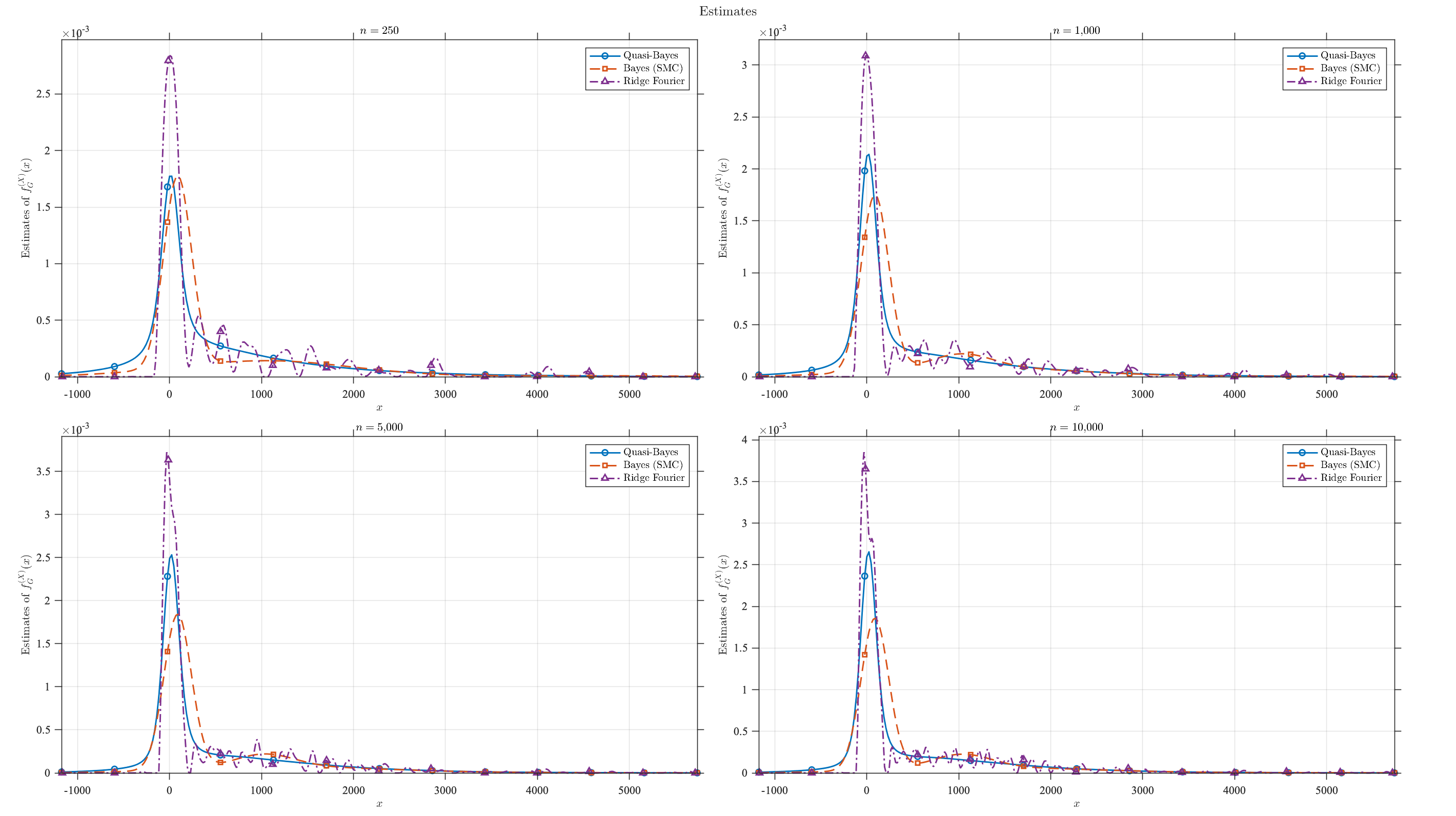}
  \caption{\footnotesize{Flow-cytometry data, with Laplace noise: comparison between quasi-Bayesian, Bayesian and ridge-regularized Fourier deconvolution approaches.}} \label{fig:supp-fluo-laplace}
\end{figure}

\begin{figure}[t]
 \centering
 \includegraphics[
   width=\textwidth,
   trim=20 15 20 0,
   clip
 ]{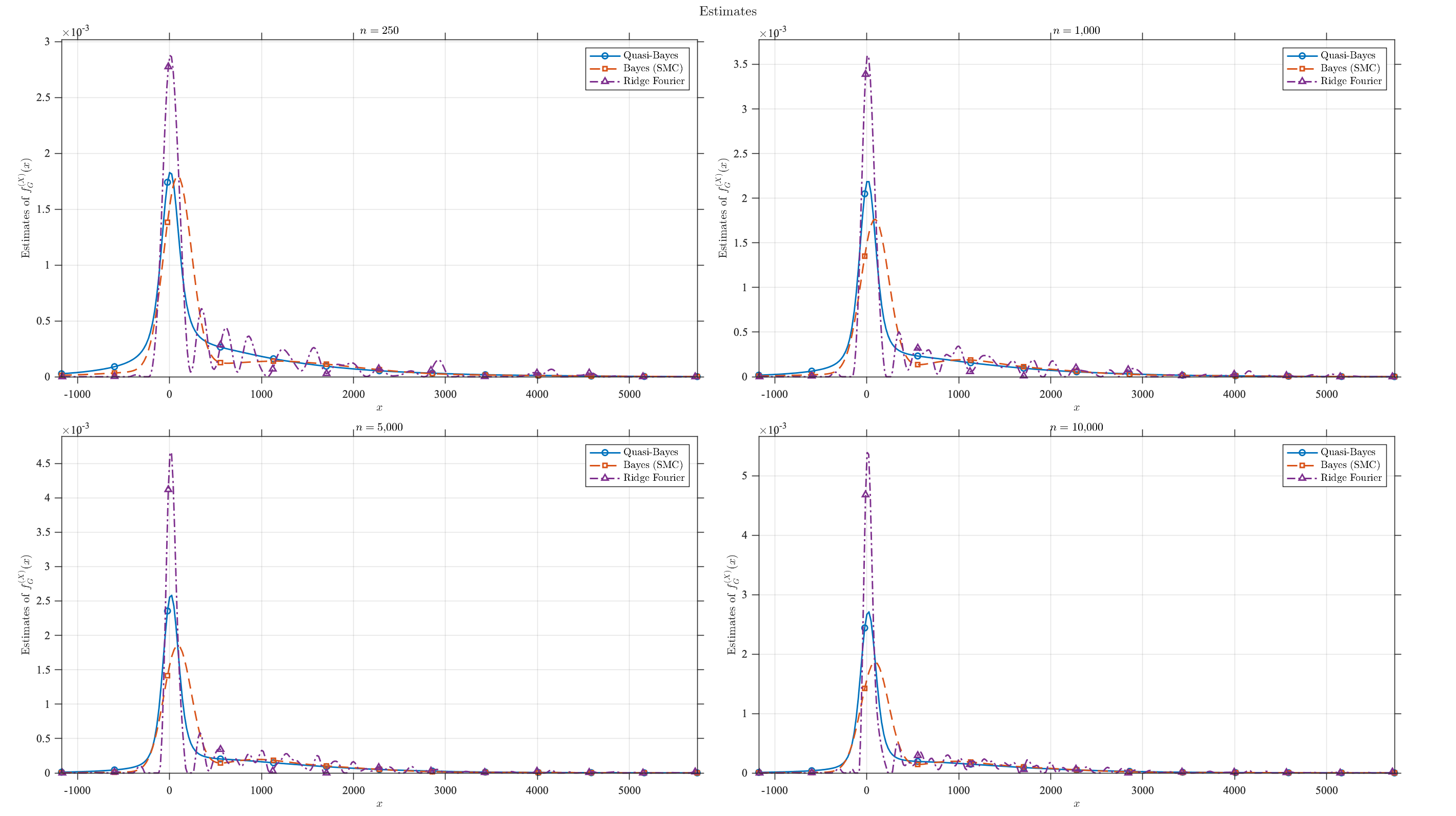}
  \caption{\footnotesize{Flow-cytometry data, with four-component Gaussian-mixture noise: comparison between quasi-Bayesian, Bayesian and ridge-regularized Fourier deconvolution approaches.}} \label{fig:supp-fluo-4gauss}
\end{figure}

To provide an additional graphical representation, we also augment the preceding figures with histograms based on a reconstructed (proxy) version of the reporter fluorescence  samples. Since the $X_{i}$'s are not observed, for each  $n\in\{250,\,1{,}000,\,5{,}000,\,10{,}000\}$ and each fitted noise specification, we make use the quasi-Bayes estimate $f_{\widetilde G_n}^{(X)}$ to define the
conditional density
\begin{displaymath}
\widehat p_n(x\mid Y_i)=\frac{f_{\widetilde G_n}^{(X)}(x)\,\widehat f_Z(Y_i-x)}{\displaystyle\int_{\mathbb R}f_{\widetilde G_n}^{(X)}(u)\,\widehat f_Z(Y_i-u)\,du},\qquad i=1,\ldots,n,
\end{displaymath}
where $\widehat f_Z$ denotes the density of the noise distribution estimated from the control samples $\widetilde Z_1,\ldots,\widetilde Z_m$. We then generate one reconstructed value $\widehat X_i^{(n)}$ from $\widehat p_n(\cdot\mid Y_i)$ for each observation and display the normalized histogram of
$\widehat X_1^{(n)},\ldots,\widehat X_n^{(n)}$. These reconstructed values should not be interpreted as observations of the $X_i$'s. Their distribution depends on both the quasi-Bayes estimate and the selected noise specification, and the resulting histogram is therefore intended only as a descriptive visualization, rather than as an independent
benchmark for comparing the estimates. The corresponding plots are reported in Figures~\ref{fig:supp-fluo-gauss_complete}--\ref{fig:supp-fluo-4gauss_complete}.

\begin{figure}[t]
 \centering
 \includegraphics[
   width=\textwidth,
   trim=20 15 20 0,
   clip
 ]{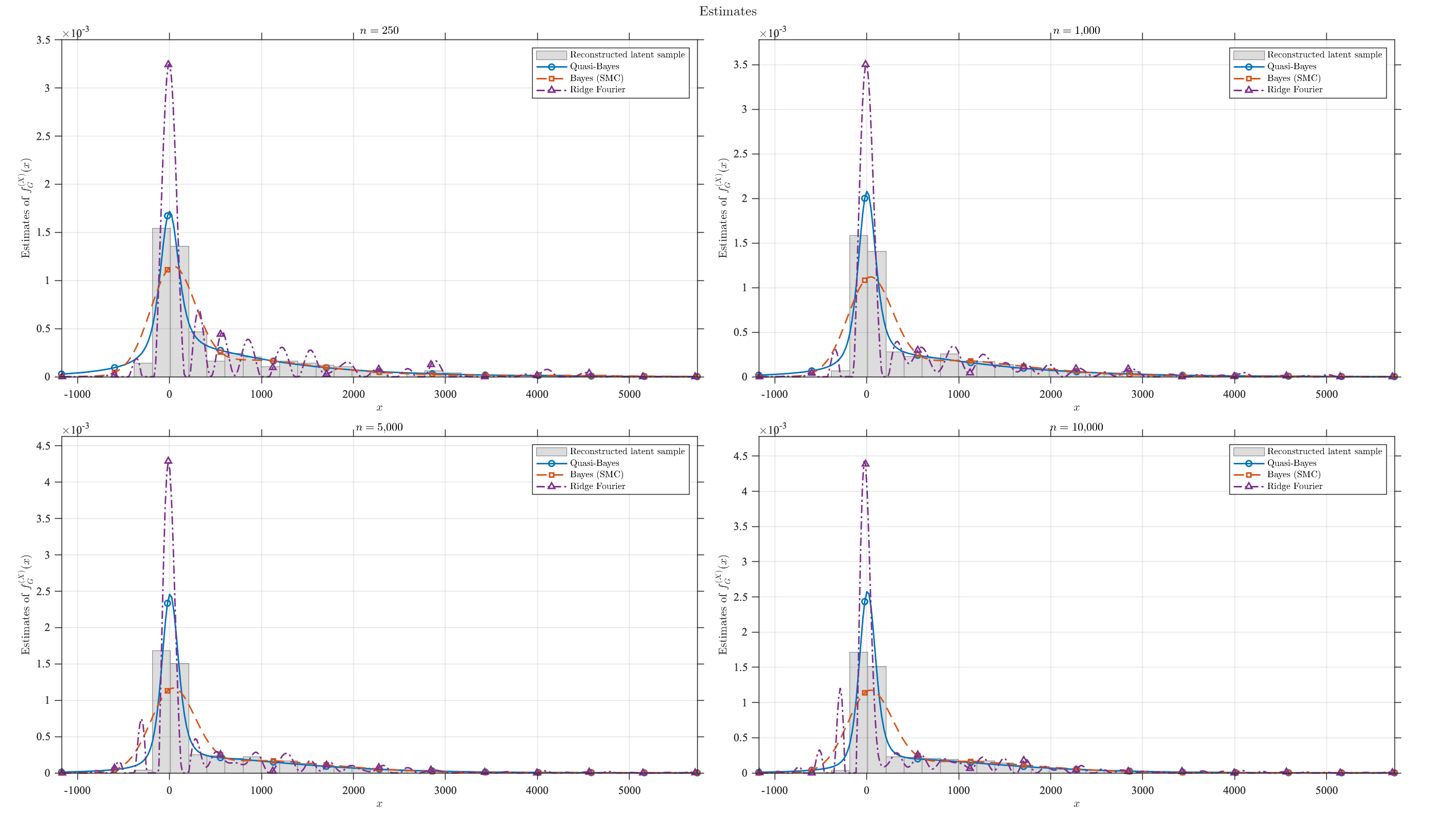}
  \caption{\footnotesize{Flow-cytometry data, with Gaussian noise: comparison between quasi-Bayesian, Bayesian and ridge-regularized Fourier deconvolution approaches.}} \label{fig:supp-fluo-gauss_complete}
\end{figure}

\begin{figure}[t]
 \centering
 \includegraphics[
   width=\textwidth,
   trim=20 15 20 0,
   clip
 ]{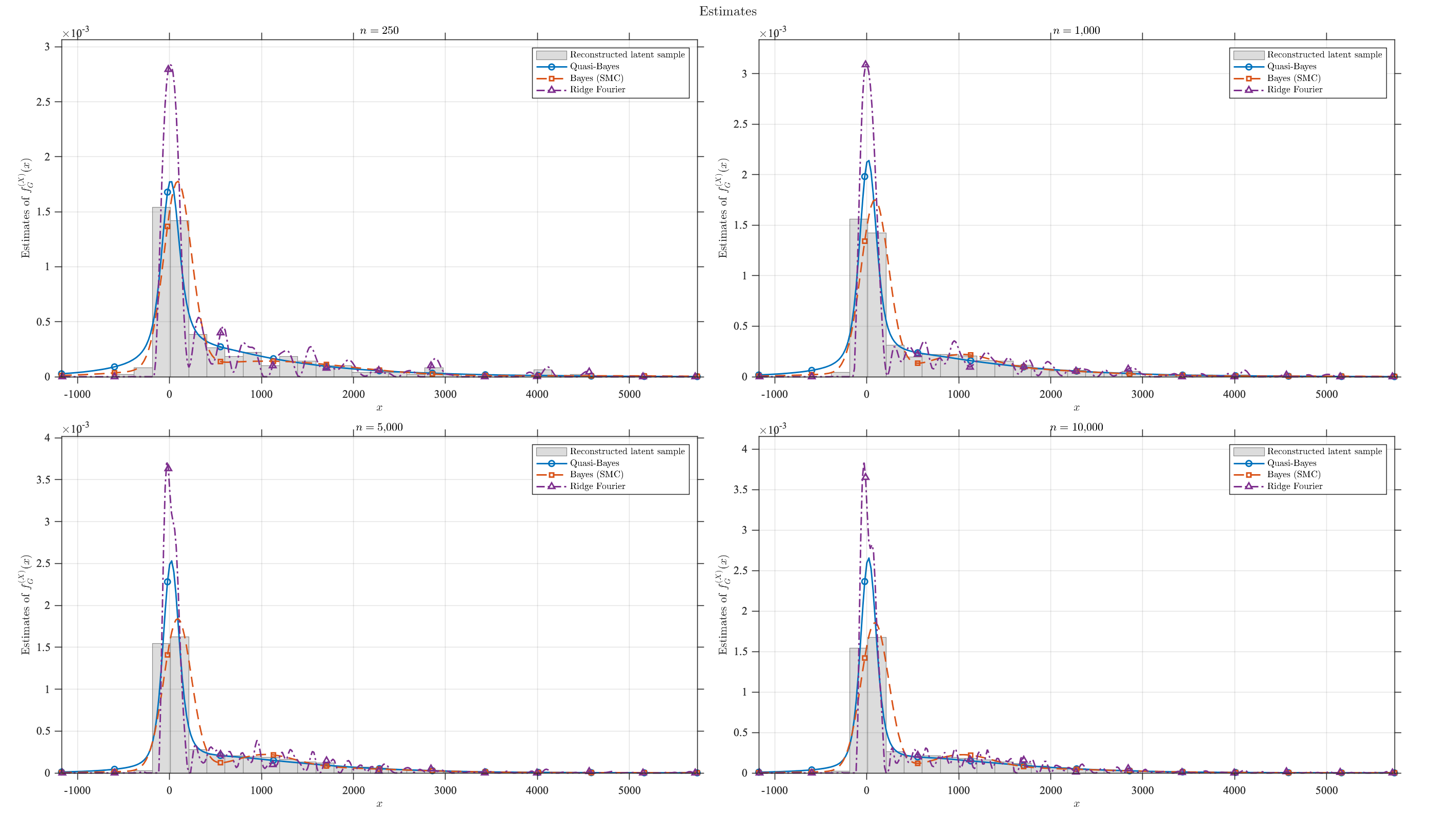}
  \caption{\footnotesize{Flow-cytometry data, with Laplace noise: comparison between quasi-Bayesian, Bayesian and ridge-regularized Fourier deconvolution approaches.}} \label{fig:supp-fluo-laplace_complete}
\end{figure}

\begin{figure}[t]
 \centering
 \includegraphics[
   width=\textwidth,
   trim=20 15 20 0,
   clip
 ]{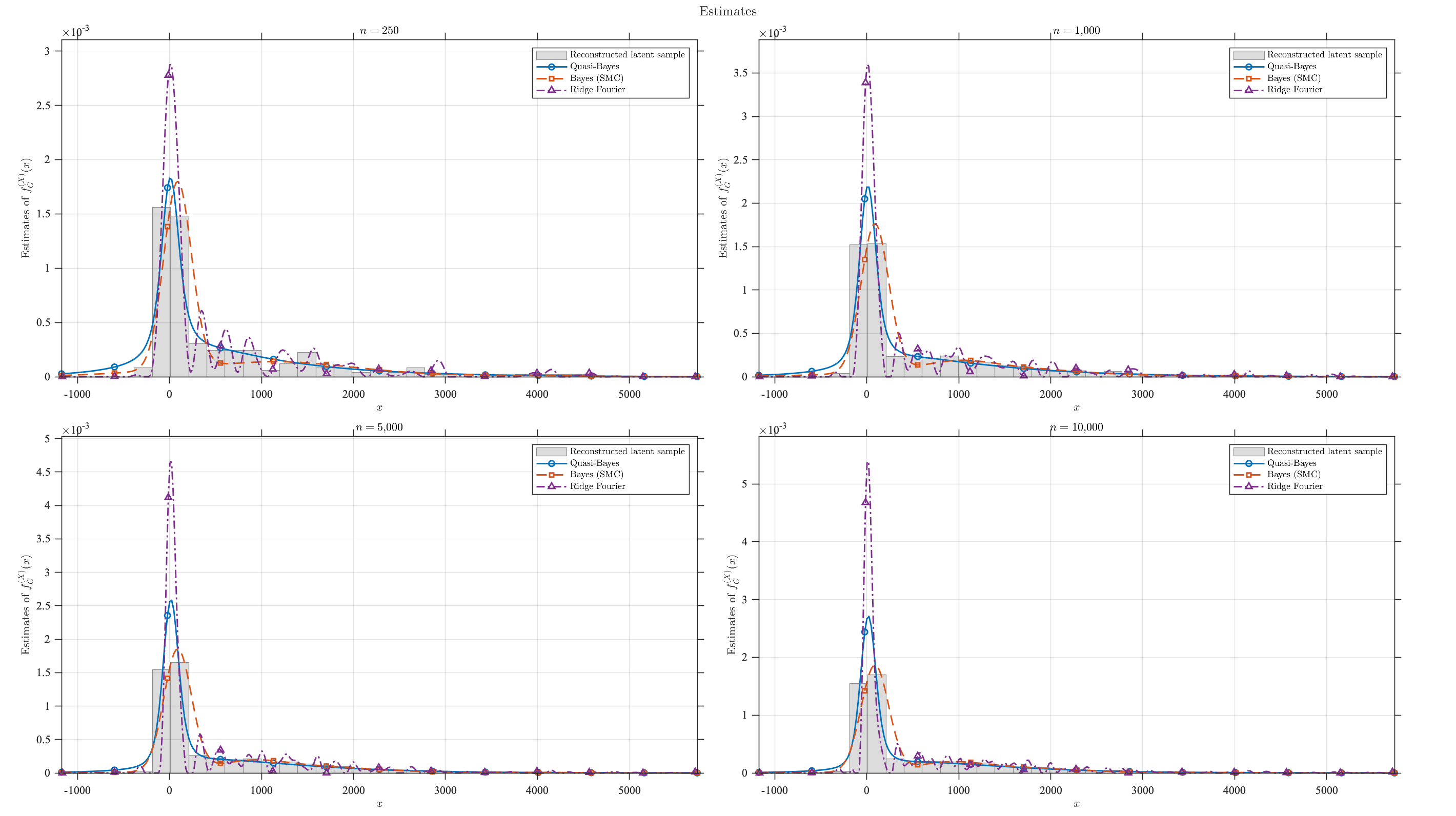}
  \caption{\footnotesize{Flow-cytometry data, with four-component Gaussian-mixture noise: comparison between quasi-Bayesian, Bayesian and ridge-regularized Fourier deconvolution approaches.}} \label{fig:supp-fluo-4gauss_complete}
\end{figure}

\subsubsection{Stellar metallicity data}\label{sec:supp-stellar}

We consider a dataset of $16{,}377$ stellar metallicity measurements obtained from the Geneva--Copenhagen Survey, which provides kinematic properties, ages, and metallicities for $F$- and $G$-type dwarf stars in the solar neighborhood; see \citet{Nor(04)} for details. Data from the same survey were previously analyzed by \citet{Bis(07)} using kernel-based deconvolution methods \citep{Ste(90),Hal(07)}. Metallicity, expressed as $[\mathrm{Fe}/\mathrm{H}]$, is the logarithmic abundance of iron relative to hydrogen and provides a measure of the fraction of heavy elements in a star, which is important for understanding stellar-formation mechanisms. Stellar metallicity cannot be measured directly, but is inferred from a star's brightness in selected spectral bands through suitable calibration procedures. The resulting calibration error is assumed to be largely stochastic and can be assessed from the dispersion between metallicity estimates obtained for the same star using different calibrations. For these data, a typical value of this dispersion is approximately $0.1$ \citep{Nor(04)}.

More precisely, let $Y_i$ denote the observed metallicity of the $i$-th star and assume the additive model $Y_i=X_i+Z_i$, for $i=1,\ldots,n$ where $X_i$ is the physical metallicity
$[\mathrm{Fe}/\mathrm{H}]$ and $Z_i$ is the corresponding calibration error. We assume that the random variables $X_i$'s are i.i.d. according to an unknown density $f_X$, that the noise variables $Z_i$'s are i.i.d. according to a density $f_Z$, and that the random variables $X_i$'s and $Z_i$'s are mutually independent. The inferential target is the density $f_X$, which is represented as the Gaussian location--scale mixture
\begin{displaymath}
f_G^{(X)}(x)=\int_{\mathbb R\times\mathbb R^{+}}\phi(x\mid\mu,\sigma^2)\,G(d\mu,d\sigma^2),\qquad x\in\mathbb R,
\end{displaymath}
where $G$ is the unknown mixing distribution. Following the analysis of \citet{Bis(07)}, for the noise distribution with density $f_{Z}$ we consider a centered (zero-mean) Laplace distribution with standard deviation $\sigma_l\in\{0.08,0.10,0.12\}$. As a sensitivity specification, we also consider a centered (zero-mean) Gaussian distribution with standard deviation $\sigma_g=0.10$. 

To implement the quasi-Bayesian approach for estimating $f_G^{(X)}$, we apply Newton's algorithm~\eqref{eq:newton} with Gaussian kernel $k(\cdot\mid\theta)=\phi(\cdot\mid\mu,\sigma^2)$, where $\theta=(\mu,\sigma^2)\in\mathbb R\times\mathbb R^+$. The parameter space $\mathbb{R}\times\mathbb R^+$ is restricted to $\Theta=[-3,2]\times[0.1,0.5]$.  We use the learning rate $\widetilde\alpha_i=(1+i)^{-1}$, for $i\geq1$, and set $\widetilde G_0$ equal to the Uniform distribution on $\Theta$. All integrals with respect to the mixing distribution are evaluated using the two-dimensional trapezoidal quadrature, and $\widetilde G_i$ is numerically renormalized after each update. For numerical implementation, $\Theta$  is discretized on a regular $100\times100$ grid.

As additional benchmarks, we also consider a Bayesian approach under the Dirichlet-process Gaussian-mixture model \eqref{eq:supp-unimodal-dp-model}, with posterior inference performed sequentially using the SMC algorithm, and the ridge-regularized Fourier deconvolution approach \citep{Ste(90),Hal(07)}. Figure~\ref{fig:supp-met_complete} reports the estimates of $f_G^{(X)}$ obtained by the quasi-Bayesian, Bayesian, and ridge-regularized Fourier deconvolution approaches.

\begin{figure}[t]
 \centering
 \includegraphics[
   width=\textwidth,
   trim=20 15 20 0,
   clip
 ]{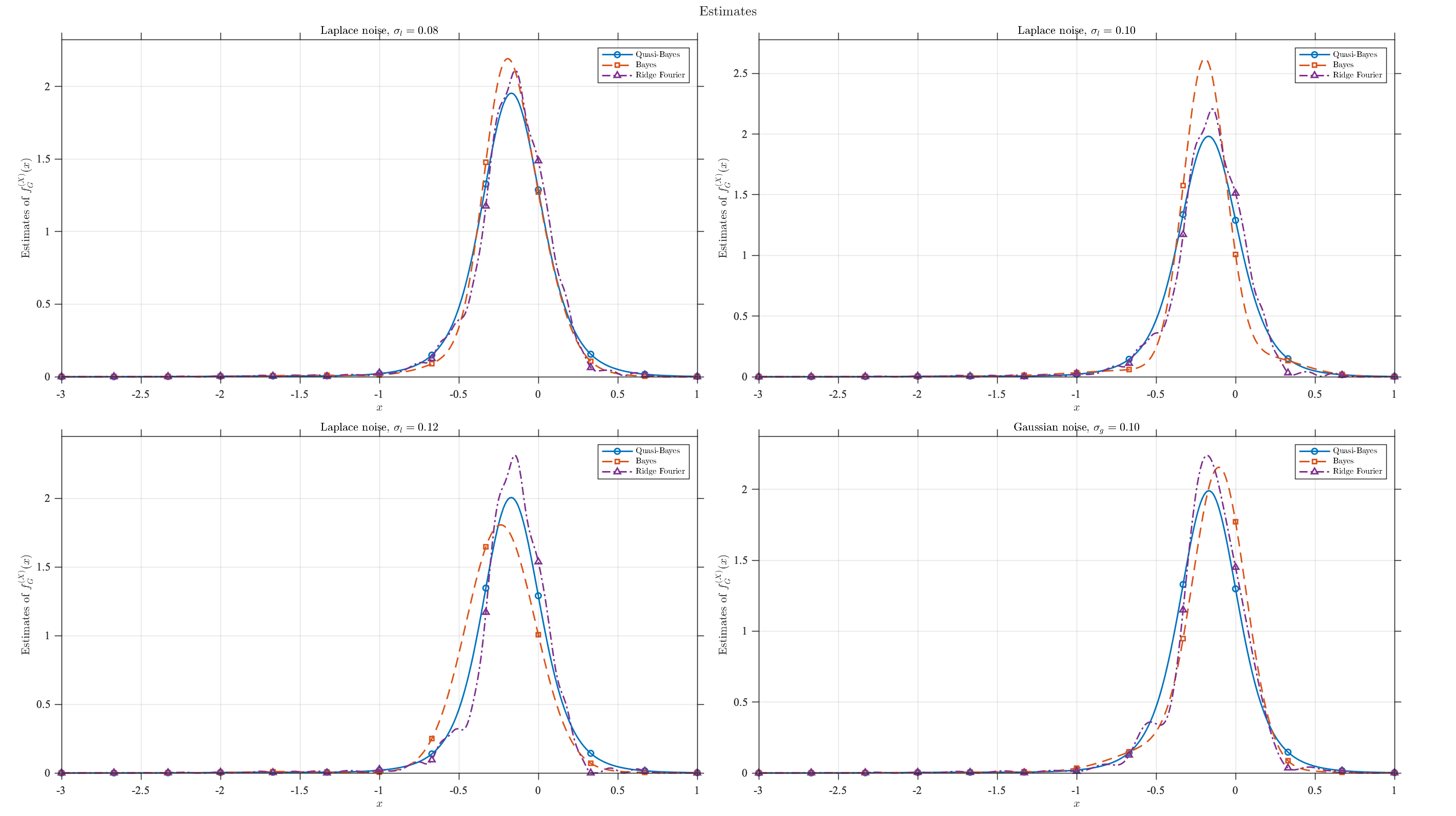}
  \caption{\footnotesize{Stellar metallicity data, with Laplace and Gaussian noise: comparison between quasi-Bayesian, Bayesian and ridge-regularized Fourier deconvolution approaches.}} \label{fig:supp-met_complete}
\end{figure}

\subsubsection{Active power-output data}\label{sec:supp-power}

We consider a semi-synthetic dataset of active power outputs, expressed in megawatts, from a 68-bus power network with 16 generators. Using GridSTAGE, a platform for emulating power-system data, we obtain a sequence of $4{,}000$ noise-free outputs at a sampling rate of $50$ observations per second; see \citet{Nan(20)} for details. The observations are retained in their original acquisition order, which, however, does not amount to assuming a time-series dependence structure; it only reflects the order in which the observations become available to the sequential procedure. Since the outputs generated by GridSTAGE are noise-free, we construct a semi-synthetic deconvolution experiment by adding independent random perturbations. This allows us to evaluate the procedures under controlled noise distributions and levels, while retaining the original GridSTAGE outputs as a benchmark for the latent signal.

More precisely, let $X_i$ denote the noise-free active power output at the $i$-th acquisition time and define the perturbed observation $Y_i=X_i+Z_i$, for $i=1,\ldots,n$, where $Z_i$ is the artificially added measurement error. We assume that the random variables $X_{i}$'s are i.i.d. according to an unknown density $f_X$, that the noise random variables $Z_{i}$'s are i.i.d. according to a density $f_Z$, and that the random variables $X_{i}$'s and $Z_{i}$'s are mutually independent. Then, the inferential target is the density $f_X$, which is represented as the Gaussian location-scale mixture,
\begin{displaymath}
f_G^{(X)}(x)=\int_{\mathbb R\times\mathbb R^{+}}\phi(x\mid\mu,\sigma^2)\,G(d\mu,d\sigma^2),\qquad x\in\mathbb R,
\end{displaymath}
where $G$ is the unknown mixing distribution. For the noise distribution with density $f_{Z}$, which is fully controlled in the experiment, we consider either a centered (zero-mean) Laplace distribution with standard deviation $\sigma_l\in\{0.5,1,2,4\}$ or a centered (zero-mean) Gaussian distribution with standard deviation $\sigma_g\in\{0.5,1,2,4\}$. Since the experiment is semi-synthetic, the values $X_1,\ldots,X_n$ are available and can be used as a benchmark to assess the estimated $f_{G}^{(X)}$.

To implement the quasi-Bayesian approach for estimating $f_G^{(X)}$, we apply Newton's algorithm~\eqref{eq:newton} with Gaussian kernel $k(\cdot\mid\theta)=\phi(\cdot\mid\mu,\sigma^2)$, where $\theta=(\mu,\sigma^2)\in\mathbb R\times\mathbb R^+$. The parameter space $\mathbb{R}\times\mathbb R^+$ is restricted to $\Theta=[50,150]\times[1,5]$.  We use the learning rate $\widetilde\alpha_i=(1+i)^{-1}$, for $i\geq1$, and set $\widetilde G_0$ equal to the Uniform distribution on $\Theta$. All integrals with respect to the mixing distribution are evaluated using the two-dimensional trapezoidal quadrature, and $\widetilde G_i$ is numerically renormalized after each update. For numerical implementation $\Theta$ is discretized on a regular $100\times100$ grid.

As additional benchmarks, we also consider a Bayesian approach under the Dirichlet-process Gaussian-mixture model \eqref{eq:supp-unimodal-dp-model}, with posterior inference performed sequentially using the SMC algorithm, and the ridge-regularized Fourier deconvolution approach \citep{Ste(90),Hal(07)}. Figures~\ref{fig:supp_power_laplace_complete}--\ref{fig:supp_power_gauss_complete} report the estimates of $f_G^{(X)}$ obtained by the quasi-Bayesian,
Bayesian, and ridge-regularized Fourier deconvolution approaches, together with the (benchmark) histogram of the noise-free power outputs $X_1,\ldots,X_n$.

\begin{figure}[t]
 \centering
 \includegraphics[
   width=\textwidth,
   trim=20 15 20 0,
   clip
 ]{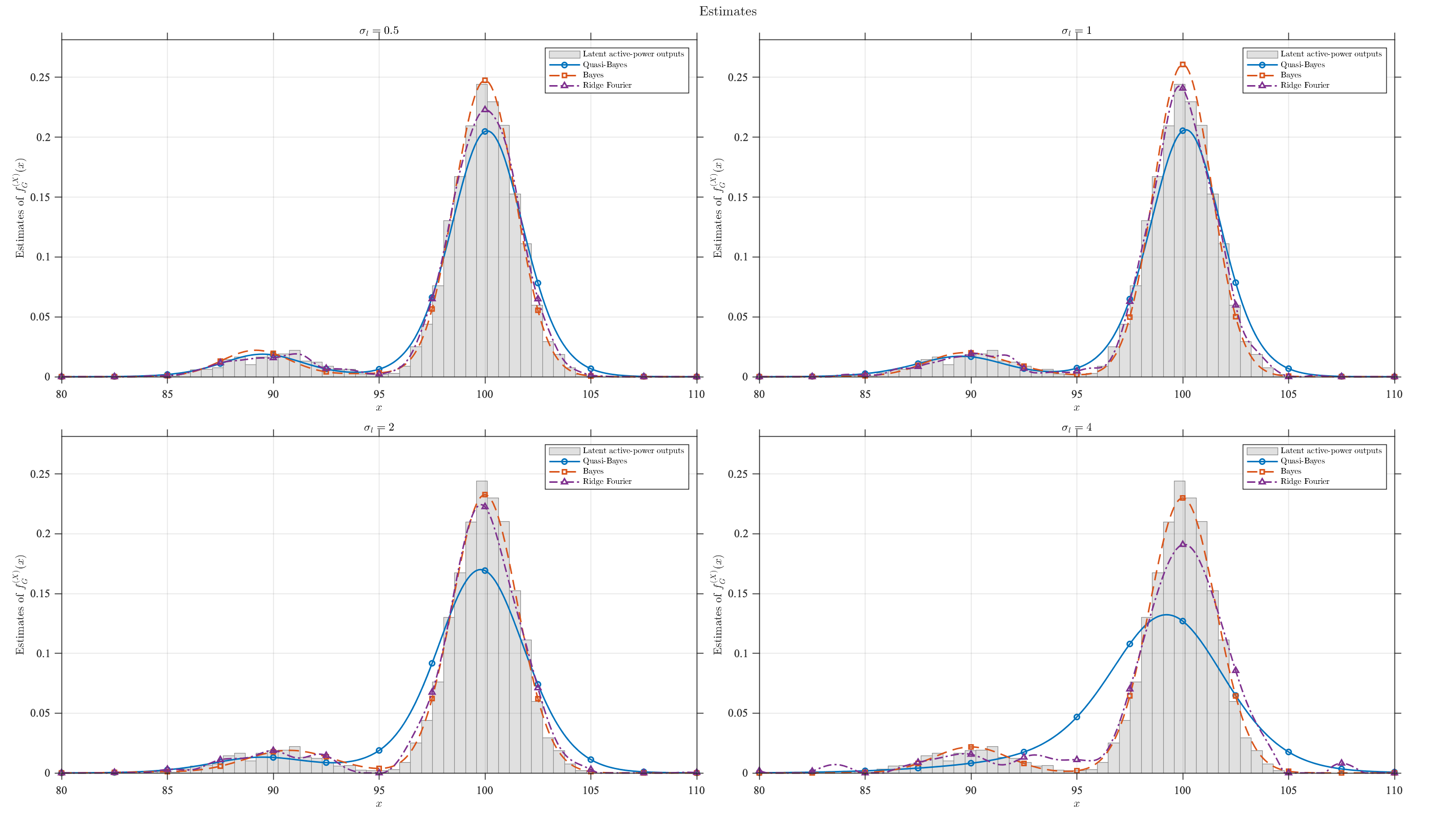}
  \caption{\footnotesize{Active power-output data, with Laplace noise: comparison between quasi-Bayesian, Bayesian and ridge-regularized Fourier deconvolution approaches.}} \label{fig:supp_power_laplace_complete}
\end{figure}

\begin{figure}[t]
 \centering
 \includegraphics[
   width=\textwidth,
   trim=20 15 20 0,
   clip
 ]{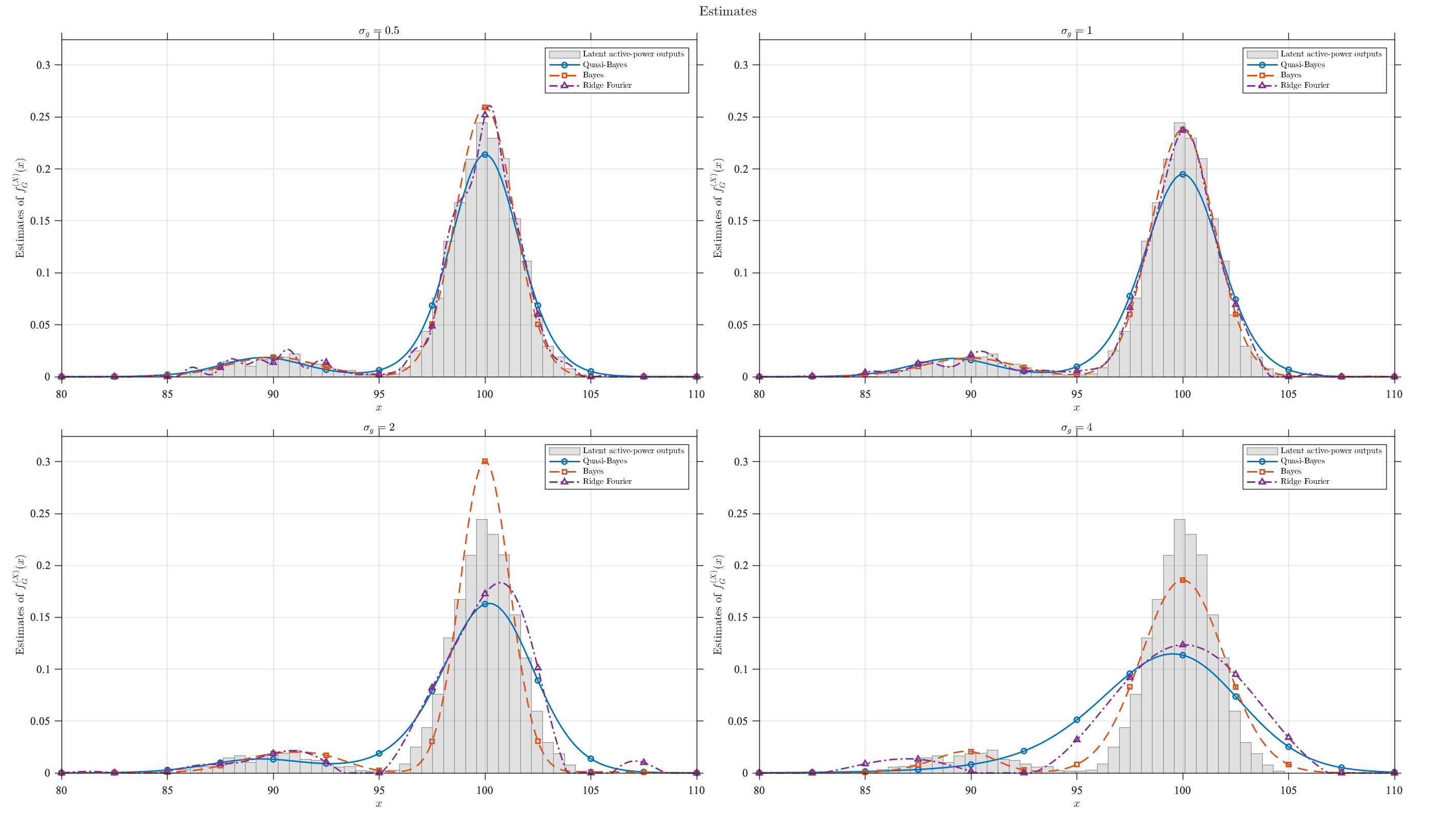}
  \caption{\footnotesize{Active power-output data, with Gaussian noise: comparison between quasi-Bayesian, Bayesian and ridge-regularized Fourier deconvolution approaches.}} \label{fig:supp_power_gauss_complete}
\end{figure}

\section{Auxiliary results}
\subsection{Learning-rate calibration}\label{sec:lrates}
In this section we propose a data-driven calibration of the learning rates $\widetilde\alpha_n$ in the quasi-Bayesian learning model \eqref{eq:Y}--\eqref{eq:newton}. The proposed calibration is based on 
comparing the evolution of the recursive estimate based on the noisy observations $(Y_n)$ with the evolution that would be obtained if Newton's algorithm were applied directly to the latent variables $(X_n)$, that is, in the absence of measurement noise. The frequentist merging result in Corollary \ref{teo_merg} provides the theoretical justification for calibrating $\widetilde\alpha_n$ by matching the updates of the two recursive procedures.

For Newton's algorithm based on the $X_{n}$'s, we assume the standard learning rate $\alpha_{n}=(\alpha+n)^{-1}$, for $\alpha>0$. Denote by $g_n(\theta)$ and $\widetilde g_n(\theta)$ the densities of $G_n$ and $\widetilde G_n$ with respect to the dominating measure $\mu$. The updates or innovations of $(g_{n})_{n\geq0}$ are
\begin{equation}\label{up1}
 I_i(\theta)=g_i(\theta)-g_{i-1}(\theta)=\frac{1}{\alpha+i}(g_{i-1}(\theta\mid X_i)- g_{i-1}(\theta))\qquad i\geq1,
\end{equation}
where $g_{i-1}(\theta\mid X_i)$ is the density with respect to $\mu$ and at the point $\theta$ of $G_{i-1}(\cdot\mid X_i)$.
Instead, for Newton's algorithm based on the $Y_{n}$'s, we assume that $\widetilde{\alpha}_{n}$ is in a class of standard learning rates defined as $\{\left(\alpha+n\right)^{-\gamma}: 2^{-1} <\gamma\leq 1\}$. Then, the updates or innovations of $(\widetilde{g}_{n})_{n\geq0}$ are of the form
\begin{equation}\label{up2}
\widetilde I_i(\theta,\gamma)=\widetilde{g}_i(\theta)-\widetilde{g}_{i-1}(\theta)=(\alpha+i)^{-\gamma}(\widetilde g_{i-1,\gamma}(\theta\mid Y_i)-\widetilde g_{i-1,\gamma}(\theta))\qquad i\geq1.
\end{equation}
Other choices for $\alpha_{n}$ and $\widetilde{\alpha}_{n}$ are possible; we refer to \citet{For(20)} and references therein for details.
 
Based on \eqref{up1}-\eqref{up2}, we calibrate the parameter $\gamma$ by matching the sizes of the updates. For $i\geq1$, we set
\begin{displaymath}
\Delta_i(\theta,\gamma)=\log|g_{i-1}(\theta\mid X_i)-g_{i-1}(\theta)|-\log|\widetilde g_{i-1,\gamma}(\theta\mid Y_i)-\widetilde g_{i-1,\gamma}(\theta)|,
\end{displaymath}
such that $\log|\widetilde I_i(\theta,\gamma)|-\log|I_i(\theta)|=\Delta_i(\theta,\gamma)-(1-\gamma)\log\left(\alpha+i\right)$. In particular, to achieve the matching we need to find the value of the parameter $\gamma\in(\frac 1 2,1]$ that minimizes, for some choice of $M\in\mathbb N$,
\begin{displaymath}
\sum_{i=1}^M \int_{\Theta}  E\left(\left(\Delta_i(\theta,\gamma)-(1-\gamma)\log\left(\alpha+i\right)\right)^2\,\middle|\,  \mathcal G_{i-1}\right)\widetilde g_{i-1,\gamma}(\theta)\mu(\ddr\theta).
\end{displaymath}
We solve the problem through Monte Carlo simulation, for $\gamma$ in a grid, by means of random sampling $(\theta_0^{(\gamma)},\dots,\theta_M^{(\gamma)})$. See Algorithm \ref{alg1} for the details on the algorithm for the calibration of the learning rate.

\begin{algorithm}[H]
  \caption{Calibrating the learning rate $\widetilde{\alpha}_{n}=(\alpha+n)^{-\gamma}$.}\label{alg1}
\SetAlgoLined
\KwIn{$N$, $L$; $M$; $B=\{b_1,\dots,b_L\}$; $g_0$; $k$; $\widetilde{k}$;}

\For{$l=1$ to $L$}{

Set $\widetilde{g}_0=g_0$

\For{$i = 1$ to $M$}{

Sample $\theta_i\sim\widetilde g_{i-1}(\cdot)$; $X_i\sim k(\cdot\mid\theta_i)$; $Z_i\sim f_Z(z)$

Compute $Y_i=X_i+Z_i$

Compute $\widetilde g_{i-1}(\ddr\theta\mid Y_i)=\frac{\widetilde k(Y_i\mid\theta)\widetilde g_{i-1}(\ddr\theta)}{\int_{\Theta} \widetilde k(Y_i\mid \theta)\widetilde g_{i-1}(\ddr\theta)}$

Compute $g_{i-1}(\ddr\theta\mid X_i)=\frac{k(X_i\mid \theta)g_{i-1}(\ddr\theta)}{\int_{\Theta} k(X_i\mid \theta)g_{i-1}(\ddr\theta)}$

Compute $\widetilde g_i(\ddr\theta)=\widetilde g_{i-1}(\ddr\theta)+\frac{\left(\widetilde g_{i-1}(\ddr\theta\mid Y_i)-\widetilde g_{i-1}(\ddr\theta)\right)}{\left(1+\frac{i}{\alpha}\right)^{\gamma_l} }$

Compute $g_i(\theta)=g_{i-1}(\theta)+\frac{\left(g_{i-1}(\theta\mid X_i)- g_{i-1}(\theta)\right)}{1+\frac{i}{\alpha}}$

Return $\Delta_i(\theta,\gamma_l)=\log\frac{|g_{i-1}(\theta\mid X_i)-g_{i-1}(\theta)|}{|\widetilde g_{i-1}(\theta\mid Y_i)-\widetilde g_{i-1}(\theta)|}$}}

\KwOut{$\hat \gamma=\mbox{argmin}_{l=1,\dots,L}\sum_{i=1}^M\left(\Delta_i(\theta,\gamma_l)-(1-\gamma_l)\log\left(\alpha+i\right)\right)^2$}
\end{algorithm}

\subsection{Rate of convergence for finite mixtures}\label{app:frequentist}

In Section \ref{sec5}  we proved that, under regularity assumptions,  the quasi-Bayesian sequential learning estimate of $f_{G^*}^{(X)}$ is consistent.  In this section, we derive an upper bound on the rate of convergence in case of finite mixtures with known support points.

Let $\Theta$ be a known finite set and let $\mu$ denote the counting measure on it. Assume that the  $X_n$ are i.i.d. with density function 
\begin{equation}\label{eq:fg*discr}
f_{g^*}^{(X)}(x)=\int k(x\mid\theta)g^*(\theta)\mu(d\theta),
\end{equation}
with respect to a dominating measure $\lambda$. The quasi-Bayesian learning process for the probability distribution of the $X_n$, based on the observation of the $Y_n$ is given by
\begin{equation}\label{eq:appfgn}
    f_{\widetilde g_n}^{(X)}(x)=\int k(x\mid\theta)\widetilde g_n(\theta)\mu(\ddr\theta)
\end{equation}
where
\begin{equation}
\label{eq:appgn}
    \widetilde g_{n+1}(\theta)=(1-\widetilde \alpha_{n+1} )\widetilde g_n(\theta)+\frac{(f_Z*k)(Y_{n+1}\mid\theta)}{\int_\Theta (f_Z*k)(Y_{n+1}\mid\theta)\widetilde g_n(\theta)\mu(\ddr\theta)}\widetilde g_n(\theta),
\end{equation}
given an initial guess $\widetilde{g}_{0}$ and a sequence $(\widetilde{\alpha}_{n})_{n\geq1}$ of real numbers in $(0,1)$ with $\sum_{n\geq1}\widetilde{\alpha}_{n}=+\infty$ and $\sum_{n\geq1}\widetilde{\alpha}_{n}^{2}<+\infty$.
To underline that the $X_n$'s are i.i.d., we denote the probability measure on $(\Omega,\mathcal F)$ by $P^*$. 

\begin{thm}\label{th:frequentist}
Let $(X_n)_{n\geq1}$, $(Y_n)_{n\geq1}$ and $(Z_n)_{n\geq1}$ be sequences of random variables defined on $(\Omega,\mathcal F,P^*)$, such that the $X_n$'s are i.i.d. with density $f_{g^*}^{(X)}$ in \eqref{eq:fg*discr}, the $Z_n$'s are i.i.d. with density $f_Z$ and independent of the $X_n$'s, and $Y_n=X_n+Z_n$ for $n\geq 1$. Let $f_{\widetilde g_n}^{(X)}$ be in \eqref{eq:appfgn}, where $\widetilde g_n$ satisfies \eqref{eq:appgn} with $\widetilde\alpha_n=(\alpha+n)^{-\gamma}$ for $\alpha>0$ and $\gamma\in (1/2,1)$ and $g_0(\theta)>0$ for every $\theta\in\Theta$. 
Under the assumptions A1)--A5) and F3)
it holds $P^*$-a.s. that, for $\delta<1-1/(2\gamma)$,
\begin{displaymath}
|f_{\widetilde g_n}^{(X)}(x)-f_{g^{\ast}}^{(X)}(x)|=o(n^{-\delta})\qquad x\in\mathbb R,
\end{displaymath}
and
\begin{displaymath}
\int_{\mathbb R}|f_{\widetilde g_n}^{(X)}(x)-f_{g^*}^{(X)}(x)|\lambda(dx)=o(n^{-\delta}).
\end{displaymath}
\end{thm}

Note that Theorem \ref{th:frequentist} specifies a learning rate of the form $\widetilde{\alpha}_n = (\alpha + n)^{-\gamma}$ with $\gamma \in (1/2, 1)$. This is a standard choice for the learning rate, and the exclusion of the value $\gamma = 1$ is not particularly restrictive; indeed calibrating $\widetilde \alpha_n$ as  in Section \ref{sec:lrates} typically results in $\gamma < 1$.

Arguing as in the proof of Lemma~\ref{lemma:rate2}, one can show that
\[
W_1\!\left(
F_{\widetilde G_n}^{(X)},
F_{G^\star}^{(X)}
\right)
=
o_{\mathbb P^\star}\!\left(
n^{-\delta}
\right),
\]
for every $\delta<1-\frac{1}{2\gamma}$.
Under the assumptions of Proposition \ref{th:rateNewton}, with
\(\gamma\in(2/3,1)\), one also has
\[
W_1\!\left(
F_{\widetilde G_n}^{(X)},
F_{G^\star}^{(X)}
\right)
=
o_{\mathbb P^\star}\!\left(
n^{-\frac{(1-\gamma)(\alpha+1)}{2(\alpha+2)}}
\right).
\]
Since
\[
1-\frac{1}{2\gamma}
>
\frac{(1-\gamma)(\alpha+1)}{2(\alpha+2)},
\qquad
\gamma\in(2/3,1),
\]
the convergence rate established in Theorem~\ref{th:frequentist} for finite mixtures
is strictly faster than that of Proposition \ref{th:rateNewton} for general mixtures under Laplace error distribution.


\subsubsection{Proof of Theorem \ref{th:frequentist}}
The proof of Theorem \ref{th:frequentist} is based on the following Lemma \ref{lem:stochapp1} and on the rate of convergence of $\widetilde g_n$ to $g^*$. This last result is known in the literature \citep{Mar(12)}, but we argue that the proof is not complete. We provide for completeness the proof in Lemma \ref{lem:stochapp3}. 

Let $d$ be the cardinality of $\Theta$.
We can identify each density function with respect to the counting measure on $\Theta$ with a vector in $\mathbb R^d$. More precisely, let $\Theta=\{\vartheta_1,\dots,\vartheta_d\}$. Then $g=[g(\vartheta_1),\dots,g(\vartheta_{d})]^T$.
We also write $\widetilde k_y$ for the $d$-dimensional vector  $\widetilde k(y\mid\theta)$ with $\theta\in \Theta$, and define, for each density function $g$ with respect to $\mu$, $f_g^{(Y)}(y)=\int_\Theta \widetilde k(y\mid\theta)g(\theta)\mu(d\theta)$.

\begin{lem}\label{lem:stochapp1}
  Under the assumptions of Theorem \ref{th:frequentist}, there exists, for $\lambda$-almost every $x\in\mathbb R$ a constant $C(x)$ such that for every $n\geq 1$
  $$
  |f_{\widetilde g_n}^{(X)}(x)-f_{g^*}^{(X)}(x)|\leq C(x) ||\widetilde g_n-g^*||_2,
  $$
 where $\widetilde g_n$ and $g^*$ are interpreted as vectors in $\mathbb R^d$ and $||\cdot||_2$ is the Euclidean norm. Furthermore, there exists a constant $C$ such that
  $$
 \int |f_{\widetilde g_n}^{(X)}(x)-f_{g^*}^{(X)}(x)|\lambda(dx)\leq C ||\widetilde g_n-g^*||_2,
  $$
  \end{lem}
\begin{proof}
By the assumption A3), there exists for $\lambda$-almost every $x$, a constant $C(x)$ such that $\sup_\theta k(x\mid\theta)\leq C(x)$. Thus
\begin{align*}
     |f_{\widetilde g_n}^{(X)}(x)-f_{g^*}^{(X)}(x)|&\leq  \int_\Theta k(x\mid\theta)|\widetilde g_n(\theta)-g^*(\theta)|\mu(d\theta)\\
     &
       \leq C(x)\int_{\Theta}|\widetilde g_n(\theta)-g^*(\theta)|\mu(d\theta)\\&
     \leq C(x) ||\widetilde g_n-g^*||_2.
\end{align*}
 Moreover, again by assumption A3), the exists a constant $C$ such that $\int \sup_{\theta\in\Theta}k(x\mid\theta)\lambda(\ddr x)\leq C$. Hence,
    \begin{align*}
       \int |f_{\widetilde g_n}^{(X)}(x)-f_{g^*}^{(X)}(x)|\lambda(dx)&\leq
       \iint_\Theta k(x\mid\theta)|\widetilde g_n(\theta)-g^*(\theta)|\mu(d\theta)\lambda(dx)\\&
       \leq \int \sup_{\theta\in\Theta}k(x\mid\theta)\lambda(\ddr x) \int_{\Theta}|\widetilde g_n(\theta)-g^*(\theta)|\mu(d\theta)\\&\leq C ||\widetilde g_n-g^*||_2
         \end{align*}
         which concludes the proof.
         \end{proof}



\begin{lem}\label{lem:stochapp3}
Let $g_n$ and $g^*$ be the representations of $g_n$ ($n\geq 0$) and $g^*$ in $\mathbb R^d$.     Under the assumptions of Theorem \ref{th:frequentist}, for every $\delta<1-1/(2\gamma)$
    $$
    ||g_n-g^*||_2=o(n^{-\delta})\quad P^*-a.s.
    $$
\end{lem}
\begin{proof}
Since $\Theta$ is finite, then the assumptions F1)--F2) hold. Therefore we can apply \textcolor{blue}{Corollary} \ref{teo_cons} and deduce that the random vector $\widetilde g_n$ converges to $g^*$, $P^*$-a.s.
Moreover, $(\widetilde g_n)$ satisfies the 
 stochastic approximation in $\mathbb R^d$:
\begin{align*}
\widetilde g_{n+1}&=\widetilde g_{n}+\widetilde \alpha_{n+1} \widetilde g_{n}\circ \left(\frac{\widetilde k_{Y_{n+1}}}{f_{\widetilde g_{n}}^{(Y)}(Y_{n+1})}-\bf 1\right)\\
&=\widetilde g_{n}+\widetilde \alpha_{n+1} h(\widetilde g_{n})+\widetilde \alpha_{n+1}\epsilon_{n+1},
\end{align*}
with initial value $\widetilde g_0$, where
$$
h(g)=g\circ \left(
\int \frac{\widetilde k_y}{f_g^{(Y)}(y)}
f_{g^*}^{(Y)}(y)\lambda(dy)-\bf 1
\right),
$$
and
$$
\epsilon_{n+1}=g\circ \left(
\frac{\widetilde k_{Y_{n+1}}}{f_g^{(Y)}(Y_{n+1})}
f_{g^*}^{(Y)}(Y_{n+1})-
\int \frac{\widetilde k_y}{f_g^{(Y)}(y)}
f_{g^*}^{(Y)}(y)\lambda(dy)
\right).
$$
    The thesis follows from Theorem 3.1.1 in \cite{Chen(05)} if we can prove that the following conditions hold:\begin{itemize}
        \item[F3.1] $\widetilde \alpha_n\rightarrow 0$, $\sum_{n=1}^\infty \widetilde \alpha_n=\infty$, $\widetilde \alpha_{n+1}^{-1}-\widetilde \alpha_n^{-1}\rightarrow 0$;
        \item[F3.2] $\sum_{n=1}^\infty \widetilde \alpha_n^{1-\delta/2}\epsilon_n$ converges, $P^*$-a.s.;
        \item[F3.3] $h$ is measurable and locally bounded, and is differentiable at $g^*$. Let $H$ be a matrix such that, as $g\rightarrow g^*$,
        $$
        h(g)=H(g-g^*)+r(g),\quad r(g^*)= 0,\quad ||r(g)||_2=o(||g-g^*||_2).
        $$
        All the eigenvalues of the matrix $H$ have negative real parts.
    \end{itemize}
\begin{remark}
        Theorem 3.1.1 in \cite{Chen(05)} also requests  the following condition:
  {\em  F3.4    There exists a continuously differentiable function $\mathcal{I}:\Delta^d\rightarrow \mathbb R$ such that $$\sup_{\delta_1\leq {\rm d}(g,g^*)\leq \delta_2}h(g)^T\nabla_{g}\mathcal{I}<0,$$
    for every $\delta_2>\delta_1>0$.}   
   However, a close inspection of the proof of Theorem 3.1.1 in Chen  shows that F3.4 is employed only to ensure convergence of $\widetilde g_n$ towards $g^*$. In our case, the convergence has been already proved.
\end{remark}
    F3.1 is obvious. To prove F3.2, notice that $(\sum_{k=1}^n\widetilde \alpha_k^{1-\delta/2}\epsilon_k)_{n\geq 1}$ is a martingale. Moreover, since $\Theta$ is finite and 
    $$
    \int\int_\Theta g(\theta)\frac{\widetilde k(y\mid\theta)}{f_g^{(Y)}(y)}f_{g^*}^{(Y)}(y)\lambda(dy)\mu(d\theta)=1,
    $$
    then $|\epsilon_{n,i}|\leq 1$ for every $i$. Together with $\gamma(2-\delta)>1$, this implies that, for every $i=1,\dots,d$,
    \begin{align*}
        \sup_nE^*((\sum_{k=1}^n\widetilde \alpha_k^{1-\delta/2}\epsilon_{k,i})^2)
        &= \sup_n\sum_{k=1}^n\widetilde \alpha_k^{2-\delta}E^*(\epsilon_{k,i})^2\\
        &\leq  \sum_{k=1}^\infty\left(\frac{1}{\alpha+k}\right)^{\gamma(2-\delta)}<+\infty.
        \end{align*}
        Thus, $(\sum_{k=1}^n \widetilde{\alpha}_k^{1-\delta/2} \epsilon_k)_{n \geq 1}$ is a uniformly integrable martingale, and therefore it converges $P^*$-a.s.
    We now prove F3.3. The function $h$ is continuously differentiable at $g^*$ and
        $$
        H(\theta,\theta')=-g^*(\theta)\int\frac{\widetilde k_y(\theta)\widetilde k_y(\theta')}{f_{g^*}^{(Y)}(y)}\lambda(dy).
        $$
        The matrix $H$ can be seen as the product of the diagonal matrix ${\rm diag}(g^*)$, which is positive definite, times $-M$, with 
        $$M(\theta,\theta')=\int \frac{\widetilde k_y(\theta)\widetilde k_y(\theta')}{f_{g^*}^{(Y)}(y)}\lambda(dy).$$ We now prove that $M$ is positive definite. Given a vector $u=[u(\theta)]_{\theta\in\Theta}\in \mathbb R^d$ with $||u||=1$, we can write that
        \begin{align*}
        u^T Mu=&\iint_{\Theta^2}u(\theta)u(\theta')\int \frac{\widetilde k_y(\theta)\widetilde k_y(\theta')}{f_{g^*}^{(Y)}(y)}\lambda(dy)\mu(d\theta)\mu(d\theta')\\&=
        \int \frac{1}{f_{g^*}^{(Y)}(y)}\left(\int_\Theta u(\theta) \widetilde k_y(\theta)\mu(d\theta)\right)^2\lambda(dy)\geq 0,
        \end{align*}
        with equality if and only if $\int_\Theta u(\theta)\widetilde k_y(\theta)\mu(d\theta)=0$, $\lambda$-almost everywhere. The only solution of the above equation in $u$ is $u(\theta)=0$ for every $\theta\in\Theta$, which is inconsistent with $||u||=1$. Since the function $u^TMu$ is continuous on the compact set $||u||=1$, then the minimum is attained, and by the above reasoning it is strictly positive. It follows that $M$ is positive definite.
    \end{proof}

    \begin{proof}[Proof of Theorem \ref{th:frequentist}]
        The proof is an immediate consequence of Lemma \ref{lem:stochapp1} and Lemma \ref{lem:stochapp3}.
    \end{proof}


\end{document}